\documentclass[12pt,a4paper,oneside]{article}
\usepackage[rreport]{cmpcover}

\usepackage[table]{xcolor}
\usepackage{cvpr-geom}
\usepackage{cvpr-abbriv}

\usepackage[tocindentauto]{tocstyle}

\usepackage{cvpr-abbriv}
\usepackage[alwaysadjust,flushright]{paralist}
\usepackage{longtable}
\usepackage{needspace}
\usepackage{colortbl}

\usepackage[pagebackref=false,breaklinks=false,a4paper,colorlinks=false,bookmarks=false,backref=page,baseurl=cmp.felk.cvut.cz/~shekhovt, hyperfigures=true,linkbordercolor={1 1 1},pdfcenterwindow=true,pdffitwindow=true,pdfstartview=FitH,pdfstartpage=1]{hyperref}
\usepackage[ruled,vlined,linesnumbered]{algorithm2e}
\newlength{\myskip}
\makeatletter
\let\proof\@undefined                        % undefine \proof
\let\endproof\@undefined                  % undefine \endproof
\makeatother
\usepackage{amsthm}
\renewenvironment{enumerate}%
  {\begin{list}{\arabic{enumi}.}%
     {\topsep=0in\itemsep=0in\parsep=0pt\partopsep=0in\usecounter{enumi}}%
   }{\end{list}}

\renewenvironment{itemize}%
  {\begin{list}{$\bullet$}%
     {\topsep=0in\itemsep=0pt\parsep=0pt\partopsep=0in\usecounter{enumi}}%
   }{\end{list}\addvspace{0pt}}

\SetAlgoSkip{skipmyskip}
\SetAlgoInsideSkip{}
\makeatletter
\let\corollary\@undefined
\let\c@corollary\@undefined
\let\endcorollary\@undefined
\let\definition\@undefined
\let\c@definition\@undefined
\let\enddefinition\@undefined
\let\proof\@undefined
\let\endproof\@undefined
\let\theorem\@undefined
\let\c@theorem\@undefined
\let\endtheorem\@undefined
\let\lemma\@undefined
\let\c@lemma\@undefined
\let\endlemma\@undefined
\makeatother

\newtheoremstyle{tightItalic}% name
  {0.5\myskip}%      Space above
  {0.5\myskip}%      Space below
  {}%         Body font
  {}%         Indent amount (empty = no indent, \parindent = para indent)
  {\itshape}% Thm head font
  {.}%        Punctuation after thm head
  { }%     Space after thm head: " " = normal interword space;
  {}%         Thm head spec (can be left empty, meaning `normal')

\newtheoremstyle{tightBf}% name
  {0.5\myskip}%      Space above
  {0.5\myskip}%      Space below
  {}%         Body font
  {}%         Indent amount (empty = no indent, \parindent = para indent)
  {\bf}% Thm head font
  {.}%        Punctuation after thm head
  {.5em}%     Space after thm head: " " = normal interword space;
  {}%         Thm head spec (can be left empty, meaning `normal')

\theoremstyle{tightBf}
\newtheorem{theorem}{Theorem}
\newtheorem{corollary}{Corollary}
\newtheorem{definition}{Definition}
\newtheorem{example}{Example}
\newtheorem{statement}{Statement}
\newtheorem{lemma}{Lemma}

\theoremstyle{tightItalic}
\newtheorem*{proof}{Proof}

\makeatletter
\renewcommand\section{\@startsection {section}{1}{\z@}%
{2\myskip}%
{2\myskip}%
{\normalfont\Large\bfseries}}
\renewcommand\subsection{\@startsection {subsection}{2}{\z@}%
{\myskip}%
{\myskip}%
{\flushleft\normalfont\bfseries}}
\makeatother

\usepackage{alg}
\usepackage{epsfig}
\usepackage{graphicx}
\usepackage{amsmath}
\usepackage{amssymb}

\usepackage{amsfonts}
\usepackage{verbatim}
\usepackage{color}
\usepackage{epsfig}

\newcommand{\Mid}{\:\Big|\,}
\renewcommand{\mid}{\:|\,}

\newcommand{\Natural}{\mathbb{N}}
\newcommand{\bbN}{\mathbb{N}}
\newcommand{\bbZ}{\mathbb{Z}}

\newcommand{\tab}{{\hphantom{bla}}}

\newcommand{\B}{\mathcal{B}}

\newcommand{\st}{\mbox{s.t. }}

\newcounter{myRomanCounter}

\renewcommand{\iff}{{\textrm{iff} }}

\newcommand{\VH}{{V\'aclav Hlav\'a\v c}}
\newcommand{\AS}{Alexander Shekhovtsov}

\newcommand{\red}{\color[rgb]{1,0,0}}

\newcommand*{\mypar}[1]{\par{\bf #1.}}

\def\leftbb{\mathopen{\rlap{$[$}\hskip1.3pt[}}
\def\rightbb{\mathclose{\rlap{$]$}\hskip1.3pt]}}

\newcommand{\IF}{\mbox{ \rm if }}

\newcommand{\AND}{\mbox{ \rm and }}
\newcommand{\Algorithm}[1]{Alg.\,\ref{#1}}

\newcommand{\Figure}[1]{Fig.\,\ref{#1}}

\title{A Distributed Mincut/Maxflow Algorithm Combining Path Augmentation and Push-Relabel\vskip10mm}
\author{\AS \and \VH}
\CMPReportNo{K333--43/11, CTU--CMP--2011--03}
\CMPAcknowledgement{This work was supported bu EU projects FP7-ICT-247870
NIFTi and FP7-ICT-247525 HUMAVIPS and the Czech project 1M0567 CAK}
\CMPEmail{shekhole@fel.cvut.cz \and hlavac@fel.cvut.cz}
\defRCSRevision{$Revision: 1.0 $}
\CMPSubtitle{(Version \RCSRevision)}
\CMPAffiliation{
Czech Technical University in Prague
}
\newcommand{\fig}{fig/}
\begin{document}%

%\begin{center}
%{\Large \sc A Distributed Algorithm for Mincut/Maxflow\\[1ex] Combining Path Augmentation and Push-Relabel}\\
%\vskip5mm
%{\large \sc Technical Report}
%\vskip5mm
%Anonymous EMMCVPR-2011 supplemental material for Paper ID 5
%%Anonymous EMMCVPR2011 supplemental material\\[5mm]
%%Paper ID 5\\
%\vskip5mm
%\end{center}
%
%\begin{spacing}{0.9}
\begin{abstract}
  We develop a novel distributed algorithm for the minimum cut problem. We primarily aim at solving large sparse problems. Assuming vertices of the graph are partitioned into several regions, the algorithm performs path augmentations inside the regions and updates of the push-relabel style between the regions. The interaction between regions is considered expensive (regions are loaded into the memory one-by-one or located on separate machines in a network). The algorithm works in sweeps -- passes over all regions.
Let $\B$ be the set of vertices incident to inter-region edges of the graph. We present a sequential and parallel versions of the algorithm which terminate in at most $2|\B|^2+1$ sweeps. The competing algorithm by Delong and Boykov uses push-relabel updates inside regions.
In the case of a fixed partition we prove that this algorithm has a tight $O(n^2)$ bound on the number of sweeps, where $n$ is the number of vertices.
We tested sequential versions of the algorithms on instances of maxflow problems in computer vision. Experimentally, the number of sweeps required by the new algorithm is much lower than for the Delong and Boykov's variant. 
Large problems (up to $10^8$ vertices and $6\cdot 10^8$ edges) are solved using under 1GB of memory in about 10 sweeps.
\end{abstract}
\clearpage
\tableofcontents
%\end{spacing}
\clearpage
\hfill
\begin{minipage}{0.4\linewidth}
\small
{\it
``Needless to say, their version not only has its own real beauty, but is somewhat ``sexy'' running depth first search on the layered network constructed by (extended) breadth first search''} --
\vskip2mm
Y. Dinitz, about Even and Itai's version of the maximum flow algorithm. Citation.
%\begin{flushright} Citation.
%\end{flushright}
\end{minipage}
\vskip10mm
\section{Introduction}
%\mypar{Applications} 
Minimum $s$-$t$ cut ({\sc mincut}) is a classical combinatorial problem with applications in many areas of science and engineering. This research was motivated by wide use of {\sc mincut}/{\sc maxflow} problems in computer vision, where large sparse instances need to be solved. In some cases an applied problem is formulated directly as a {\sc mincut}. More often, however, {\sc mincut} problems in computer vision originate from the energy minimization framework (maximum a posteriori solution in a Markov random field model).
A large subclass of Energy minimization is formed by submodular problems, which reduce to {\sc mincut} ~\cite{Ishikawa03,DSchlesinger-K2}. Instances originating from submodular energies can be very large if the number of labels in the energy minimization is large. When the energy minimization is intractable, {\sc mincut} is employed in relaxation and local search methods. The linear relaxation of pairwise energy minimization with two labels reduces to {\sc mincut}~\cite{Boros:TR91-maxflow,Kolmogorov-Rother-07-QBPO-pami} as well as the relaxation of problems reformulated in two labels~\cite{kohli:icml08}. Expansion-move, swap-move~\cite{Boykov99} and fusion-move~\cite{Lempitsky-09-fusion} algorithms formulate a local improvement step as a {\sc mincut} problem.
\par
Many applications of {\sc mincut} in computer vision use graphs of a regular structure, with vertices arranged into an N-D grid and edges uniformly repeated (\eg 3D segmentation models~\cite{BJ01,BF06,BK03}, 3D reconstruction models~\cite{LB06,BL06,LB07}). Because of this regular structure, the graph itself need not be stored in the memory, only the edge capacities, allowing relatively large instances to be solved by such a specialized implementation. However, in many cases, it is advantageous to have a non-regular structure (\eg using an adaptive tetrahedral volume in 3D reconstruction~\cite{Labutit09,Jancosek11}). Such applications would benefit from a large-scale generic {\sc mincut} solver.
%
%\mypar{Related Work}
\par
The previous research mostly focused on speeding up {\sc mincut} by parallel computations. The following important distinction is to be made: the {\em parallel} computational model assumes that there are several computation units which share the same memory, whereas the {\em distributed} computational model assumes that the computation units have their own separate memory and exchanging the information between them is expensive. A distributed algorithm has therefore to divide the computation and the problem data between the units and keep the communication rate low. %Clearly, a distributed algorithm can run in a parallel mode, as well as in a sequential mode on a single machine using limited memory and a disk storage.
We will consider distributed algorithms, operating in the following two practical usage modes:
\begin{itemize}
\item Sequential (or {\em streaming}) mode, which uses a single computer with a limited memory and a disk storage, reading, processing and writing back a part of data at a time. Since it is easier for analysis and implementation, this mode will be the main focus of this work.
\item Parallel mode, in which the units are computers in a network. We show that sequential algorithms we consider admit full parallelization and prove the correctness and termination of the parallel versions. We also propose their experimental evaluation and comparison to two other state-of-the-art methods.%However, we have not made the corresponding implementation yet, hence the experiments are only with the sequential mode.
\end{itemize}
To represent the cost of information exchange between the units, we use a special related measure of complexity. We call a {\em sweep} the event when all units of a distributed algorithm recalculate their data once. The number of sweeps is roughly proportional to the amount of communication in the parallel mode or disk operations in the streaming mode.
\par
%It was demonstrated that the fastest solver for computer vision problems 
\mypar{Previous Work}
A variant of path augmentation algorithm was shown in~\cite{BK-maxflow} to have the best performance on computer vision problems among sequential solvers. There were several proposals how to parallelize it. Partially distributed implementation~\cite{Liu10} augments paths within disjoint regions first and then merges regions hierarchically. In the end, it still requires finding augmenting paths in the whole problem. Therefore it cannot be used to solve a large problem by distributing it over several computers or by using a limited memory and a disk storage. For the shared memory model there was reported~\cite{Liu10} a near-linear speed-up with up to 4 CPUs for 2D and 3D segmentation problems. %Nevertheless, it archives good speedup for real 2D segmentation problems~\cite{Liu09}, which appear to be simple enough that most of the decision can be made locally (see discussion in Section~\ref{sec:preprocessing}).
\par
A distributed algorithm was obtained in~\cite{Strandmark10} using dual decomposition approach. The subproblems are {\sc mincut} instances on the parts of the graph (regions) and the master problem is solved using the subgradient method. This approach requires solving {\sc mincut} subproblems with real valued capacities and does not have a polynomial bound on the number of iterations. The integer algorithm proposed in~\cite{Strandmark10} is not guaranteed to terminate. Our experiments (Sect.~\ref{sec:exp_parallel}) showed that it did not terminate on some of the instances in 1000 sweeps. In Sect.~\ref{sec:DD} we relate dual variables in this method to flows. 
\par
%The push-relabel algorithm~\cite{GT88} performs many local atomic operations, which makes it a good choice for a parallel or distributed implementation. A distributed version~\cite{Goldberg91} runs in $O(n^2)$ time using $O(n)$ processors and $O(n^2\sqrt{m})$ messages. 
The push-relabel algorithm~\cite{GT88} performs many local atomic operations, which makes it a good choice for a parallel or distributed implementation. A distributed version~\cite{Goldberg91} runs in $O(n^2)$ time using $O(n)$ processors and $O(n^2\sqrt{m})$ messages. However, it is crucial to implement gap relabel and global relabel heuristics for good practical performance~\cite{Cherkassky-94}. The global relabel heuristic can be parallelized~\cite{Anderson95}, but it is difficult to distribute. We should note however, that in the experiments with computer vision problems we made, the global relabel heuristic was not essential.
Delong and Boykov~\cite{Delong08} proposed a coarser granulation of push-relabel, associating a subset of vertices (a region) to each processor. Push and relabel operations inside a region are decoupled from the rest of the graph. This allows to process several non-interacting regions in parallel or run in a limited memory, processing few regions at a time. The gap and relabel heuristics, restricted to the regions~\cite{Delong08} are powerful and distributed at the same time.
\par
%We revisit the region discharge approach by Delong and Boykov~\cite{Delong08} for the case of a fixed partition into regions. We introduce a generic sequential and a generic parallel region discharge algorithms and prove $O(n^2)$ bound on the number of sweeps for the case when region discharge is push-relabel as in~\cite{Delong08}. The parallel algorithm allows computation on neighboring interacting regions to run concurrently. These results may be considered novel for the algorithm~\cite{Delong08}, however they largely would follow as a special case of asynchronous distributed push-relabel~\cite{Goldberg91}.
\mypar{Contribution}
We revisit the algorithm of Delong and Boykov~\cite{Delong08} in the case of a fixed partition into regions. We study a sequential variant and a novel parallel variant of their algorithm, which allows computation on neighboring interacting regions to run concurrently using a conflict resolution similar to the asynchronous parallel push-relabel~\cite{Goldberg91}.
We prove that both variants have a tight $O(n^2)$ bound on the number of sweeps. We then construct a new algorithm, which works with the same partition of the graph into regions but is guided by a different distance function than push-relabel.
%We should note however, that 
%\par
%A coarser granulation was proposed in~\cite{Delong08}, associating a subset of nodes, called regions to each processor. Push and relabel operations inside a region are decoupled from the rest of the graph. Several non-interacting regions can execute computations in parallel until the regions are fully ``discharged''.
%\par
%Related ideas to this work are appear in 
%
%We are solving {\sc mincut} problem by finding a maximum preflow. We do not find the maximum flow in the original graph. This later step can be accomplished (when the maximum preflow is found) in $O(m \log m)$ time using the flow decomposition. Both in theory and practice, it is much easier than finding the maximum preflow itself. Because we are primarily interested in the minimum cut, we did not consider this step or whether it can be distributed.
%This is one step short of finding the maximum flow. 
%This step can be accomplished in $O(m log m)$ time using flow decomposition. Both in theory and practice it is much easier then finding the maximum preflow itself. However we are primarily interested in the minimum cut
%
%
%\mypar{The New Algorithm} 
\par
Given a fixed partition into regions, we introduce a distance function which counts the number of region boundaries crossed by a path to the sink. Intuitively, it corresponds to the amount of costly operations -- network communications or loads-unloads of the regions in the streaming mode. The algorithm maintains a labeling, which is a lower bound on the distance function. Within a region, we first augment paths to the sink and then paths to the boundary nodes of the region in the order of their increasing labels. Thus the flow is pushed out of the region in the direction given by the distance estimate.
We present a sequential and parallel versions of the algorithm which terminate in at most $2|\B|^2+1$ sweeps, where $\B$ is the set of all boundary nodes (incident to inter-region edges). 
%\par
%Finally, in Secion~\ref{sec:preprocessing} is an additional study of local problem simplification which can be used prior to applying the generic algorithm.
%In this paper we focus on the analysis and experimental study of the sequential version. The parallel version and several other important related results are described in~\cite{TR}.
%\par
%The algorithm works in sweeps, where sweep is a higher-level iteration, processing each region once. We present a sequential and parallel versions of the algorithm which terminate in $2|\B|^2+1$ sweeps, where $\B$ is the set of all boundary nodes (incident to inter-region edges).
%\par
\mypar{Other Related Work}
The following works do not consider a distributed implementation but are relevant to our design.
Partial Augment-Relabel algorithm (PAR)~\cite{Goldberg-PAR} in each step augments a path of length $k$. It may be viewed as a lazy variant of push-relabel, where actual pushes are delayed until it is known that a sequence of $k$ pushes can be executed. The algorithm~\cite{Goldberg-Rao-98} incorporates the notion of a length function and a valid labeling \wrt this length. It can be seen that the labeling maintained by our algorithm corresponds to the length function assigning 1 to boundary edges and 0 to intra-region edges. In~\cite{Goldberg-Rao-98} this generalized labeling is used in the context of blocking flow algorithm but not within push-relabel.

\section{Mincut and Push-Relabel}\label{sec:push-relabel}
We will be solving {\sc mincut} problem by finding a maximum preflow\footnote{A maximum preflow can be completed to a maximum flow using flow decomposition, in $O(m \log m)$ time. Because we are primarily interested in the minimum cut, we do not consider this step or whether it can be distributed.}. In this section, we give basic definitions and introduce the push-relabel framework~\cite{GT88}.
\par
By a {\em network} we call the tuple $G=(V,E,s,t,c,e)$, where $V$ is a set of vertices; $E\subset V\times V$, thus $(V,E)$ is a directed graph; $s,t\in V$, $s\neq t$, are {\em source} and {\em sink}, respectively; $c\colon E \to \bbN_0$ is a capacity function; and $e\colon V\backslash\{s,t\} \to \bbN_0$ is an {\em excess} function. Excess can be equivalently represented as additional edges from the source, but we prefer this explicit form. For convenience we let $e(s)=\infty$ and $e(t)=0$. We also denote $n=|V|$ and $m=|E|$.
\par
For $X,Y\subset V$ we will denote $(X,Y) = E\cap (X\times Y)$. %For a function $f$ on $V$ we will denote by $f(A)$ its restriction to $A\subset V$.
For $C\subset V$ such that $s\in C$, $t\notin C$, the set of edges $(C,\bar C)$, with $\bar C = V\backslash C$ is called an $s$-$t$ {\em cut}. The %(extended\footnote{When excess function is zero, this coincide with standard definition of minimum cut problem.}) 
{\sc mincut} problem is
\begin{equation}
\min \Big\{ \sum\limits_{(u,v)\in (C,\bar C)}c(u,v) + \sum_{v\in \bar C } e(v) \Mid C\subset V,\ s\in C,\ t\in \bar C \Big\}.
\end{equation}
The objective is called the {\em cost} of the cut. Without a loss of generality, we assume that $E$ is symmetric -- if not, the missing edges are added and assigned zero capacity. %, which does not change the value of any $s${-}$t$ cut.
\par%\noindent
A {\em preflow} in $G$ is a function $f\colon E\to \bbZ$ satisfying the following constraints:
\begin{subequations}
\label{preflow-constraints}
\begin{align}
\label{flow-1}
f(u,v) \leq c(u,v) \tab \forall (u,v)\in E &\tab \mbox{(capacity constraint),}\\
\label{flow-2}
f(u,v) = -f(u,v) \tab \forall (u,v)\in E &\tab \mbox{(antisymmetry),}\\
\label{preflow-3}
e(v)+\sum\limits_{u\mid (u,v)\in E} f(u,v) \geq 0 \tab \forall v\in V &\tab \mbox{(preflow constraint).}
\end{align}
\end{subequations}
\par%\noindent
A {\em residual network} \wrt preflow $f$ is a network $G_f=(V,E,s,t,c_f,e_f)$ with the capacity and excess functions given by
\begin{subequations}
\begin{align}
c_f &= c-f,\\
e_f(v) &= e(v)+\sum_{u\mid(u,v)\in E} f(u,v), \tab \forall v\in V\backslash\{t\}.
\end{align}
\end{subequations}
By constraints~\eqref{preflow-constraints} it is $c_f\geq 0$ and $e_f \geq 0$. The costs of all $s$-$t$ cuts differ in $G$ and $G_f$ by a constant called the {\em flow value}, $|f| = \sum\limits_{u\mid(u,t)\in E} f(u,t)$. Network $G_f$ is thus up to a constant {\em equivalent} to network $G$ and $|f|$ is a trivial lower bound on the cost of a cut.
%\par
%The total flow to the sink $|f| = \sum\limits_{u\mid(u,t)\in E} f(u,t)$ is called the {\em flow value}. 
Dual to {\sc mincut} is the problem of maximizing this lower bound, \ie finding a maximum preflow:
\begin{equation}
\max\limits_{f} |f| \tab \st \mbox{ constraints~\eqref{preflow-constraints}}.
\end{equation}
\par
We say that $w\in V$ is {\em reachable} from $v\in V$ in network $G$ if there is a path (possibly of length 0) from $v$ to $w$ composed of edges with strictly positive capacities. This relation is denoted by $v\rightarrow w$. If $w$ is not reachable from $v$ we write $v\nrightarrow w$. For any $X,Y \subset V$, we write $X\rightarrow Y$ if there exist $x\in X$, $y\in Y$ such that $x\rightarrow y$. Otherwise we write $X\nrightarrow Y$.
\par
A preflow $f$ is maximum \iff $\{v\mid e(v)>0\} \nrightarrow t$ in $G_f$. In that case the cut $(\bar T, T)$ with $T=\{v\in V \mid v\rightarrow t\mbox{ in } G_f\}$ has value $0$ in $G_f$. Because all cuts are non-negative it is a minimum cut. %is a minimum cut of value $|f|$.
%In that case there exists a cut of cost 0 in $G_f$. Because all cuts are non-negative, it is a minimum cut. The zero cut we will find is $(\bar T, T)$ with $T=\{v\in V \mid v\rightarrow t\mbox{ in } G_f\}$. 
\par%\noindent
A {\em Distance} function $d^*\colon V \to \Natural_0$ in $G$ assigns to $v\in V$ the length of the shortest path from $v$ to $t$, or $n$ if no such path exists. A shortest path cannot have loops, thus its length is not greater than $n-1$. Let us denote $d^{\infty} = n$.
%
%\mypar{The generic push-relabel algorithm} Algorithm~\cite{GT88} operates with a preflow and a labeling. 
\par%\noindent
A {\em labeling} $d\colon V \to \{0,\dots, d^{\infty}\}$ is {\em valid} in $G$ if $d(t)=0$ and $d(u) \leq d(v)+1$ for all $(u,v)\in E$ such that $c(u,v)>0$. Any valid labeling is a lower bound on the distance $d^*$ in $G$. Not every lower bound is a valid labeling. A vertex $v$ is called {\em active} \wrt $(f,d)$ if $e_f(v)>0$ and $d(v)<d^{\infty}$. 
\par
%Throughout this paper will assume that the network is in the equivalent form where capacities of edges $(\{s\},V)$ are zero (obtained by taking residual network \wrt the preflow saturating all $(\{s\},V)$ edges).
All algorithms in this paper will use the following common initialization.
%\halfskip\hrule\halfskip
\par\noindent
\begin{procedure}[H]
\SetKwFunction{Init}{Init}
\SetKwFunction{Push}{Push}
\SetKwFunction{Relabel}{Relabel}
\SetKwFunction{Discharge}{Discharge}
$f:=$ preflow saturating all $(\{s\},V)$ edges; $G := G_f$; $f:=0$\;
$d:=0$, $d(s):= d^{\infty}$\;
\caption{Init()}
\end{procedure}
The generic push-relabel algorithm~\cite{GT88} starts with \Init and applies the following \Push and \Relabel operations while possible:
\begin{itemize}
%\par\noindent
\item
\Push{$u,v$} is applicable if $u$ is active and $c_f(u,v)>0\ \mbox{and}\ d(u)=d(v)+1$. The operation increases $f(u,v)$ by $\Delta$ and decreases $f(v,u)$ by $\Delta$, where $\Delta = \min(e_f(u),c_f(u,v))$.
%\par\noindent
\item
\Relabel{$u$} is applicable if $u$ is active and $\forall v\mid (u,v)\in E$, $c_f(u,v)>0$  it is $d(u)\leq d(v)$. It sets $d(u):= \min\big(d^{\infty},\min\{ d(v)+1 \mid (u,v)\in E,\ c_f(u,v)>0\}\big)$.
\end{itemize}
If $u$ is active then either \Push or \Relabel operation is applicable to $u$. The algorithm preserves validity of labeling and stops when there are no active nodes. Then for any $u$ such that $e_f(u)>0$, we have $d(u)=d^{\infty}$ and therefore $d^*(u)=d^{\infty}$ and $u\nrightarrow t$ in $G_f$, so $f$ is a maximum preflow.
%Both operations preserve the property that $d$ is a valid labeling. The algorithm stops when there are no active nodes. Then for any $v$ such that $e_f(v)>0$, we have $d(v)=n$ and therefore $d^*(v)=n$ and $v\nrightarrow t$ in $G_f$, so $f$ is a maximum preflow. %So there is no path in $G$ from source to the sink and we have reached an optimal equivalent network. 
%
%\par
%
%
\section{Region Discharge Revisited}\label{sec:PRD}
We now review the approach of Delong and Boykov~\cite{Delong08} and reformulate it for the case of a fixed graph partition. We then describe generic sequential and parallel algorithms which can be applied with both push-relabel and augmenting path approaches.
\par
Delong and Boykov~\cite{Delong08} introduce the following operation.
%This discharge operation was extended by Delong and Boykov~\cite{Delong08} to regions. 
The {\em discharge} of a {\em region} $R\subset V\backslash\{s,t\}$ applies \Push and \Relabel to $v\in R$ until there are no active vertices left in $R$. This localizes computations to $R$ and its {\em boundary}, defined as
%Let us call {\em region boundary} of $R$ the set\ \ 
%\vspace{-\baselineskip}
\begin{equation}
%$
B^R = \{w \mid \exists u\in R\ (u,w) \in E, w\notin R,\ w\neq s,t \}.
%$
\end{equation}
 When a \Push is applied to an edge $(v,w)\in (R,B^R)$, the flow is sent out of the region.
%\par
We say that two regions $R_1, R_2\subset V\backslash\{s,t\}$ {\em interact} if $(R_1,R_2)\neq \emptyset$. Discharges of  non-interacting regions can be performed in parallel since the computations in them do not share the data. %because nodes and edges modified by discharge of $R_1$ are different from those modified by discharge of $R_2$.
The algorithm proposed in~\cite{Delong08} %for solving {\sc mincut} problem on $G$ can be informally described as follows:
repeats the following steps until there are no active vertices in $V$:
\par
\begin{compactenum}
%\item[] {\bf Repeat}
\item Select several non-interacting regions, containing active vertices.
\item Discharge the selected regions in parallel, applying region-gap and region-relabel heuristics.%\footnote{All heuristics (global-gap, region-gap, region-relabel) serve to improve the distance estimate. Details in~\cite{Cherkassky-94,Delong08,TR}. They are very important in practice, but do not affect theoretical properties. }.
\item Apply global gap heuristic.
\end{compactenum}
\vskip0.5ex
\par
All heuristics (global-gap, region-gap, region-relabel) serve to improve the distance estimate. They are very important in practice, but do not affect theoretical properties and will be discussed in Section~\ref{sec:implementation}, devoted to the implementation.
\par 
While the regions in~\cite{Delong08} are selected dynamically in each iteration trying to divide the work evenly between CPUs and cover the most of the active nodes, we restrict ourselves to a fixed collection of regions $(R_k)_{k=1}^K$ forming a partition of $V\backslash\{s,t\}$ and let each region-discharge to work on its own separate subnetwork. % and communicate to other regions by passing ``messages'' over the region boundary. 
%\par\noindent
We define a {\em region network} $G^R = (V^R, E^R,s,t, c^R, e^R)$, where $V^R=R\cup B^R\cup\{s,t\}$; $E^R = (R\cup \{s,t\}, R\cup \{s,t\})\cup(R,B^R)\cup (B^R,R)$; 
$c^R(u,v) = c(u,v)$ if $(u,v)\in E^R\backslash (B^R,R)$ and $0$ otherwise; $e^R=e|_{R\cup\{s,t\}}$ (the restriction of function $e$ to its subdomain ${R\cup\{s,t\}}$). This network is illustrated in \Figure{fig:region-network}(b).
\begin{figure}[!ht]
\centering
\begin{tabular}{ccc}
\begin{tabular}{c}
\includegraphics[width=0.25\linewidth]{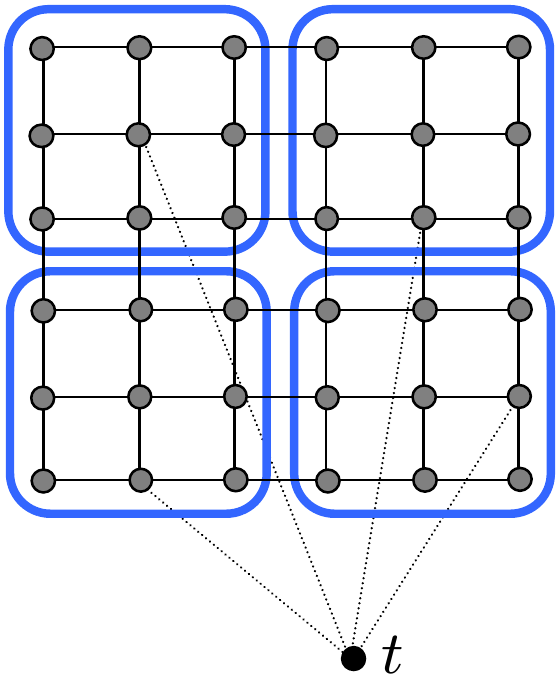}
\end{tabular} & \parbox{0.1\linewidth}{\ } &
\begin{tabular}{c}
\includegraphics[width=0.25\linewidth]{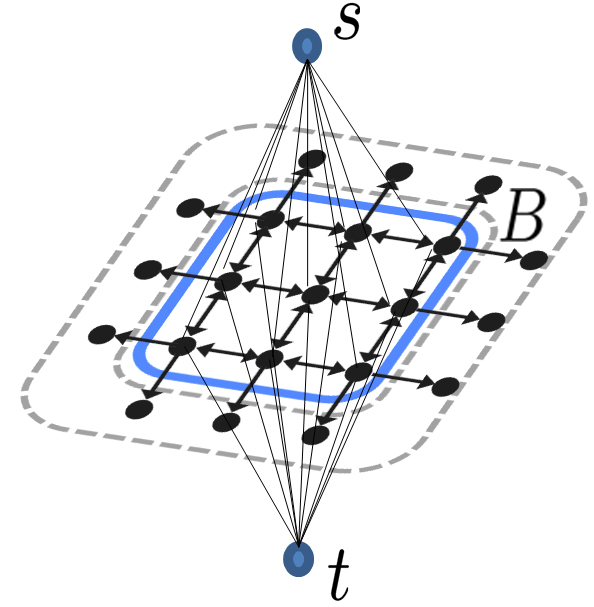}
\end{tabular}\\
(a) & & (b)
\end{tabular}
\caption{(a) Partition of a network into 4 regions. (b) Region Network}
\label{fig:region-network}
\end{figure}
Note that the capacities of edges coming from the boundary, $(B^R,R)$, are set to zero. Indeed, these edges belong to a neighboring region network. The region discharge operation of~\cite{Delong08}, which we refer to as Push-relabel Region Discharge (PRD), can now be defined as follows.
\par
%\begin{algorithm}[H]
\begin{procedure}[H]
\DontPrintSemicolon
\SetKwFunction{PRD}{PRD}
%\SetKwFor{Sub}{Sub}{}{}
%\Sub(\tcc*[f]{ assume $d\colon V^R\to \bbN_0$ valid in $G^R$}){}{
\tcc{assume $d\colon V^R\to \{0,\dots, d^{\infty}\}$ valid in $G^R$}
\While{$\exists v\in R$ active}{
	apply \Push or \Relabel to $v$ \tcc*{changes $f$ and $d$}
	apply region gap heuristic (Section~\ref{sec:implementation}) \tcc*{optional}
}
%}
%\caption{Push-Relabel Region {\sf Discharge} (PRD)~\cite{Delong08}\label{PRD}}
\SetAlgoProcName{Procedure $(f,d)=$}{PRD}
\caption{PRD($G^R$,$d$)\label{PRD}}
\end{procedure}

\mypar{Generic Region Discharge Algorithms}
We now define generic sequential \Algorithm{alg:sequential} and parallel \Algorithm{alg:parallel}, which use a black-box \Discharge function as a subroutine. The sequential algorithm takes regions one-by-one from the partition and applies the \Discharge operation to them. 
For non-interacting regions, their \Discharge operations are independent and can be executed in parallel. The sequential algorithm can be implemented as several phases, where in each phase a subset of non-interacting regions from the partition is taken and processed in parallel. The number of phases required in a general case will correspond to the minimal coloring of the region interaction graph. Instead, our parallel algorithm calls \Discharge for all regions concurrently and then resolves conflicts in the flow similarly to the asynchronous parallel push-relabel~\cite{GT88}. A conflict occurs if two interacting regions increase their labels on the vertices of a boundary edge $(u,v)$ simultaneously and try pushing flow over it. In such a case, we accept the labels, but do not allow the flow to cross the boundary in one of the directions.
\par
 %Both algorithms \ref{alg:sequential},~\ref{alg:parallel} are generalized to take a black-box \Discharge function. 
In the case $\Discharge$ is \PRD the sequential and parallel algorithms are implementing the push-relabel approach and will be referred to as S-PRD and P-PRD respectively. S-PRD is a sequential variant of~\cite{Delong08} and P-PRD is a novel variant, based on results of~\cite{GT88} and~\cite{Delong08}.
\par
\begin{algorithm}[H]
\DontPrintSemicolon
%\Begin{
%$d|_{V\backslash \{s\}}:= 0$, $d(s):= d^{\infty}$ \tcc*{init}
%$f:=$ preflow saturating $(\{s\},V)$ edges; $G:=G_f$ \;%\tcc*{init}
\Init\;
\While(\tcc*[f]{a sweep}){there are active vertices}{
	\For{$k = 1,\dots K$}{
		Construct $G^{R_k}$ from $G$ \;
		$(f',d') := \Discharge(G^{R_k},d|_{V^{R_k}})$ \;
		$G := G_{f'}$ \tcc*{apply $f'$ to $G$}
		$d|_{R_k} := d'|_{R_k}$ \tcc*{update labels}
		apply global gap heuristic (Section~\ref{sec:implementation}) \tcc*{optional}
	}
}
%(f',d') = {\bf Discharge} $(G^R,d)$ \;
\caption{Sequential Region Discharge}\label{alg:sequential}
\end{algorithm}
\vspace{-2\myskip}\vspace{-2pt}
\begin{algorithm}[H]
\DontPrintSemicolon
%$d|_{V\backslash \{s,t\}}:= 0$, $d(t):= 0$, $d(s):= n$ \tcc*{init}
%$d|_{V\backslash \{s\}}:= 0$, $d(s):= d^{\infty}$ \tcc*{init}
\Init\;
\While(\tcc*[f]{a sweep}){there are active vertices}{
		%Construct $G^{R_k}$ from $G$ \;
		$(f'_k,d'_k) := \Discharge(G^{R_k},d|_{V^{R_k}})$ \ $\forall k$ \tcc*{in parallel}
		$d'|_{R_k} := d'_k|_{R_k}$ $\forall k$ \tcc*{fuse labels}
		$\alpha(u,v) := \leftbb d'(u) \leq d'(v) +1\rightbb$ \ $\forall (u,v)\in (\B,\B)$\label{fuse-flow-construction} \tcc*{valid pairs}
		\tcc{fuse flow}
		$f'(u,v) := \begin{cases}
			\alpha(v,u) f'_k(u,v)+\alpha(u,v) f'_j(u,v) & \IF (u,v)\in (R_k,R_j)\\
			f'_k(u,v) & \IF (u,v)\in (R_k,R_k)\\
		\end{cases}
		$ \;
		$G := G_{f'}$ \tcc*{apply $f'$ to $G$}
		$d := d'$ \tcc*{update labels}
		global gap heuristic (Section~\ref{sec:implementation}) \tcc*{optional}
}
%(f',d') = {\bf Discharge} $(G^R,d)$ \;
\caption{Parallel Region Discharge}\label{alg:parallel}
\end{algorithm}
\par

\par
We prove below that both S-PRD and P-PRD terminate with a valid labeling in at most $2n^2$ sweeps. Parallel variants of push-relabel~\cite{Goldberg-PhD} have the same bound on the number of sweeps. However, they perform much simpler sweeps, processing every node only once, compared to S/P-PRD. A natural question is whether the $O(n^2)$ bound is not too loose for S/P-PRD. %More specifically, we can ask whether there is a better bound in terms of $n$ for a partition into $K$ regions and 
In Appendix A we give an example of a graph, its partition into two regions and a sequence of valid \Push and \Relabel operations, implementing S/P-PRD which takes $O(n^2)$ sweeps to terminate. The set of inter-region edges in this example is also constant, which shows that a better bound in terms of these characteristics is not possible. %We might even seek for the bound which is dependent on the characteristics of the region partitioning, such as region size and the cardinality of the separator set. An example in Appendix A shows that the bound is tight. 
%
%It can be seen that S-PRD is doing some extra work 
%For S-PRD, a bound on the number of sweeps can be proven the same way as for synchronous parallel push-relabel~\cite{Goldberg-PhD}. It can be observed that in each sweep of Sequential Algorithm~\ref{alg:sequential} every active vertex of $G$ is touched at least once. It follows that Sequential Algorithm~\ref{alg:sequential} requires at most $O(n^2)$ sweeps. Next, we will present a detailed proof of this bound, based on somewhat different considerations. An example in Appendix A shows that the bound is tight. 
%
%With the new {\rm Discharge} method it will take only $O(|\B|^2)$ sweeps. %The same example demonstrates that $|\B|^2$ 
%In the next section we present a new distance function and region discharge method, for which Sequential RD terminates in $O(|\B|^2)$ sweeps.
%In the next section we present a new distance function and region discharge method, for which Sequential RD terminates in $O(|\B|^2)$ sweeps. %But first we prove the $O(n^2)$ bound for a class of region discharge algorithms.
%
%
%
\subsection{Complexity of Sequential Push-Relabel Region Discharge}
%We now prove an $O(n^2)$ bound for the Sequential Algorithm~\ref{alg:sequential} with a {\rm Discharge} operation, which somewhat generalizes PRD. 
%We now prove an $O(n^2)$ bound for the Sequential Algorithm~\ref{alg:sequential} with a discharge operation .
%We now prove an $O(n^2)$ bound for the Sequential Algorithm~\ref{alg:sequential} with the Discharge operation being PRD. 
Our proof follows the main idea of the similar result for parallel push-relabel in~\cite{Goldberg-PhD}. The main difference is that we try to keep \Discharge operation as abstract as possible. Indeed, it will be seen that proofs of termination of other variants follow the same pattern, using several important properties of the \Discharge operation, abstracted from the respective algorithm. Unfortunately, to this end we do not have a unified proof, so we will analyze all cases separately.
%, so the needed properties of the \Discharge operation will differ in 
%The proof for the case when Discharge is ARD, in the next section, will have a very similar pattern. We think that it is possible to unify these results. It would lead, however, to some unnecessary complications.
%The proof follows the main idea f
%This result is complementary. It is not used in the main development and can be safely skipped.
%\par
%The next definition generalizes description of what PRD algorithm computes for the network $G^R$.
%The next definition lists 
%\par
%Towards having a unified proof, we asked the question: what are the properties to be satisfied by an output $(d',f')$ of \PRD in order that Sequential Algorithm~\ref{alg:sequential} terminated with a maximum preflow? If the pair $(d',f')$ is computed by a black box algorithm it is difficult to check whether it may be obtained as a sequence of push and relabel operations. Fortunately, the following simpler conditions are sufficient for the correctness and complexity bound of the Sequential PRD.
\par
%For the pair $(d',f')$ is computed by a black box algorithm it is difficult to check whether it may be obtained as a sequence of push and relabel operations.
%The following conditions on the output $(f',d')$ of a \Discharge function are sufficient for the correctness and complexity bound of S-PRD.
%
%Let us prove the following key properties which are satisfied by PRD and are sufficient for correctness and complexity bound of the Sequential PRD.
%The generality is in that it does not specify how the resulting flow and labeling are computed. For a given pair of input-output they are also easy to verify.
%These properties a sufficient for correctness and complexity bound and for a given pair $(f',d')$ they are easy to verify.
%
%
%\begin{definition}[Region Discharge $(d_E)$]\label{def:RD}
% An operation is called {\em Generic PRD} of $R$ %\wrt distance in $G$ 
%if it satisfies the following:\\
%{Input}: Network $G^R$, valid labeling $d$ of $R\cup B$\\
%{Result}: preflow $f'$ and labeling $d'$\\
\begin{statement}[Properties of PRD]\label{PRD properties}\ \\%[-10pt]
Let $(f',d') = \PRD(G^R,d)$, then
%\parbox{0.7\linewidth}{
\begin{compactenum}
\item\label{RD-flow} there are no active nodes in $R$ \wrt $(f',d')$ \hfill (optimality)
\item\label{RD-monotone} $d' \geq d$; $d'|_{B^R} = d|_{B^R}$ \hfill (labeling monotony)
\item\label{RD-LB} $d'$ is valid in $G^R_{f'}$ \hfill (labeling validity)
\item\label{RD-flow-dir} $f'(u,v) > 0 \Rightarrow d'(u)> d(v)$, $\forall (u,v)\in E^R$ \hfill (flow direction)
\end{compactenum}
%}
\end{statement}
\begin{proof}
%Push operations applied to $v\in R$ modify only edges of the region network and satisfy preflow constraints, so $f'$ is a preflow in $G^R$.
\begin{enumerate}
\item[1.] Optimality. This is the stopping condition of \PRD.
\item[2,3.] Labeling validity and monotony: labels are never decreased and the \Push operation preserves labeling validity~\cite{GT88}. Labels not in $R$ are not modified.
%
%We need to prove validity on edges $(u,v)$, $u\in R_k$, $v\in B_k$, for which network $G^{R_k}$ have capacity $c^R(v,u)=0$ and labeling $d'$, valid in $G^R$ is not yet proven to satisfy inequality $d'(v)\leq d'(u)+1$. If $c(v,u)=0$ we don't need to prove anything. 
%Let $c{f'}(v,u)>0$. Because $f'(u,v)\geq 0$ it must be $c(v,u)>0$. Let $d$ be a valid labeling in $G$ before Discharge. Then $d(v) \leq d(u)+1$. The output labeling $d'$ satisfy $d'(u) \geq d(u)$ and $d'(v) = d(v)$. We thus have $d'(v) = d(v)\leq d(u)+1 \leq d'(u)+1$. So $d'$ is valid in $G_{f'}$.
%
\item[4.] Flow direction: let $f'(u,v)>0$, then there was a push operation from $u$ to $v$ in some step. Let $\tilde d$ be the labeling on this step. We have $\tilde d(u) = \tilde d(v)+1$. Because labels never decrease, $d'(u)\geq \tilde d(u)> \tilde d(v) \geq d(v)$.
%\item[(unneeded)] Claim: push is not applicable to $v\in B$.
%Let $v\in B$ be active. We have $d'(v) = d(v)$ and there must exist $u\in R$ such that $f'(u,v)>0$. Then $d'(u) > d'(v)$ by flow direction property. So push is not applicable in $v$. %there $f()$%, so $f'$ is maximal in $G^R$.
%
\qed
\end{enumerate}
\end{proof}
These properties are sufficient to prove correctness and the complexity bound of S-PRD. They are abstract from the sequence of \Push and \Relabel operation done by \PRD and for a given pair $(f',d')$ they are easy to verify. For correctness of S-PRD we need to verify that it maintains a labeling, which is globally valid.
\begin{statement}\label{labeling-extension}
Let $d$ be a valid labeling in $G$. Let $f'$ be a preflow in $G^R$ and $d'$ a labeling in $G^R_{f'}$ satisfying properties~\ref{RD-monotone} and~\ref{RD-LB} of Statement~\ref{PRD properties}. Extend $f'$ to $E$ by letting $f'|_{E\backslash E^R}=0$ and extend $d'$ to $V$ by letting $d'|_{V\backslash R} = d$. Then $d'$ is valid in $G_{f'}$.
\end{statement}
\begin{proof}
We have that $d'$ is valid in $G^R_{f'}$. For edges $(u,v)\in(V\backslash R,V\backslash R)$ labeling $d'$ coincides with $d$ and $f'(u,v)=0$. It remains to verify validity on edges $(v,u)\in (B^R,R)$ in the case $c^R_f(v,u)=0$ and $c_f(v,u) > 0$. (These are the incoming boundary edges which are zeroed in the network $G^R$). Because $0 = c^R_f(v,u)= c^R(v,u) - f(v,u) = -f(v,u)$, we have $c_f(v,u)=c(v,u)$. 
Since $d$ was valid in $G$,  $d(v) \leq d(u)+1$. The new labeling $d'$ satisfies $d'(u) \geq d(u)$ and $d'(v) = d(v)$. It follows that $d'(v) = d(v)\leq d(u)+1 \leq d'(u)+1$. Hence $d'$ is valid in $G_{f'}$.
\qed
\end{proof}
We can now state the termination.
\begin{theorem}\label{T1}
Sequential PRD terminates in at most $2 n^2$ sweeps.
\end{theorem}
\begin{proof}
\begin{itemize}
\item Value of $d$ does not exceed $n$ for every node.
\item Because there are $n$ nodes, $d$ can be increased at most $n^2$ times.
\item Let $\Phi = \max \{d(v)\mid v\in V,\ \mbox{ $v$ is active in $G$ } \}$. This value may go up and down during the algorithm, but the total number of times it can change is bounded.
%\item After the first sweep, active excess is only on $\B$.
\end{itemize}
\begin{enumerate}
\item Each sweep the increase of $\Phi$ is no more than the total increase of $d$.
\begin{quote}
Let us consider a discharge of region $R_k$. %Let all quantities after the discharge be distinguished with a prime. We need to show that
Let $(f',d')$ be the the flow and labeling computed by the discharge. Let $f'$ be extended to $E$ by setting $f'|_{E\backslash E^R} = 0$ and $d'$ be extended to $V$ by setting $d'|_{V\backslash R} = d|_{V\backslash R}$.
Let $G'=G_{f'}$ and $\Phi'$ be the new function after applying the discharge. We need to show that
\begin{equation}
\Phi' - \Phi \leq \sum_{v\in R_k} [d'(v)-d(v)]
\end{equation}
Let the maximum in the definition of $\Phi'$ be achieved at a node $v$, so $\Phi' = d'(v)$. Then either $v\notin R_k\cup B^{R_k}$, in which case $\Phi' \leq \Phi$ (because the label and the excess of $v$ in $G$ and $G'$ are the same), or $v\in R_k \cup B^{R_k}$ and there exists a path $(v_0,v_1,\dots v_l)$, $v_l=v$, $v_0,\dots v_{l-1}\in R_k$, such that $f'(v_{i-1},v_{i}) > 0$, $i=1\dots l$ and $v_0$ is active in $G$. We have $\Phi \geq d(v_0)$, therefore
\begin{equation}\label{T1eq1}
\begin{split}
\Phi' - \Phi \leq d'(v_l) - d(v_0) = \sum_{i=1}^l[d'(v_i) - d'(v_{i-1})] + [d'(v_0) - d(v_0)] \\
\stackrel{(a)}{\leq} \sum_{i=0}^l[d'(v_i) - d(v_{i})]\\
\stackrel{(b)}{\leq} \sum_{v\in R_k \cup B^{R_k}} [d'(v)-d(v)] 
\stackrel{(c)}{=} \sum_{v\in R_k} [d'(v)-d(v)],
%[d'(v_l) - d'(v_0)] + [d'(v_0) - d(v_0)].
\end{split}
\end{equation}
where inequality (a) is due to the flow direction property (statement~\ref{PRD properties}.\ref{RD-flow-dir}) which implies $d'(v_{i-1})>d(v_{i})$, the inequality (b) is due to monotony property (statement~\ref{PRD properties}.\ref{RD-monotone}) and to $v_i \subset R_k \cup B^{R_k}$; and the equality (c) is due to $d'|_{B^{R_k}} = d|_{B^{R_k}}$.
\par
Summing over all regions, which are disjoint, we obtain the claim.
\end{quote}
\item If $d$ has not increased during a sweep ($d'=d$) then $\Phi$ decreases at least by 1. Indeed, let us consider the set of vertices having the label greater or equal to the label of the highest active node, $H=\{v\mid d(v)\geq \Phi\}$. These vertices do not receive flow during all discharge operations due to the flow direction property. After discharging $R_k$, there are no active vertices in $R_k \cap H$ (statement~\ref{PRD properties}.1). Therefore, there are no active vertices in $H$ after the full sweep.
%\begin{itemize}
%\item after RD there is no active excess in $R$
%\item nodes having the highest label in $\B$ do not receive flow from outside.
%\end{itemize}
\end{enumerate}
\par
In the worst case, $\Phi$ can increase by one $n^2$ times and decrease by one $n^2$ times. In at most $2 n^2$ sweeps there is no active excess.
\qed
\end{proof}
Once the algorithm terminated with network $G$, equivalent to the initial one, we have that labeling $d$ is valid in $G$ and there are no active vertices. Hence the cut $(\bar T,T)$, with $T=\{v\mid v\rightarrow t \mbox{ in } G\}$ is a minimum cut.
%\par
%Its main property is that preflow is only 
%\section{Fixed Region Split}
%Let us introduce a more general definition of region discharge operation, which does not specify how the result should be computed but rather characterize the properties of a possible output. As we will see these properties are easy to check, they are satisfied by the implementation described above and are sufficient to prove correctness and complexity bounds.
%
%

%
%
\subsection{Complexity of Parallel Push-Relabel Region Discharge}\label{sec:PRDparallel}
Let us show that the following properties hold for a sweep of P-PRD. 
\begin{statement}
Let $d$ be a valid labeling in the beginning of a sweep of P-PRD. Then the pair of fused flow and labeling $(f',d')$ satisfies:
\begin{enumerate}
\item $d' \geq d$; \hfill {\em (labeling monotony)}
\item $d'$ is valid in $G_{f'}$; \hfill {\em (labeling validity)}
\item $f'(u,v) > 0 \Rightarrow d'(u)> d(v)$, $\forall (u,v)\in E$ ; \hfill {\em (flow direction)}
%$f'$ satisfies the flow direction property.
\end{enumerate}
\end{statement}
\begin{proof}
\begin{enumerate}
\item We have $d'_{R^k} \geq d|_{R^k}$ for all $k$.
\item We have to prove validity for the boundary edges, where the flow and the labeling are fused from different regions. It is sufficient to study the two regions case. Denote the regions $R_1$ and $R_2$. The situation is completely symmetric \wrt orientation of a boundary edge $(u,v)$. Let $u\in R_1$ and $v\in R_2$. Let only $d'(v) \leq d'(u)+1$ be satisfied and not $d'(u) \leq d'(v)+1$. By the construction in step~\ref{fuse-flow-construction} of Alg.~\ref{alg:parallel} flow $f_2$ is canceled and $f'(u,v) = f'_{1}(u,v) \geq 0$. Suppose $c_{f_1'}(u,v) >0$, then we have 
that $d_1'(u) \leq d_1'(v)+1$, because $d_1'$ is valid in $G^{R_1}_{f'_1}$. It follows that $d'(u) = d_1'(u) \leq d_1'(v)+1 = d(v)+1 \leq d_2'(v)+1 = d'(v)+1$, where we also used labeling monotonicity property. The inequality $d'(u) \leq d'(v)+1$ is a contradiction, therefore it must be that $c_{f'}(u,v) = 0$. The labeling $d'$ is valid on $(u,v)$ in this case. Note that 
inequalities $d'(v) \leq d'(u)+1$ and $d'(u) \leq d'(v)+1$ cannot be violated simultaneously. In the remaining case, when both are satisfied, the labeling is valid for arbitrary flow on $(u,v)$, so no flow is canceled.
%$d'(u) = d_1'(u) \leq d_1'(v)+1 = d(v)+1 \leq d_2'(v) = d'(v)$. So $d'$ is valid on $(u,v)$.
%If $c_{f'}(u,v)=0$ then constraint $d'(v) \leq d'(u)+1$ is not required for validity.
%
%, then augmenting any positive flow on $(u,v)$ may create a positive residual edge $(v,u)$, which does not break validity.
\item If $f'(u,v)>0$ then $f'_k(u,v)>0$ and there was a push operation from $u$ to $v$ in the discharge of region $R_k\ni u$. Let $\tilde d_k$ be the labeling in $G^R_k$ on this step. We have $d'(u) \geq \tilde d_k(u) = \tilde d_k(v)+1 \geq d(v)+1 > d(v)$. %$(f'_k,d'_k)$ satisfies flow direction property on $G^R_k$. 
%For all pairs $(u,v)$, either flow is canceled or the flow direction property is preserved. 
%Note, on the boundary edges $(u,v)\in \B$  inequalities hold: $d'(u)> d(v)$ and $d'(v)> d(u)$.
%On pairs where flow was not canceled, the flow direction property is preserved.
%On pairs $(u,v)$, such that $u \in R_k$ $v\in R_k$, the flow direction property is preserved. On pairs $(u,v)$, such that $u \in R_k$ $v\in B_k$ 
%$f'$ is a sum of flows, satisfying flow direction property. Some of the path flows got shortened, because we didn't let them go over the boundary, where it would invalidate the labeling. However, shortened paths also satisfy (non-strict) flow direction property. This would suffice for the proof of complexity. In case that no label is risen, all flow is accepted and satisfy strict flow direction.
\end{enumerate}
\qed
\end{proof}
\begin{theorem}
Parallel PRD terminates in at most $2n^2$ sweeps.
\end{theorem}
\begin{proof}
As before, total increase of $d$ is at most $n^2$. Let us verify that if $d$ has increased during a sweep, then increase in $\Phi$ is no more than total increase of $d$. Consider the pair $(f',d')$ of fused flow and labeling, constructed by Parallel PRD in a sweep. 
%Consider a single region $R=V\backslash\{s,t\}$ and a discharge operation on this region constructing flow and labeling $(f',d')$ on $G^R = G$ by fusing ((f'_k,d'_k)\mid k\in 1\dots K) . 
As shown above, this pair satisfies properties 2-4 of statement~\ref{PRD properties} for region $R=V\backslash\{s,t\}$ and may be considered as a single Discharge operation on this region. Part 1 of the proof of Theorem~\ref{T1} can be applied, considering region $R$ and $(f',d')$ on it.
\par
If $d$ has not increased during a sweep ($d'=d$) then no relabel operation has occurred and $\alpha(u,v)=1$ for all $(u,v)\in \B$. Moreover, for each $(u,v)\in (R_k, R_j)$ either $f'_k(u,v) = 0$ or $f'_j(u,v) = 0$ so no flow is canceled by $\alpha$ or by the opposite flow on the fusing step. Let $H=\{v\mid d(v)\geq \Phi\}$. These vertices do not receive flow during all elementary push operations. After the sweep, all active vertices which were in $H$ are discharged. Because there is no active vertices with label $\Phi$ or above left, it is $\Phi' < \Phi$.
%
%It is a sequential algorithm with a single region $R=V\backslash\{s,t\}$ and Discharge operation constructing $(f',d')$ on $G^R = G$ satisfying all properties in statement~\ref{PRD properties}.
By the same argument as in Theorem~\ref{T1}, the algorithm will terminate in at most $2n^2$ sweeps.
\qed
\end{proof}
\section{Augmented Path Region Discharge}\label{sec:ARD}
We will now use the same setup of the problem distribution, but replace the discharge operation and the labeling function. %Because this is our main contribution, it is presented in full detail.
%In this section, we introduce a new distance function and a new Discharge operation.
%In other words $\B$ consists of nodes which have a neighbor in a different region. Sequential and parallel algorithms~\ref{alg:sequential},\,\ref{alg:parallel} with the new Discharge operation will be proven to terminate in $O(|\B|^2)$ sweeps.
%We show that the new algorithm satisfy definition~\ref{def:RD}. Because it also has a special additional property of not rising the flow up in the labels, the sequential algorithm~\ref{alg:sequential} has a tighter bound of $O(|\B|^2)$ sweeps.
%
\subsection{New Distance Function}
%Our new distance counts the number of regions which have to be traversed by a path to reach the sink. This is directly the complexity measure we are interested in: the smaller the distance, the smaller the number of region discharge operations required.
\par
Let the {\em boundary} \wrt partition $(R_k)_{k=1}^{K}$ be the set $\B = \bigcup_k B^{R_k}$.
%\par\noindent
The {\em region distance}
$d^{*\B}(u)$ in $G$ is the minimal number of inter-region edges contained in a path from $u$ to $t$, or $|\B|$ if no such path exists:
\begin{equation}
d^{*\B}(u) = \begin{cases}
\min\limits_{P = ((u,u_1),\dots, (u_r,t)) } |P \cap (\B,\B)| & \IF u\rightarrow t, \\
|\B| & \IF u\nrightarrow t.
\end{cases}
\end{equation}
This distance corresponds well to the number of region discharge operations required to transfer the excess to the sink (see Figure~\ref{fig:distance-ARD}(a)).
\begin{figure}[!ht]
\centering
\begin{tabular}{ccc}
\begin{tabular}{c}
\includegraphics[width=0.25\linewidth]{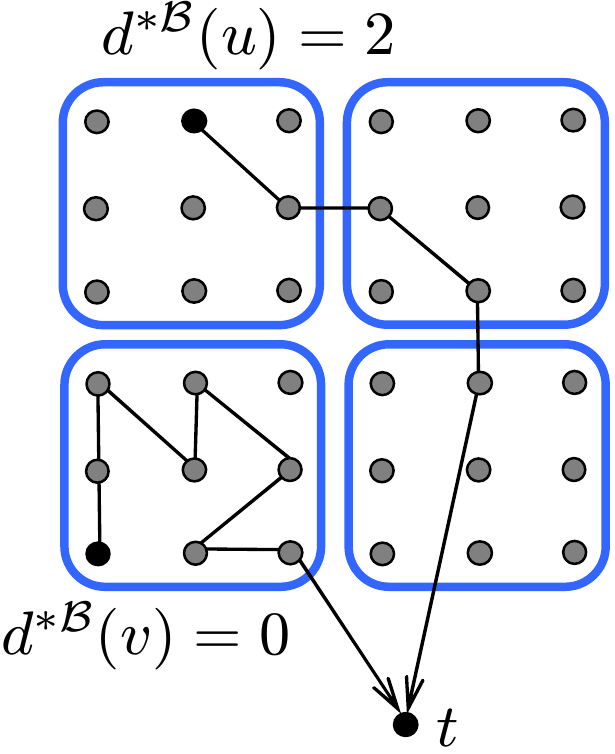}
\end{tabular} & \parbox{0.1\linewidth}{\ } &
\begin{tabular}{c}
\includegraphics[width=0.35\linewidth]{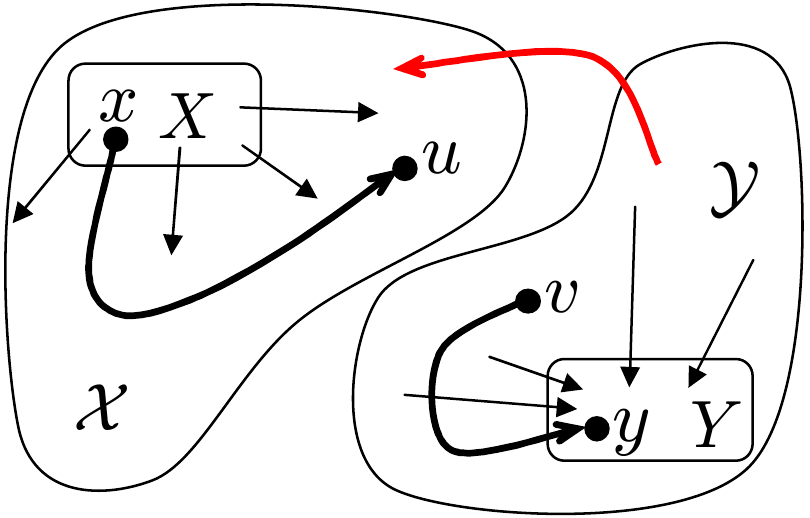}
\end{tabular}\\
(a) & & (b)
\end{tabular}
\caption{(a) Illustration of region distance. (b) Illustration of Lemma~\ref{ARD lemma}: augmentation on pathes from $x$ to $u$ or from $v$ to $y$ preserves $X\nrightarrow Y$, but not the augmentation on the red path.}
\label{fig:distance-ARD}
\end{figure}
\begin{statement}
If $u\rightarrow t$ then $d^{*\B}(u) < |\B|$.
\end{statement}
\begin{proof}
Let $P$ be a path from $u$ to $t$ given as a sequence of edges. If $P$ contains a loop then it can be removed from $P$ and $|P \cap (\B,\B)|$ will not increase. A path without loops goes through each vertex at most once. For $\B\subset V$ there is at most $|\B|-1$ edges in the path which have both endpoints in $\B$.
\qed
\end{proof}
We now let $d^{\infty} = |\B|$ and redefine a valid labeling \wrt to the new distance.
A labeling $d\colon V\to \{0,\dots,d^{\infty}\}$ is {\em valid} in $G$ if $d(t) = 0$ and for all $(u,v)\in E$ such that $c(u,v)>0$:
\begin{align}
&d(u) \leq d(v)+1 &\mbox{ if } (u,v)\in (\B,\B),\\
&d(u) \leq d(v) &\mbox{ if } (u,v)\notin (\B,\B).
\end{align}
\begin{statement}\label{LB property} A valid labeling $d$ is a lower bound on $d^{*\B}$.
\end{statement}
\begin{proof}
If $u\nrightarrow t$ then $d(u)\leq d^{*\B}$. Otherwise, let $P = ((u,v_1),\dots, (v_l,t))$ be a shortest path \wrt $d^{*\B}$, \ie $d^{*\B}(u) = |P \cap (\B,\B)|$. Applying the validity property to each edge in this path, we have $d(u) \leq d(t) + |P \cap (\B,\B)| = d^{*\B}(u)$.
\qed
\end{proof}
%
%\begin{statement}
%Let $d$ be a valid labeling. If $u\rightarrow t$ then $d(u)=0$. If $u\rightarrow w$, $w\in B^R$, then $d(u)\leq d(v)+1$.
%\end{statement}
%\begin{proof}
%Let $u\rightarrow t$. Then there exists a residual path $(u,v_1,\dots, v_l,t)$, $v_i\in R$.
%\end{proof}
%
%\par
%Let $f$ be a preflow in $G$ and $d$ a labeling. A vertex $v$ is called {\em active} \wrt $(f,d)$ if $e_f(v)>0$ and $d(v)<d^{\infty}$.
%
%
%
%
%
%
\subsection{New Region Discharge}\label{sec:ARDalg}
In this subsection, reachability relations ``$\rightarrow$'', ``$\nrightarrow$'', residual paths, and labeling validity will be understood in the region network $G^R$ or its residual $G^R_f$.
\par
The new \Discharge operation, called Augmented Path Region Discharge (\ref{ARD}), works as follows. 
It first pushes excess to the sink along augmenting paths inside the network $G^R$. When it is no longer possible, it continues to augment paths to nodes in the region boundary, $B^R$, in the order of their increasing labels. This is represented by the sequence of nested sets $T_0 = \{t\}$, $T_1 = \{t\}\cup\{v\in B^R \mid d(v)=0\}$, \dots, $T_{d^{\infty}} = \{t\}\cup \{v\in B^R \mid d(v)<d^{\infty}\}$. Set $T_k$ is the destination of augmentations in stage $k$. As we prove below, in stage $k>0$ residual paths may exist only to the set $T_k \backslash T_{k-1} = \{v\mid d(v)=k-1\}$.
\par
%The sets in the nested sequence $(T_k\mid k)$ include the sink and region boundary nodes in the order of their increasing labels. These sets are not modified by the algorithm. Let us verify that augmentations in stage $k$ do not conflict with augmentations in earlier stages.
\begin{procedure}[H]
\SetKwFunction{Augment}{Augment}
\SetKwFunction{ARD}{ARD}
\SetKwFunction{RegionRelabel}{RegionRelabel}
%\SetKwFor{Sub}{Sub}{}{}
%\DontPrintSemicolon
\tcc{assume $d\colon V^R\to \{0,\dots, d^{\infty}\}$ valid in $G^R$}
%{\bf Input:} Region network $G^R$, labeling $d|_{B^R}$. \;
%\Sub(\tcc*[f]{assume $d\colon V^R\to \bbN_0$}){ $(f,d) = \ARD(G^R,d)$}{
%$f:=$  preflow saturating $(\{s\},R)$ edges. \tcc*{init}
\For(\tcc*[f]{stage $k$}){$k=0,1,\dots, d^{\infty}$}{%{$k$ = sorted by ascent set $\{0\}\cup\{d(v)+1 \mid v \in B^R\}$}{
	$T_k = \{t\}\cup \{v\in B^R \mid d(v) < k\}$\;
	$\Augment(R,T_k)$\;
%	\tcc{Augment$(R,T_k)$}
%		\While{$\exists$ a residual path $(v_0\in R,\dots, v_l\in T_k)$, $e_f(v_0)>0$}{
%			augment $\Delta = \min (e_f(v_0), \min\limits_{i} c_f(v_{i-1},v_i))$ along the path.
%			\label{ARD:inner}
%	}
}
\tcc{Region-relabel}
$
%\begin{array}{lr}
d(u) := \begin{cases}
\min \{k \mid u\rightarrow T_k \} & \ u\in R, u\rightarrow T_{d^{\infty}}, \\
d^{\infty} & \ u\in R, u\nrightarrow T_{d^{\infty}}, \\
d(u) & \ u\in B^R.
\end{cases}
%& \phantom{wwwwwww}
%\end{array}
$%
\DontPrintSemicolon
\label{ARD:region-relabel}\;
%}
%\Return preflow $f$ and labeling $d':=\RegionRelabel(G^R_f,d|_{B^R})$ \;
%\Return $(f,d)$
%\setcounter{AlgoLine}{0}
%\vspace{0.5ex}
%\hrule
%\vspace{0.5ex}
%\SetKwFor{Sub}{Procedure}{}{}
%\Sub{ $\Augment{X,Y}$}{
%\While{there exist a path $(v_0,v_1\dots v_l)$, $c_f(v_{i-1},v_i)>0$, $e_f(v_0)>0$, $v_0\in X$, $v_l\in Y$}{
%	augment $\Delta = \min (e_f(v_0), \min\limits_{i} c_f(v_{i-1},v_i))$ units along the path.
%}
%}
%\setcounter{AlgoLine}{0}
%\Sub{ $\RegionRelabel(G^R_f,d|_{B^R})$}{
%$
%\begin{array}{lr}
%d(u) := \begin{cases}
%\min \{k \mid u\rightarrow T_k \} & \ u\in R, u\rightarrow \{t\}\cup B^R, \\
%|\B| & \ u\in R, u\nrightarrow \{t\}\cup B^R, \\
%d(u) & \ u\in B^R.
%\end{cases}
%& \phantom{wwwwwww}
%\end{array}
%$ \;
%%\For{$v\in R$}{
%%\lIf{$v\rightarrow t$ in $G^R_f$}{$d(v) := 0$}
%%\lElse{
%%$d(v) := \min\big(|\B|,\min \{d(w)+1 \mid w\in B^R, v\rightarrow w \mbox{ in $G^R_f$ }\}\big)$
%%}
%%}
%\Return $d$
%}\SetAlgoProcName{Procedure $(f,d)=$}{ARD}
%\caption{Augmented Path Region Discharge (ARD)}\label{ARD}
\vspace{0.5ex}
\caption{ARD($G^R$,$d$)\label{ARD}}
\end{procedure}
\vspace{-2\myskip}\vspace{-2pt}
\begin{procedure}[H]
\While{there exist a path $(v_0,v_1,\dots, v_l)$, $c_f(v_{i-1},v_i)>0$, $e_f(v_0)>0$, $v_0\in X$, $v_l\in Y$}{
	augment $\Delta = \min (e_f(v_0), \min\limits_{i} c_f(v_{i-1},v_i))$ units along the path.
}
\caption{Augment($X$,$Y$)\label{Augment}}
\end{procedure}

The labels on the boundary, $d|_{B^R}$, remain fixed during the algorithm and the labels $d|_{R}$ inside the region do not participate in augmentations and therefore are updated only in the end.
\par
%
%\subsection{Correctness}\label{sec:ARDcorrectness}
%
We claim that \ref{ARD} terminates with no active nodes inside the region, preserves validity and monotonicity of the labeling, and pushes flow from higher labels to lower labels \wrt the new labeling. These properties will be required to prove finite termination and correctness of S-ARD. Before we prove them (Statement~\ref{ARD properties}) we need the following intermediate results: %They are shown in the following steps: %The analysis is organized in the following steps:
\begin{itemize}
\item Properties of the network $G^R_f$ maintained by the algorithm (Statement~\ref{ARD property}, Corollaries~\ref{C1} and~\ref{C2}).
\item Properties of valid labellings in the network $G^R_f$ (Statement~\ref{labeling property}).
\item Properties of the labeling constructed by region-relabel (line~\ref{ARD:region-relabel} of \ARD) in the network $G^R_f$ (Statement~\ref{Relabel properties}).
%\item The claim (Statement~\ref{ARD properties}).
\end{itemize}
%We will need the following lemma.
\begin{lemma}\label{ARD lemma}
Let $X,Y\subset V^R$, $X\cap Y=\emptyset$, $X\nrightarrow Y$. Then $X\nrightarrow Y$ is preserved after i) augmenting a path $(x,\dots,v)$ with $x\in X$ and $v\in V^R$; ii) augmenting a path $(v,\dots,y)$ with $y\in Y$ and $v\in V^R$.
\end{lemma}
\begin{proof}
(See Figure~\ref{fig:distance-ARD}(b)).
Let $\mathcal{X}$ be the set of vertices reachable from $X$. Let $\mathcal{Y}$ be the set of vertices from which $Y$ is reachable. Clearly $\mathcal{X} \cap \mathcal{Y} = \emptyset$, otherwise $X\rightarrow Y$. We have that $(\mathcal{X}, \bar{\mathcal{X}})$ is a cut separating $X$ and $Y$ and having all edge capacities zero.
Any residual path starting in $X$ or ending in $Y$ cannot cross the cut its augmentation change the edges of the cut. Hence, $X$ and $Y$ will stay separated.
\qed
\end{proof}
\begin{figure}[!ht]
\centering
\begin{tabular}{ccc}
\begin{tabular}{c}
\includegraphics[width=0.35\linewidth]{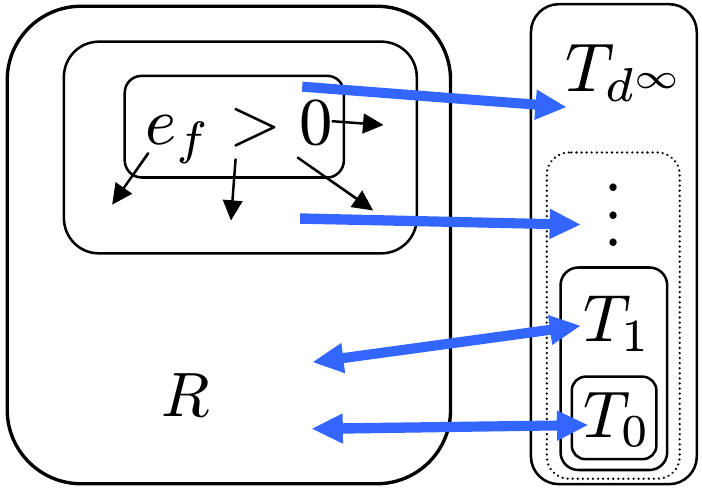}
\end{tabular} & \parbox{0.1\linewidth}{\ } &
\begin{tabular}{c}
\includegraphics[width=0.25\linewidth]{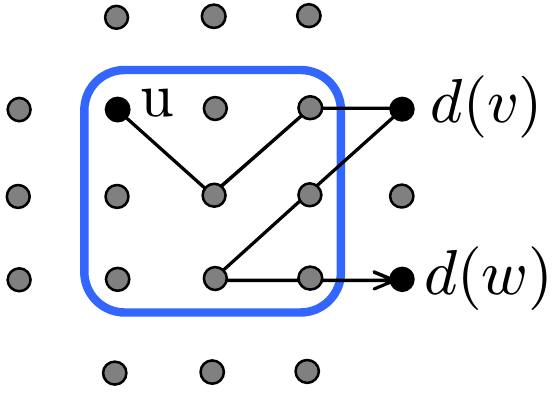}
\end{tabular}\\
(a) & & (b)
\end{tabular}
\caption{(a) Reachability relations in the network $G^R_f$ at the end of stage 1 of ARD: $\{v\mid e_f(v)>0\}\nrightarrow T_{1}$;  $T_{d^{\infty}}\backslash T_{1} \nrightarrow R$. (b) Example of a path which may exist in the network $G^R_f$, by Corollary~\ref{C2} it must be $d(v)\leq d(w)$.}
\label{fig:ARD}
\end{figure}
\begin{statement}\label{ARD property}
Let $v\in V^R$ and $v\nrightarrow T_a$ in $G_f$ in the beginning of stage $k_0$, where $a, k_0\in \{0,1,\dots d^{\infty}\}$. Then $v\nrightarrow T_a$ holds until the end of the algorithm.
\end{statement}
\begin{proof}
We need to show that $v\nrightarrow T_a$ is not affected by augmentations performed by the algorithm. If $k_0 \leq a$, we first prove $v\nrightarrow T_a$ holds during stages $k_0 \leq k\leq a$.
Consider augmentation of a path $(u_0,u_1,\dots, u_l)$, $u_0\in R$, $u_l \in T_k \subset T_a$, $e_f(u_0)>0$. Assume $v\nrightarrow T_a$ before augmentation. By Lemma~1 with $X = \{v\}$, $Y=T_{a}$ (noting that $X\nrightarrow Y$ and the augmenting path ends in $Y$), after the augmentation $v\nrightarrow T_a$. By induction, it holds till the end of stage $a$ and hence in the beginning of stage $a+1$.
\par
We can assume now that $k_0 > a$. Let $A = \{u\in R\mid e_f(u)>0\}$. At the end of stage $k_0-1$ we have $A \nrightarrow T_{k_0-1}\supset T_{a}$ by construction. Consider augmentation in stage $k_0$ on a path $(u_0,u_1\dots,u_l)$, $u_0\in R$, $u_l \in T_{k_0}$, $e_f(u_0)>0$.
By construction, $u_0\in A$. Assume $\{v\}\cup A\nrightarrow T_a$ before augmentation.
Apply Lemma~\ref{ARD lemma} with $X = \{v\}\cup A$, $Y=T_{a}$ (we have $X\nrightarrow Y$ and $u_0\in A\subset X$). After augmentation it is $X \nrightarrow T_a$. By induction, $X \nrightarrow T_a$ till the end of stage $k_0$. By induction on stages, $v\nrightarrow T_a$ until the end of the algorithm.
\qed
\end{proof}
\begin{corollary}\label{C1}If $w\in B^R$ then $w\nrightarrow T_{d(w)}$ throughout the algorithm.
\end{corollary}
\begin{proof}
At initialization, it is fulfilled by construction of $G^R$ due to $c^R(B^R,R)=0$. It holds then during the algorithm by Statement~\ref{ARD property}.\qed
\end{proof}
In particular, we have $B^R\nrightarrow t$ during the algorithm.
%Corollary~\ref{C1} shows that a residual path from $u\in R$ to $t$ in $G_R$ cannot go through boundary $B^R$. 
\begin{corollary}\label{C2}
Let $(u,v_1\dots v_l,w)$ be a residual path in $G^R_f$ from $u\in R$ to $w\in B^R$ and let $v_r\in B^R$ for some $r$. Then $d(v_r)\leq d(w)$.
\end{corollary}
\begin{proof}
%Let $a = d(v_r)$. By the same argument as in Corollary~\ref{C1} we have $v_r \nrightarrow T_a$.
We have $v_r \nrightarrow T_{v_r}$. Suppose $d(w) < d(v_r)$, then $w\in T_{v_r}$ and because $v_r\rightarrow w$ it is $v_r \rightarrow T_{v_r}$ which is a contradiction.
\qed
\end{proof}
The properties of the network $G^R_f$ established by Statement~\ref{ARD property} and Corollary~\ref{C2} are illustrated in Figure~\ref{fig:ARD}.
%
%\par
%{\bf Properties of valid labellings in $G^R_f$}.
\begin{statement}
\label{labeling property}
Let $d$ be a valid labeling, $d(u) \geq 1$, $u\in R$. Then $u\nrightarrow T_{d(u)-1}$. 
\end{statement}
\begin{proof}
%Let $d(u)=1$. 
Suppose $u\rightarrow T_0$. Then there exist a residual path $(u,v_1\dots v_l, t)$, $v_i\in R$ (by Corollary~\ref{C1} it cannot happen that $v_i\in B^R$). By validity of $d$ we have $d(u)\leq d(v_1)\leq\dots \leq d(v_l)\leq d(t) = 0$, which is a contradiction.
\par
%Let $d(u)>1$. 
Suppose $d(u)>1$ and $u\rightarrow T_{d(u)-1}$. Because $u\nrightarrow T_0$, it must be that 
$u\rightarrow w$, $w\in B^R$ and $d(w)<d(u)-1$. Let $(v_0\dots v_l)$ be a residual path with $v_0=u$ and $v_l=w$. Let $r$ be the minimal number such that $v_r\in B^R$. 
By validity of $d$ we have $d(u)\leq d(v_1)\leq\dots \leq d(v_{r-1})\leq d(v_r)+1$.
%Suppose $d(v_r)>d(w)$. Because $v_r\rightarrow w$ we have $v_r \in T_{d(w)}$.
By corollary~\ref{C2} we have $d(v_r)\leq d(w)$, hence $d(u) \leq d(w)+1$ which is a contradiction.
\qed
\end{proof}
%
%\par
%{\bf Properties of region-relabel}.
\begin{statement}%[Properties of \RegionRelabel]
\label{Relabel properties}
For $d$ computed on line~\ref{ARD:region-relabel} and any $u\in R$ it holds:
%For $d = \RegionRelabel(G^R,d|_{B^R})$, for all $u\in R$ holds:
\begin{enumerate}
\item $d$ is valid;
\item $u\nrightarrow T_a \Leftrightarrow d(u)\geq a+1$.
%\item Let $\tilde G^R$ be the residual network on some step of ARD and let $\tilde d$ be a valid labeling in $\tilde G$. Then $d'\geq \tilde d$.
\end{enumerate}
\end{statement}
\begin{proof}
\begin{enumerate}
\item Let $(u,v)\in E^R$ and $c(u,v)>0$. Clearly $u\rightarrow v$. Consider four cases:
\begin{itemize}
\item case $u\in R$, $v\in B^R$: %If $d(u) = 0$ then $d(u) \leq d(v)$.
%Let $d(u) = a > 0$. 
Then $u \rightarrow T_{d(v)+1}$, hence $d(u)\leq d(v)+1$.
\item case $u\in R$, $v \in R$: If $v\nrightarrow T_{d^{\infty}}$ then $d(v)= d^{\infty}$ and $d(u)\leq d(v)$. If $v\rightarrow T_{d^{\infty}}$, then $d(v) = \min\{k\mid v\rightarrow T_k\}$.
Let $k=d(v)$, then $v\rightarrow T_k$ and $u\rightarrow T_k$, therefore $d(u)\leq k = d(v)$.
\item case $u\in B^R$, $v \in R$: By Corollary~\ref{C1}, $u\nrightarrow T_{d(u)}$. Because $u\rightarrow v$, it is $v\nrightarrow T_{d(u)}$, therefore $d(v)\geq d(u)+1$ and $d(u)\leq d(v)-1 \leq d(v)+1$.
\item case when $u=t$ or $v=t$ is trivial.
\end{itemize}
%then $(u,v)$ is a path to the boundary and by definition of {\sf \RegionRelabel}, $d(u) \leq d(v)+1$. 
%If $v\in R$ we have
%\begin{align}
%d(u) \leq \min\{ d(w)+1 \mid w\in B^R,  u\rightarrow w\} \\
%\leq \min_{v' \in R \mid c(u,v')>0} \min \{ d(w)+1 \mid w\in B^R,  v'\rightarrow w\}\\
%= \min_{v'\in R \mid c(u,v')>0}d(v') \leq d(v).
%\end{align}
%
\item[2.] The ``$\Leftarrow$'' direction follows by Statement~\ref{labeling property} applied to $d$, which is a valid labeling.
The ``$\Rightarrow$'' direction: we have $u\nrightarrow T_a$ and $d(u) \geq \min\{k \mid u\rightarrow T_k\} = \min\{k>a \mid u\rightarrow T_k\} \geq a+1$. %, $d(u) > a$. The %Using $u\nrightarrow T_a$ in the definition of \RegionRelabel, we have $d(u) = 1+\min\{d(w) \mid w\in B^R\backslash T_{a},\ u\rightarrow w\} \geq a+1$.
%Let $a=0$. Assume $u\rightarrow t$. Then there exist a residual path $(u,v_1\dots v_l, t)$, $v_i\in R$. By validity of $d$, $d(u)\leq d(v_1)\leq\dots \leq d(v_l)\leq d(t) = 0$, which contradicts to $d(u)\geq a+1$.
%
%\par
%Let $a>0$. Assume $u\rightarrow w$, $w\in B^R$, $d(w)<a$. Let $(v_0\dots v_l)$ be a residual path with $v_0=u$ and $v_l=w$. Let $r$ be the minimal number such that $v_r\in B^R$. 
%By validity of $d$, $d(u)\leq d(v_1)\leq\dots \leq d(v_{r-1})\leq d(v_r)+1$. Suppose $d(v_r)>d(w)$. Because $v_0\rightarrow v_r$ we have $v_r \in T_{d(w)}$.
%By corollary~\ref{C2}, $d(v_r)\leq d(w)$, so $d(u) \leq d(w)+1 < a+1$ which is a contradiction.
\qed
\end{enumerate}
\end{proof}
\begin{statement}[Properties of ARD]\label{ARD properties} Let $d$ be a valid labeling in $G^R$. The output $(f',d')$ of \ref{ARD} satisfies:
\begin{enumerate}
\item There are no active vertices in $R$ \wrt $(f',d')$;\hfill ({\em optimality})
\item $d' \geq d$, $d'|_{B^R} = d|_{B^R}$; \hfill ({\em labeling monotonicity})
\item $d'$ is valid in $G^{R}_{f'}$; \hfill ({\em labeling validity})
\item %{\em Flow direction property:} $f'$ is the result of path augmentations, where each path is from a vertex $u\in R$ to a vertex $v\in \{t\}\cup B^R$ and it is $d'(u)>d(v)$ if $v\in B^R$.
$f'$ is a sum of path flows, where each path is from a vertex $u\in R$ to a vertex $v\in \{t\}\cup B^R$ and it is $d'(u)>d(v)$ if $v\in B^R$. \hfill ({\em flow direction})\par

%is from a vertex $u\in R$ to a vertex $v\in B^R$ %$f'$ is the sum of path flows, where each path flow is from a vertex $u\in R$ to a vertex $v\in \{t\}\cup B^R$ %is from a vertex $u\in R$ to a vertex $v\in B^R$, 
%and satisfies $d'(u)>d(v)$ if $v\in B^R$.
\end{enumerate}
\end{statement}
\begin{proof}\begin{enumerate}
\item In the last stage, the algorithm augments all paths to $T_{d^{\infty}}$. After this augmentation a vertex $u\in R$ either has excess $0$ or there is no residual path to $T_{d^{\infty}}$ and hence $d'(u)=d^{\infty}$ by construction.
\item For $d(u)=0$, we trivially have $d'(u) \geq d(u)$. Let $d(u) = a+1 > 0$. By Statement~\ref{labeling property}, $u\nrightarrow T_{a}$ in $G^R$ and it holds also in $G^R_f$ by Statement~\ref{ARD property}. From Statement~\ref{Relabel properties}.2, we conclude that $d'(u) \geq a+1 = d(u)$. The equality $d'|_{B^R}=d|_{B^R}$ is by construction.
\item Proven by Statement~\ref{Relabel properties}.1. 
\item Consider a path from $u$ to $v\in B^R$, augmented in stage $k>0$. 
It follows that $k = d(v)+1$. At the beginning of stage $k$ it is $u\nrightarrow T_{k-1}$. By Statement~\ref{ARD property}, this is preserved till the end of the algorithm. By~Statement~\ref{Relabel properties}.2, $d'(u) \geq k = d(v)+1 > d(v)$. % = d'(v)+1$.
\qed
\end{enumerate}
%\qed
\end{proof}
\par\noindent
Algorithm~\ref{alg:sequential} and~\ref{alg:parallel} for \Discharge being \ref{ARD} will be referred to as S-ARD and P-ARD, respectively.
\subsection{Complexity of Sequential Augmented Path Region Discharge}
Statement~\ref{labeling-extension} holds for S-ARD as well, so S-ARD maintains a valid labeling.
%Let us first verify that the labeling in S-ARD is globally valid.
%\begin{statement}For a labeling $d$ valid in $G$ and $(f',d') = \ARD(G^R,d)$, the extension of $d'$ to $V$ defined by $d'|_{\bar R} = d|_{\bar R}$ is valid in $G_{f'}$.
%\end{statement}
%\begin{proof}
%Statement~\ref{Relabel properties} established validity of $d'$ in $G^R_{f'}$. For edges $(u,v)\in(V\backslash R,V\backslash R)$ labeling $d'$ coincides with $d$ and $f'(u,v)=0$. It remains to verify validity on edges $(v,u)\in (B^R,R)$ in the case $c^R_f(v,u)=0$ and $c_f(v,u) > 0$. Because $0 = c^R_f(v,u)= c^R(v,u) - f(v,u) = -f(v,u)$, we have $c_f(v,u)=c(v,u)$. 
%Since $d$ was valid in $G$,  $d(v) \leq d(u)+1$. The new labeling $d'$ satisfies $d'(u) \geq d(u)$ and $d'(v) = d(v)$. It follows that $d'(v) = d(v)\leq d(u)+1 \leq d'(u)+1$. Hence $d'$ is valid in $G_{f'}$.
%\qed
%\end{proof}
%
\begin{theorem}\label{T:S-ARD}
S-ARD terminates in at most $2|\B|^2+1$ sweeps.
\end{theorem}
\begin{proof} The value of $d(v)$ does not exceed $|\B|$ and $d$ is non-decreasing. The total increase of $d|_{\B}$ during the algorithm is at most $|\B|^2$.
\par 
After the first sweep, active vertices are only in $\B$. Indeed, discharging region $R_k$ makes all vertices $v\in R_k$ inactive and only vertices in $\B$ may become active. So by the end of the sweep, all vertices $V\backslash \B$ are inactive.
\par
Let us introduce the quantity
\begin{align}
\Phi = \max \{d(v)\mid v\in\B,\ v \mbox{ is active in $G$ }\}.
\end{align}
We will prove the following two cases for each sweep but the first one: %1) If $d|_\B$ is increased then increase in $\Phi$ is no more than total increase in $d|_\B$; 2) If $d|_\B$ is not increased then $\Phi$ is decreased at least by 1.
\begin{enumerate}
\item
If $d|_\B$ is increased then the increase in $\Phi$ is no more than total increase in $d|_\B$.
%We will show this property for ARD on each $R_k$. 
Consider discharge of $R_k$. Let $\Phi$ be the value before ARD on $R_k$ and $\Phi'$ the value after. Let $\Phi' = d'(v)$. It must be that $v$ is active in $G'$. If $v\notin V^{R_k}$, then $d(v)=d'(v)$ and $e(v)=e_{f'}(v)$ so  $\Phi \geq d(v) = \Phi'$.
\par
Let $v\in V^{R_k}$. After the discharge, vertices in $R_k$ are inactive, so $v\in B^{R_k}$ and it is $d'(v)=d(v)$.
If $v$ was active in $G$ then $\Phi \geq d(v)$ and we have $\Phi'-\Phi \leq d'(v)-d(v)=0$.
If $v$ was not active in $G$, there must exist an augmenting path from a vertex $v_0$ to $v$ such that $v_0\in R_k\cap \B$ was active in $G$. For this path, the flow direction property implies $d'(v_0) \geq d(v)$. We now have $\Phi'-\Phi \leq d'(v)-d(v_0) = d(v)-d(v_0) \leq d'(v_0)-d(v_0) \leq \sum_{v\in R_k \cap \B}[d'(v)-d(v)]$. Summing over all regions, we get the result.
\item 
If $d|_\B$ is not increased then $\Phi$ is decreased at least by 1.
We have $d'=d$. Let us consider the set of vertices having the highest active label or above, $H=\{v\mid d(v)\geq \Phi\}$. These vertices do not receive flow during all discharge operations due to the flow direction property. After the discharge of $R_k$ there are no active vertices left in $R_k \cap H$ (property~\ref{ARD properties}.1). After the full sweep, there are no active vertices in $H$.
\end{enumerate}
\par
In the worst case, starting from sweep 2, $\Phi$ can increase by one $|\B|^2$ times and decrease by one $|\B^2|$ times. In at most $2 |\B|^2+1$ sweeps, there are no active vertices left.
\qed
\end{proof}
On termination we have that the labeling is valid and there are no active vertices in $G$. %The proof that P-ARD terminates is similar and is given in~\cite{TR}.
%\par
%Note that this bound is better than $O(n^2)$ only when $|\B| \ll n$, which is true for sparse graphs if we select the regions properly.
%
\subsection{Complexity of Parallel Augmented Path Region Discharge}\label{sec:ARDparallel}
\begin{statement}[Properties of Parallel ARD]\label{ARDP properties}
Let $d$ be a valid labeling in the beginning of a sweep of P-ARD. Then the pair of fused flow and labeling $(f',d')$ satisfies:
\begin{enumerate}
\item Vertices in $V\backslash \B$ are not active in $G_{f'}$. \hfill ({\em optimality})
\item $d' \geq d$; \hfill ({\em labeling monotony})
\item $d'$ is valid; \hfill ({\em labeling validity})
\item $f'$ is the sum of path flows, where each path is from a vertex $u\in V$ to a vertex $v\in \B$, satisfying $d'(u)\geq d(v)$.\hfill ({\em weak flow direction})
\end{enumerate}
\end{statement}
%
%\begin{statement}
%\begin{enumerate}
%\item $f'$ satisfy flow direction property.
%\end{enumerate}
%\end{statement}
\begin{proof}
\begin{enumerate}
\item For each $k$ there are no active vertices in $R_k$ \wrt $(f'_k, d'_k)$. The fused preflow $f'$ may differ from $f'_k$ only on the boundary edges $(u,v)\in (\B,\B)$. So there are no active vertices in $V\backslash \B$ \wrt $(f',d')$. 

%, and we cancel the flow only on edges $(u,v)$, $u\in R_k$, $v\in R_j$, therefore $u\in \B$.
\item By construction.
\item Same as in P-PRD.
\item Consider the augmentation of a path from  $u\in R_k$ to $v\in B^{R_k}$ during \ARD on $G^{R_k}$ and canceling of the flow on the last edge during the flow fusion step. Let the last edge of the path be $(w,v)$. We need to prove that $d'(u)\geq d(w)$.
Let $\tilde d$ be the labeling in $G^R_k$ right before augmentation, as if it was computed by region-relabel. Because $\tilde d$ is valid it must be that $\tilde d(w)\leq \tilde d(v)+1$.
%
%For \ARD it holds $d'_k(u) > d(v)$. Because $f'_k(v,w)$ $d'$ is valid, $d(v)\leq \tilde d(v)+1
%Right before augmentation we have $w\nrightarrow T_{d(v)}$ in $G^{R_k}$, by properties of \ARD there will hold $d'_k(w)>d(v)$, therefore  
We have $d'_k(u) > d(v) \geq \tilde d(v) \geq \tilde d(w)-1 \geq d(w)-1$. So $d'(u) \geq d(w)$.
%\par
%Resulting $f'$ is a sum of flows, satisfying flow direction property. Some of the path flows got shortened, because we didn't let them go over the boundary, where it would invalidate the labeling. However, shortened paths also satisfy (non-strict) flow direction property. This would suffice for the proof of complexity. In case that no label is risen, all flow is accepted and satisfy strict flow direction.
\end{enumerate}
\end{proof}
%
%
%
%\begin{statement}
%Parallel ARD terminates in $|\B|^2+1$ sweeps.
%\end{statement}
%\begin{proof}
%With the correction to the flow made, we can consider the total change as a sequence of (corrected) ARD updates in some order. For each update it holds that $d'$ is valid, $d'\geq d$, $d'(B) = d(B)$ and weak flow direction property. The excess is still pushed to $R\cap \B$. This is sufficient to bound the increase of $\Phi$ the same way as in Theorem 2. When no label is risen during the sweep, all pairs are valid and $f'$ satisfy strong flow direction, so case 1 of Theorem 2 applies as well.
%\end{proof}
%
%
%
\begin{theorem}
Parallel algorithm~\ref{alg:parallel} with ARD terminates in $2|\B|^2+1$ sweeps.
\end{theorem}
\begin{proof}
%The excess is still pushed to $R\cap \B$. This is sufficient to bound the increase of $\Phi$ the same way as in Theorem 2. When no label is risen during the sweep, all pairs are valid and $f'$ satisfy flow direction, so Point 2 of Theorem 2 applies as well.
\par\noindent
As before, total increase of $d|_{\B}$ is at most $|\B|^2$.
\par\noindent
After the first sweep, active vertices are only in $\B$ by Statement~\ref{ARDP properties}.1.
\par\noindent
%As before,
%\begin{align}
%\Phi = \max \{d(v)\mid v\in\B,\ v \mbox{ is active in $G$ }\}
%\end{align}
For each sweep after the first one:
\begin{itemize}
\item If $d|_{\B}$ is increased then increase in $\Phi$ is no more than the total increase of $d|_{\B}$.
\begin{quote}
Let $\Phi'$ be the value in the network $G'= G_{f'}$. Let $\Phi' = d'(v)$. It must be that $v$ is active in $G'$ and $v\in \B$.
\par
If $v$ was active in $G$ then $\Phi \geq d(v)$ and we have $\Phi'-\Phi \leq d'(v)-d(v)$.
\par
If $v$ was not active in $G$ then there must exist a path flow in $f'$ from a vertex $v_0$ to~$v$ such that $v_0\in \B$ was active in $G$. For this path, the weak flow direction property implies $d'(v_0) \geq d(v)$. We have $\Phi'-\Phi \leq d'(v)-d(v_0)  = d'(v)-d(v)+d(v)-d(v_0) \leq d'(v)-d(v)+d'(v_0)-d(v_0) \leq \sum_{v\in \B}[d'(v)-d(v)]$.
\end{quote}
\item If $d|_{\B}$ is not increased then $\Phi$ is decreased at least by 1. 
\begin{quote}
	In this case, $f'$ satisfies the strong flow direction property and the proof of Theorem~\ref{T:S-ARD} applies.
\end{quote}
\end{itemize}
After total of $2 |\B|^2+1$ sweeps, there are no active vertices left.
\end{proof}
\section{Implementation}\label{sec:implementation}
%There are several important implementation details.
In this section we first discuss heuristics commonly used in the push-relabel framework. They are essential for the practical performance of the algorithms. We then describe our base implementations of S-ARD/S-PRD and the solvers they rely on. In the next section we describe an efficient implementation of ARD, which is more sophisticated but has a much better practical performance.
%give details on the referenced implementations of the augmented path and p, describe the common heuristic, which are essential for the practical performance of the algorithms. 
%first describe common heuristics, which are essential for the practical performance of the algorithms. Our basic implementations and details about referenced implementations are given in Sect.~\ref{sec:ref_impl}-\ref{PRD_basic}. Then we describe efficient 
%without which algorithms would run incomparably longer and further details of our implementation.
\subsection{Heuristics}
{\bf Region-relabel} heuristic computes labels $d|_{R}$ of the region vertices, given the distance estimate on the boundary, $d|_{B^R}$. There is a slight difference between PRD and ARD variants (using distance $d^*$ and $d^{*\B}$, resp.), displayed by the corresponding if conditions.
\par
\begin{algorithm}[H]
%\DontPrintSemicolon
%\item {\bf input:} Region network $G^R$, labeling of the boundary $d(B)$.
$d(t):=0$; $O:=\{t\}$; $d|_{R} := d^\infty$; $d^{\rm current}:=0$\tcc*{init}
\lIf{ARD}{$d|_{B^R}:=d|_{B^R}+1$}\tcc*{(for ARD)}
\tcc{$O$ is a list of open vertices, having the current label $d^{\rm current}$}
$d^{\rm max} := \max \{d(w)\mid w\in B^R,\, d(w)<d^{\infty}\}$\;
\While{$O\neq \emptyset$ or $d^{\rm current}< d^{\rm max}$}{
\tcc{if $O$ is empty raise $d^{\rm current}$ to the next seed}
\lIf{$O=\emptyset$}{$d^{\rm current} := \min \{d(w) \mid w\in B^R,\, d(w)>d^{\rm current},\, d(w)<d^{\infty}\}$}\;
\tcc{add seeds to the open set}
$O := O\cup\{w\in B^R \mid d(w) = d^{\rm current} \}$\;
\tcc{find unlabeled vertices from which $O$ can be reached}
$O := \{u\in R \mid (u,v)\in E^R,\, v\in O,\, c(u,v)>0,\, d(u)=d^\infty\}$\;
\lIf{PRD}{$d^{\rm current}\leftarrow d^{\rm current}+1$}\tcc*{(for PRD)}
$d|_O := d^{\rm current}$\tcc*{label them}
}
\lIf{ARD}{$d|_{B^R}:=d|_{B^R}-1$}\tcc*{(for ARD)}
\caption{Region-relabel$(G^R,d|_{B^R})$}\label{alg:region-relabel}
\end{algorithm}
%
%\begin{myalgorithm}{Region-relabel}
%\label{alg:region-relabel}
%\item {\bf input:} Region network $G^R$, labeling of the boundary $d(B)$.
%\item \verb=//=$O$ \verb=is a list of 'open' nodes, having the current label=
%\item {\bf init: } $d(t)=0$, $O=\{t\}$, $d(R) = \infty$
%%$\bar B$ - ordered by ascend sequence of boundary vertices, 
%${\rm current}=0$
%\item {\bf while } $O\neq \emptyset$ or ${\rm current}< \max d(B)$
%\item \hspace{1cm} \verb=//Raise current to the next seed=
%\item \hspace{1cm} {\bf if} $O=\emptyset$, ${\rm current} = \min \{d(w) \mid w\in B, d(w)>{\rm current}\}$
%\item \hspace{1cm} \verb=//Add seeds to the open set=
%\item \hspace{1cm} $O \leftarrow O\cup\{w\in B \mid d(w) = {\rm current} \}$
%\item \hspace{1cm} \verb=//Find unlabeled vertices from which= $O$ \verb=can be reached=
%\item \hspace{1cm} $O \leftarrow \{u\in R \mid (u,v)\in E^R, v\in O, c(u,v)>0, d(u)=\infty\}$
%\item \hspace{1cm} ${\rm current}\leftarrow {\rm current}+1$\ \ (for PRD)
%\item \hspace{1cm} \verb=//Label them=
%\item \hspace{1cm} $d(O) \leftarrow {\rm current}$
%\item {\bf end}
%\end{myalgorithm}
In the implementation, the set of boundary vertices is sorted in advance, so the algorithm runs in $O(|E^R|+ |V^R|+|B^R|\log |B^R|)$ time and uses $O(|V^R|)$ space. %The Region-relabel for the distance $d^{*\B}$ (in the case of ARD) is the same, just omitting line 7. 
The resulting labeling $d'$ is valid and satisfies $d'\geq d$ for arbitrary valid $d$. %So this algorithm does not break correctness.%For PRD updating the distance by Region-relabel is a heuristic  %Because PRD relabels nodes in the discharge, this region-relabel is optional.%While for ARD this is a necessary component, for PRD it is 
%Note that this heuristic can be implemented as a sequence of relabel operations (Dijkstra's algorithm).
\par
{\bf Global gap heuristic}. Let us briefly explain the global gap heuristic~\cite{Cherkassky-94}. It is a sufficient condition to identify that the sink is unreachable from a set of nodes. Let there be no nodes with label $g>0$: $\forall v\in V$ $d(v)\neq g$, and let $d(u)>g$. For a valid labeling $d$, it follows that there is no node $v$ for which $c(u,v)>0$ and $d(v)<g$. Assuming there is, we will have $d(u) \leq d(v)+1 \leq g$, which is a contradiction. Therefore the sink is unreachable from all nodes $\{u \mid d(u)>g\}$ and their labels may be set to $d^\infty$.
%For region discharge algorithms detecting a gap 
\par
{\bf Region gap heuristic}~\cite{Delong08} detects if there is no nodes inside region $R$ having label $g>0$. Such nodes can be connected to the sink in the whole network only through one of the boundary nodes, so they may be relabeled up to the closest boundary label. Here is the algorithm\footnote{Region-gap-relabel in~\cite[fig.\,10]{Delong08} seems to contain an error: only nodes above the gap should be processed in step 3.}:
\begin{algorithm}[H]
\label{Region-gap}
\tcc{Input: region network $G^R$, labeling $d$, gap $g$: $\forall v\in R$ $d(v)\neq g$}
$d^{\rm next} := \min\{d(w)\mid w\in B^R, d(w)>g \}$\;
\For{$v \in R$ such that $g<d(v)<d^{\rm next}$}{
$d(v) := d^{\rm next}$+1
}
\caption{Region-gap$(G^R,d,g)$}
\end{algorithm}
%says that nodes above the gap may push the flow further only to outside of the region and thus can be lifted at least to the minimal label of the boundary above the gap
If no boundary vertex is above the gap, then $d^{\rm next}=d^\infty$ in step 1 and all nodes above the gap are disconnected from the sink in the network $G$. Interestingly, this sufficient condition does not imply a global gap. In PRD, we detect the gap efficiently after each node relabel operation, by discovering an empty bucket (see below).
%This is somewhat weaker than standard global gap heuristic. However it only need to track counts of how many region verticies are assigned a particular label and can be applied with little overhead after each push operation.
%
%
\subsection{Referenced Implementations}\label{sec:ref_impl}
%When it comes to measuring real time, the details of the actual implementation are essential.
\mypar{Boykov-Kolmogorov (BK)} The reference augmenting path implementation~\cite{BK-maxflow}. There is only a trivial $O(m n^2|C|)$ complexity bound known for this algorithm\footnote{The worst-case complexity of breads-first search shortest path augmentation algorithm is just $O(m |C|)$. Tree adaptation step, introduced in~\cite{BK-maxflow} to speed-up the search, does not have a good bound and introduces additional $O(n^2)$ factor.}, where $C$ is the minimum cut value.
%, which proved the fastest for typical vision problems. We will use it as a core in our ARD algorithm. 
%BK augmenting path algorithm is the most widely used in CV. 
%Note that [BK] seems to misspecify the complexity of the algorithm, which should be $O(m |C|)$, where $m$ is number of edges and $C$ is value of the minimum cut (alternatively, $O(m^2 |c|)$, where $c$ is maximal capacity of arcs in some finite cut.).
%
\mypar{Highest level Push-Relabel (HIPR)} is reference push-relabel implementation~\cite{Cherkassky-94} (v3.6, \url{http://www.avglab.com/andrew/soft.html}). This implementation has two stages: finding the maximum preflow / minimum cut and upgrading the maximum preflow to a maximum flow. Only the first stage was executed and benchmarked. We tested two variants with frequency of the global relabel heuristic equal to 0.5 (the default value in HIPR v3.6) and equal to 0. These variants will be denoted HIPR$0.5$ and HIPR$0$ respectively. HIPR$0$ executes only one global update at the beginning. Global updates are essential for difficult problems. However, HIPR0 was always faster than HIPR$0.5$ in our experiments with real test instances. This is a good indication that global relabel and region-relabel heuristics are not essential for PRD as well. The worst case complexity is $O(n^2\sqrt{m})$.
\subsection{S/P-ARD implementation}\label{ARD_basic}
%This is ``streaming'' implementation of Sequential ARD, which 
The basic implementation of \ARD simply invokes BK solver as follows. On stage $0$ we compute the maximum flow on the $G^R$ by BK, augmenting pathes from source to the sink. On the stage $k$, infinite capacities are added from the boundary nodes having label $k-1$ to the sink, using the possibility of dynamic changes in BK. The flow augmentation to the sink is then continued, reusing the search trees. The relabel procedure is implemented as Alg.~\ref{alg:region-relabel}. In the beginning of next discharge we clear the infinite link from boundary to the sink and continue as above. Some part of the sink search tree, linked through the boundary nodes, get destroyed, but the larger part of it and the source search tree are reused. A more efficient implementation is described in Sect~\ref{sec:ARD_efficient}. It includes additional heuristics and maintenance of separate boundary search trees.
\par
{\bf S-ARD}. In the streaming mode we keep only one region in the memory at a time. After a region is processed by \ARD all the internal data structures have to be saved to the disk and cleared from memory until the region is discharged next time. We manage this by allocating all the region's data into a fixed page in the memory, which can be saved and loaded, preserving the pointers. By doing the load/unload manually (rather than relying on the system swapping mechanism), we can accurately measure the pure time needed for computation (CPU) and the amount of disk I/O. We also can use 32bit pointers with larger problems.
\par
When the sequential algorithm~\ref{alg:sequential} terminates, it has found an optimal preflow. However, we still need to do some extra work to find an optimal cut. On termination the labeling $d$ is only a lower bound on the distance to the sink. Therefore if $d(v)=d^{\infty}$ we are sure that $v\nrightarrow t$ in $G$ and hence $v$ must be in the source set, but if $d(v)<d^{\infty}$ it is still possible that $v\nrightarrow t$ in $G$. Therefore we execute several extra sweeps, performing only region-relabel and gap heuristic, until labels stop changing. In practice it takes from 0 to 2 extra sweeps.
\par
A region with no active vertices is skipped. The global gap heuristic is executed after each region discharge. Because it is based on labels of boundary nodes only, it is sufficient to maintain a label histogram with $|\B|$ bins to implement it. S-ARD uses $O(|\B| + |(\B,\B)|)$ ``shared'' memory and $O(|V^R+E^R|)$ ``region'' memory, to which regions are loaded one at a time.
%\par
%A more efficient implementation of ARD is possible, maintaining search trees for all boundary nodes. The Relabel procedure need not be computed separately by alg.~\ref{alg:region-relabel}, because the required reachability information is efficiently computed during the growth of the search trees in BK.
%When we need to go to the next region we save the whole page of memory associated to the current region (keeping all the data structures of the algorithm) so the next time this region is considered it 
%\mypar{Splitter}
\par
To solve large problems, which do not fit in the memory, we have to create region graphs without ever loading the full problem. We implemented a tool called {\bf splitter}, which reads the problem from a file and writes edges corresponding to the same region to the region's separate ``part'' file. Only the boundary edges (linking different regions) are withheld to the memory.
\par
{\bf P-ARD}. We implemented this algorithm for a shared-memory system using OpenMP language extension. All regions are kept in memory, the discharges are executed concurrently in separate threads, while the gap heuristic and messages exchange are executed synchronously by the master thread.
\subsection{S/P-PRD implementation}\label{PRD_impl}
To solve region discharge subproblems in PRD in the highest label first fashion we designed a special reimplementation of HIPR, which will be denoted {\bf HPR}. We intended to use the original HIPR implementation to make sure that PRD relies on the state-of-the art core solver. It was not possible directly. 
A subproblem in PRD is given by a region network with fixed distance labels on the boundary (let us call them {\em seeds}). Distance labels in PRD may go up to $n$ in the worst case. The same applies to region subproblems as well. Therefore, keeping an array of buckets corresponding to possible labels (like in HIPR), would not be efficient. It would require $O(|V|)$ memory and an increased complexity. However, because a region has only $|V^R|$ nodes, there are no more than $|V^R|$ distinct labels at any time. This allows to keep buckets as a doubly-linked list with at most $|V^R|$ entries. Highest label selection rule and the region-gap heuristic can then be implemented efficiently with just a small overhead. We tried to keep other details similar to HIPR (current arc data structure, etc.). HPR with arbitrary seeds has the worst case complexity $O(|V^R|^2 \sqrt{|E^R|})$ and uses $O(|V^R|+|V^E|)$ space.
%It was made as close as possible to GT0, but with two important modifications. 
%Note that the distance labels inside the region can go well above $|\V^R|$ (the number of nodes in the region), in theory up to $|\V|$ -- the total number of nodes in the whole problem. This is resolved by keeping label buckets in a linked list rather than in a linear array. The algorithm has $O(|\E^R|+|\V^R|)$ memory complexity and same time complexity as GT.
%Region-relabel and region-gap heuristics account for the fact that distance labels may form several independent chains, starting at the seeds. 
When the whole problem is taken as a single region then HPR should be equivalent to HIPR0. Though the running time on the real instances can be somewhat different.
 %Current implementation still uses $O(|V|)$ memory to. 
%This tool still uses additional $O(|V|)$ memory to accomplish its job.

\mypar{S-PRD} Our reimplementation of Delong and Boykov~\cite{Delong08} for an arbitrary graph and a fixed partition, using HPR as a core solver. It uses the same memory model, paging mechanism and the splitter tool as S-ARD. The region discharge is always warm-started. %For each discharge the region-relabel is called only in the first sweep or 
%Region-relabel is called in the beginning of a region discharge only in two cases: in the first sweep and when a global gap was discovered. 
We found it inefficient to run the region-relabel after every discharge. In the current experiments, we run it once at the beginning and then only when a global gap is discovered. To detect a global gap, we keep a histogram of all labels, $O(n)$ memory, and update it after each region discharge (in $O(|V^R|)$ time). In practice, this $O(n)$ memory is not a serious limitation -- labels are usually well below $n$. If they are not then we should consider a weaker gap heuristic with a smaller number of bins. Applying the gap (raising the corresponding nodes to $d^\infty$) for all regions is delayed until they are loaded. So we keep the track of the best global gap detected for every region. Similar to how the sequential algorithm~\ref{alg:sequential} represents both S-ARD and P-ARD, it constitutes a piece of generic code in our implementation, where the respective discharge procedure and gap heuristics are plugged.
\mypar{P-PRD} Implementation of parallel \PRD for shared-memory system with OpenMP.
%and we can still run a weaker global gap heuristic if they are not, considering a smaller number of bins.
%Noting that inside the region ver
%: we keep distance buckets in a linked list rather than in a linear array and 
%We could not use GT as a region discharge solver in PRD directly. One reason, is that gap relabel heuristic need to be changed.
%for the following reason: if the region is 
\section{Efficient Implementation of ARD}\label{sec:ARD_efficient}
The basic implementation of S-ARD, as described in the previous section, worked reasonably fast (comparable to BK) on simple problems like 2D stereo and 2D random segmentation (Sec.~\ref{sec:exp_rand}). However, on some 3D problems the performance was unexpectedly bad. For example, to solve \verb=LB07-bunny-lrg= instance (Sec.~\ref{sec:exp_streaming}) the basic implementation required 32 minutes of CPU time. In this section we describe an efficient implementation which is more robust and is comparable in speed with BK on all tested instances. In particular, to solve \verb=LB07-bunny-lrg= it takes only 15 seconds of CPU time. The problem why the basic implementation is so slow is in the nature of the algorithm: sometimes it has to augment the flow to the boundary, without knowing of whether it is a useful work or not. If the particular boundary was selected wrongly the work is wasted. This happens in \verb=LB07-bunny-lrg= instance, where the data seeds are sparse. The huge work is performed on pushing the flow around in the first few iterations, before a reasonable labeling is established. We introduce two heuristics how to overcome this problem: the boundary-relabel heuristic and partial discharges. An additional speed-up is obtained by dynamically maintaining boundary search trees and the current labeling.
\subsection{Boundary Relabel Heuristic}
We would like to have a better distance estimate, but we cannot run a global relabel because implementing it in a distributed fashion would take several full sweeps, which would be too wasteful. Instead, we go for the following cheaper lower bound.
%We know that region-relabel becomes the exact distance in the case of a single region. %We would like to find a heuristic, 
Our implementation keeps all the boundary edges (including their flow and distance labels of the adjacent vertices) in the shared memory. Fig.~\ref{fig:bnd_relabel}(a) illustrates this boundary information. We want to improve the labels by analyzing only this boundary part of the graph, not looking inside the regions. Since we don't know how the vertices are connected inside the regions, we have to assume that every boundary vertex might be connected to any other within the region, except of the following case. If $u$ and $v$ are in the same region $R$ and $d(u)>d(v)$ then we know for sure that $u\nrightarrow v$ in $R$. It follows from validity of labeling $d$. We now can calculate a lower bound on the distance $d^{*\B}$ in $G$ assuming that all the rest of the nodes are potentially connected within the regions.
\begin{figure}[ht]
\begin{center}
\setlength{\tabcolsep}{0pt}
%\begin{tabular}{ccc}
\begin{tabular}{c}\includegraphics[width=0.25\linewidth]{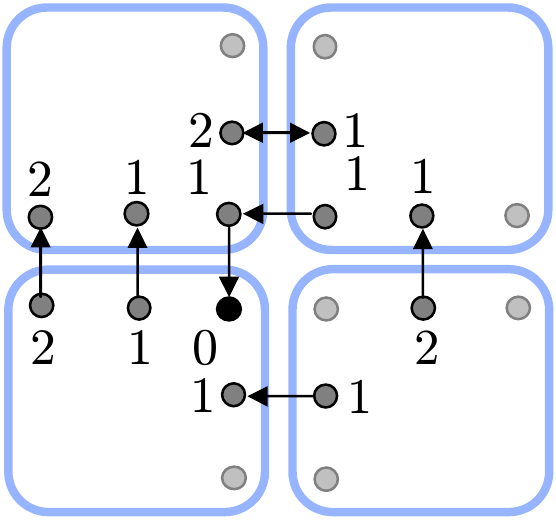}\\
(a)
\end{tabular}\hspace{1cm} %&
\begin{tabular}{c}\includegraphics[width=0.25\linewidth]{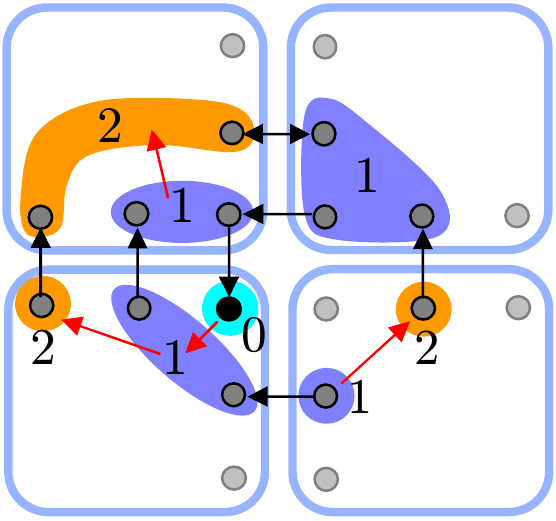}\\
(b)
\end{tabular}\hspace{1cm} %&
\begin{tabular}{c}\includegraphics[width=0.25\linewidth]{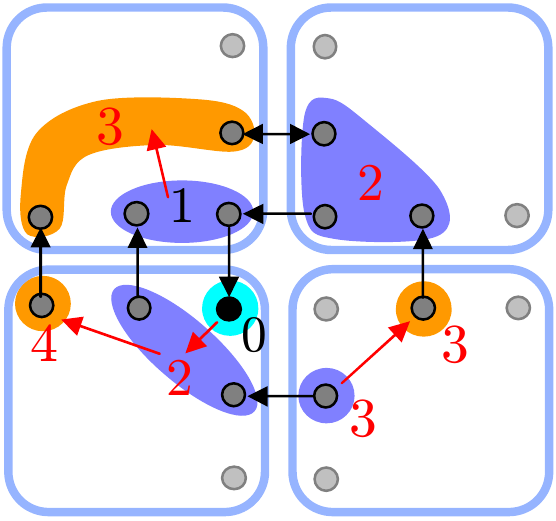}\\
(c)
\end{tabular}
%\end{tabular}
\end{center}
\caption{Boundary relabel heuristic: (a) Boundary vertices of the network and a valid labeling. Directed arcs correspond to non-zero residual capacities. Vertices without numbers have label $d^{\infty}$ and do not participate in the construction. (b) Vertices having the same label are grouped together within each region and arcs of zero length (of red color) are added from a group to the next label's group. It is guaranteed that \eg vertices with label 1 are not reachable from vertices with label 2 within the region, hence there is no arc $2{\to}1$. Black arcs have the unit length. (c) The distance in the auxiliary graph is a valid labeling and a lower bound on the distance in the original network.}
\label{fig:bnd_relabel}
\end{figure}
\par
If $d(v) = d(u)$ we have to assume that $v\rightarrow u$ and $u\rightarrow v$ in $R$, therefore the new lower bound for $u$ and $v$ will coincide. Hence we group vertices having the same label within a region together as shown in Fig.~\ref{fig:bnd_relabel}(b). In the case $d(v) < d(u)$ we know that $u\nrightarrow v$ but have to assume $v\rightarrow u$ in $R$. We thus add a directed arc of length zero from the group of $v$ to the group of $u$ (Fig.~\ref{fig:bnd_relabel}(b)). Let $d_1<d_2<d_3$ be labels of groups within one region. There is no need to create an arc $d_1{\to}d_3$, because two arcs $d_1{\to}d_2$ and $d_2{\to}d_3$ of length zero are an equivalent representation. Therefore it is sufficient to connect only groups having consequent labels. Let us denote thus constructed graph $\bar G$. We can calculate the distance to nodes with label 0 in $\bar G$ by running Dijkstra's algorithm in $\bar G$. Let this distance be denoted $d'$. We then update the labels as
\begin{equation}
d(u) := \max\{d(u), d'(u)\}.
\end{equation}
We have to prove the following two points:
\begin{enumerate}
\item $d'$ is a valid labeling;
\item If $d$ and $d'$ are valid labellings, then $\max(d,d')$ is a valid.
\end{enumerate}
\begin{proof}
\begin{enumerate}
\item Let $c(u,v)>0$. Let $u$ and $v$ be in the same region. It must be that $d(u)\leq d(v)$. Therefore either $u$ and $v$ are in the same group or there is an arc of length zero from group of $u$ to group of $v$. In any case it must be $d'(u)\leq d'(v)$. If $u$ and $v$ are in different regions, there is an arc of length 1 from group of $u$ to group of $v$ and therefore $d'(u)\leq d'(v)+1$.
\item Let $l(u,v) = 1$ if $(u,v)\in (\B,\B)$ and $l(u,v)=0$ otherwise. We have to prove that if $c(u,v)>0$ then
\begin{equation}\label{eq:bnd_relabel 1}
\max\{d(u),d'(u)\} \leq \max\{d(v),d'(v)\} + l(u,v).
\end{equation}
Let $\max\{d(u),d'(u)\} = d(u)$. From validity of $d$ we have $d(u) \leq d(v) + l(u,v)$.
If $d(v)\geq d'(v)$, then $\max\{d(v),d'(v)\} = d(v)$ and~\eqref{eq:bnd_relabel 1} holds.
If $d(v) < d'(v)$, then $d(u) \leq d(v) + l(u,v) < d'(v)+l(u,v)$ and~\eqref{eq:bnd_relabel 1} holds again.
\end{enumerate}
\end{proof}
The complexity of this algorithm is $O(|(\B,\B)|)$. It is relatively inexpensive and can be run after each sweep.
%but because we don't know anything of arcs inside the regions we must assume that every vertex is reachable from any other inside the region
%
\subsection{Partial Discharges}
Another heuristic which proved very efficient was simply to postpone path augmentations to higher boundary vertices to further sweeps. In combination with boundary-relabel this allows to save a lot of unnecessary work. More precisely, on sweep $s$ the algorithm~\ARD is allowed to execute only stages up to $s$. This way, in sweep $0$ only paths to the sink are augmented and not any path to the boundary. Nodes which cannot reach the sink (but can potentially reach the boundary) get label 1. These initial labels may already be improved by boundary-relabel. In sweep $1$ paths to the boundary with label $0$ are allowed to be augmented and so on.
\subsection{Boundary Search Trees}\label{sec:ARD_efficient3}
\begin{figure}[ht]
\begin{center}
\setlength{\tabcolsep}{0pt}
%\begin{tabular}{ccc}
\hfill (a) \hfill \ \ \ (b) \hfill \ \ \ (c) \hfill \ \hfill \\ 
\begin{tabular}{c}\includegraphics[width=0.75\linewidth]{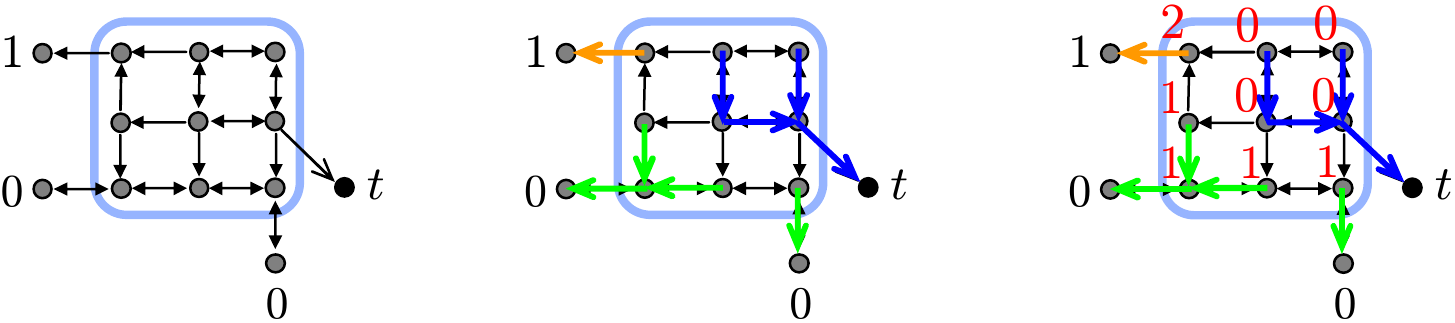}
\end{tabular}
%\end{tabular}
\end{center}
\caption{Search trees. (a) A region with some residual arcs. The region has only 3 boundary nodes, for simplicity. (b) Search trees of the sink and boundary nodes: vertex $v$ is in the tree with the lowest possible label of the root. The sink is assigned a special label $-1$. The source search tree is empty in this example. (c) Labels of the inner vertices are determined as the tree's root label$+1$.}
\label{fig:ARD_trees}
\end{figure}
We now redesign the algorithm so that not only the sink and source search trees are maintained but also the search trees of boundary nodes. This allows to save computation when the labeling of many boundary nodes remains constant during the consequent sweeps, but only a small fraction is changed. Additionally, knowing the search tree for each inner vertex of the region determines its actual label, so the region-relabel procedure becomes obsolete. The design of the search tree data structures, their updates and other detail are the same as in~\cite{Kolmogorov-04-PhD}, only few changes to the implementation are necessary. For each vertex $v\in R$ we introduce a mark $\tilde d(v)$ which corresponds to the root label of its tree or is set to a special {\em free} mark if $v$ is not in any tree. For each tree we keep a list of open vertices (called active in~\cite{Kolmogorov-04-PhD}). A vertex is {\em open} if it is not blocked by the vertices of the trees with the same or lower root label. The trees may grow only at the open vertices.
\par
Figure~\ref{fig:ARD_trees} shows the correspondence between search trees and the labels. The sink search tree is assigned label $-1$. In the stage $0$ of \ARD we grow the sink tree and augment all found paths if it touches the source search tree. Vertices, which are added to the sink tree are marked with label $\tilde d = -1$. In stage $k+1$ of \ARD we grow trees with root at a boundary vertices $w$ with label $d(w)=k$, all vertices added to the tree are marked with $\tilde d = k$. When the tree touches the source search tree, the found path is augmented. If the tree touches a vertex $u$ with label $\tilde d(u) < k$, it means that $u$ is already in the search tree with a lower root and no action is taken. It cannot happen that during growth of a search tree with root label $k$ a vertex is reached with label $\tilde d>k$, this would contradict to properties of \ARD. The actual label of a vertex $v$ at any time is determined as $\tilde d(v)+1$ if $v\in R$ and $d(v)$ if $v\in B^R$.
\par
Let us now consider the situation when region $R$ has build some search trees and the label of a boundary vertex $w$ is risen from $d(w)$ to $d'(w)$ (as a result of update from the neighboring region or one of the heuristics). All the vertices in the search tree starting from $w$ were previously marked with $d(w)$ and have to be declared as free vertices or adopted to any other valid tree with root label $d(w)$. The adaptation is performed by the same mechanism as in BK. The situation when a preflow is injected from the neighboring region and (a part of) a search tree becomes disconnected is also handled by the orphan adaptation mechanism.
\par
The combination of the above improvements allows S-ARD to run in about the same time as BK on all tested vision instance (Sect.~\ref{sec:exp_streaming}), sometimes being even significantly faster (154 s. vs. 245 s. on \verb=BL06-gargoyle-lrg=).
%\par
%
%To fully maintain the search trees 
%
%In the end of \ARD we need to calculate the new label of every node u as $\min\{k \mid u\rightarrow T_k\}$, where $T_k$ is the set of boundary nodes with label below $k$. It can be done by the region-relabel algorithm~\ref{alg:region-relabel}. However, this computation is redundant. At stage $0$ of \ARD we find augmenting paths to the sink, this process builds a search tree of vertices from which the sink is reachable. These are exactly the vertices which have to be assigned label $0$. When we find What we need is to maintain also search trees for the boundary vertices.
%
\section{Experiments}\label{sec:experiments}
All experiments are conducted on a system with Intel Core 2 Quad CPU@2.66Hz, 4GB memory, Windows XP 32bit and Microsoft VC compiler. There is 3 series of experiments:
\begin{itemize}
\item Synthetic experiments, where we observe general dependencies of the algorithms, with some statistical significance, \ie not being biased to a particular problem instance. It also serves as an empirical validation, as thousands of instances are solved. Here, the basic implementation of S-ARD was used.
\item Sequential competition. We study sequential versions of the algorithms, running them on real vision instances. Only a single core of the CPU is utilized. We fix the region partition and study how much disk I/O it would take to solve each problem when only one region can be loaded in the memory at a time. In this and the next experiment we used the efficient implementation of ARD. Note, in~\cite{mixed_maxflow-11} we reported worse results with the earlier implementation.
\item Parallel competition. Parallel algorithms are tested on the instances which can fully fit in 2GB of memory. All 4 cores of the CPU are allowed to be used. We compare our algorithms with two other state-of-the-art distributed implementations.
\end{itemize}
\subsection{General Dependences: Synthetic Problems}\label{sec:exp_rand}
We generated simple synthetic problems to validate the algorithms. The network is constructed as a 2D grid with a regular connectivity structure. Figure~\ref{fig:graph}(a) shows an example of such a network. 
\begin{figure}[ht]
\begin{center}
\setlength{\tabcolsep}{0pt}
\begin{tabular}{cc}
\begin{tabular}{c}\includegraphics[width=0.25\linewidth]{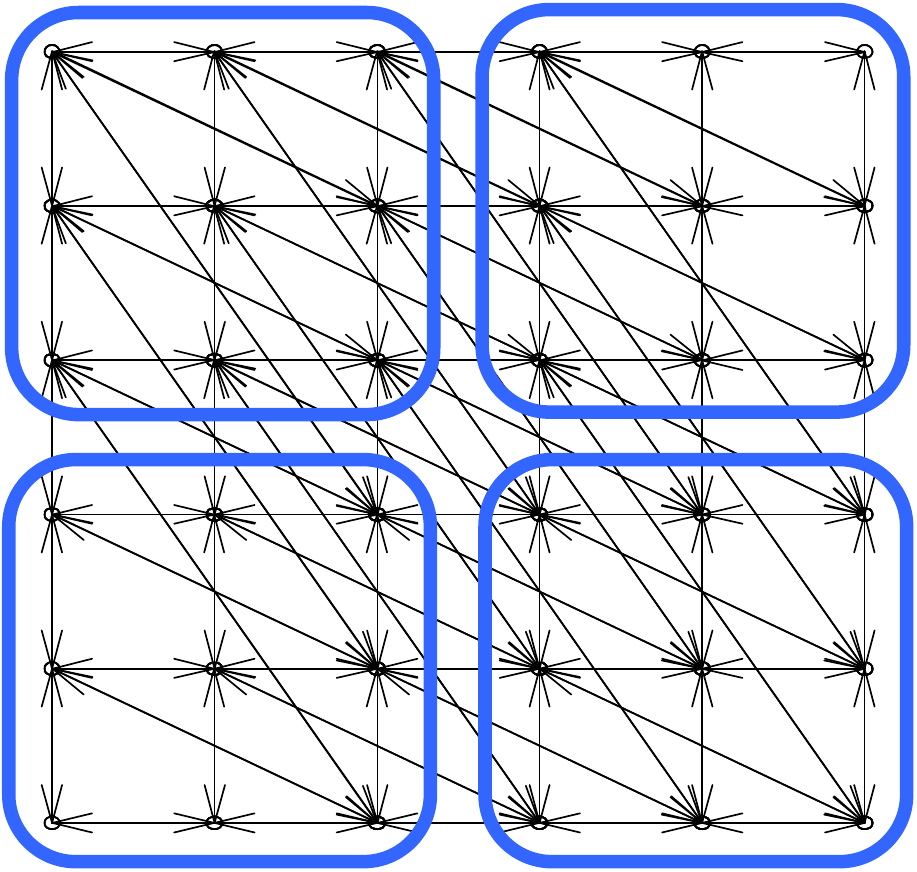}\\
(a)
\end{tabular}\ \ \ &
\begin{tabular}{c}\includegraphics[width=0.6\linewidth]{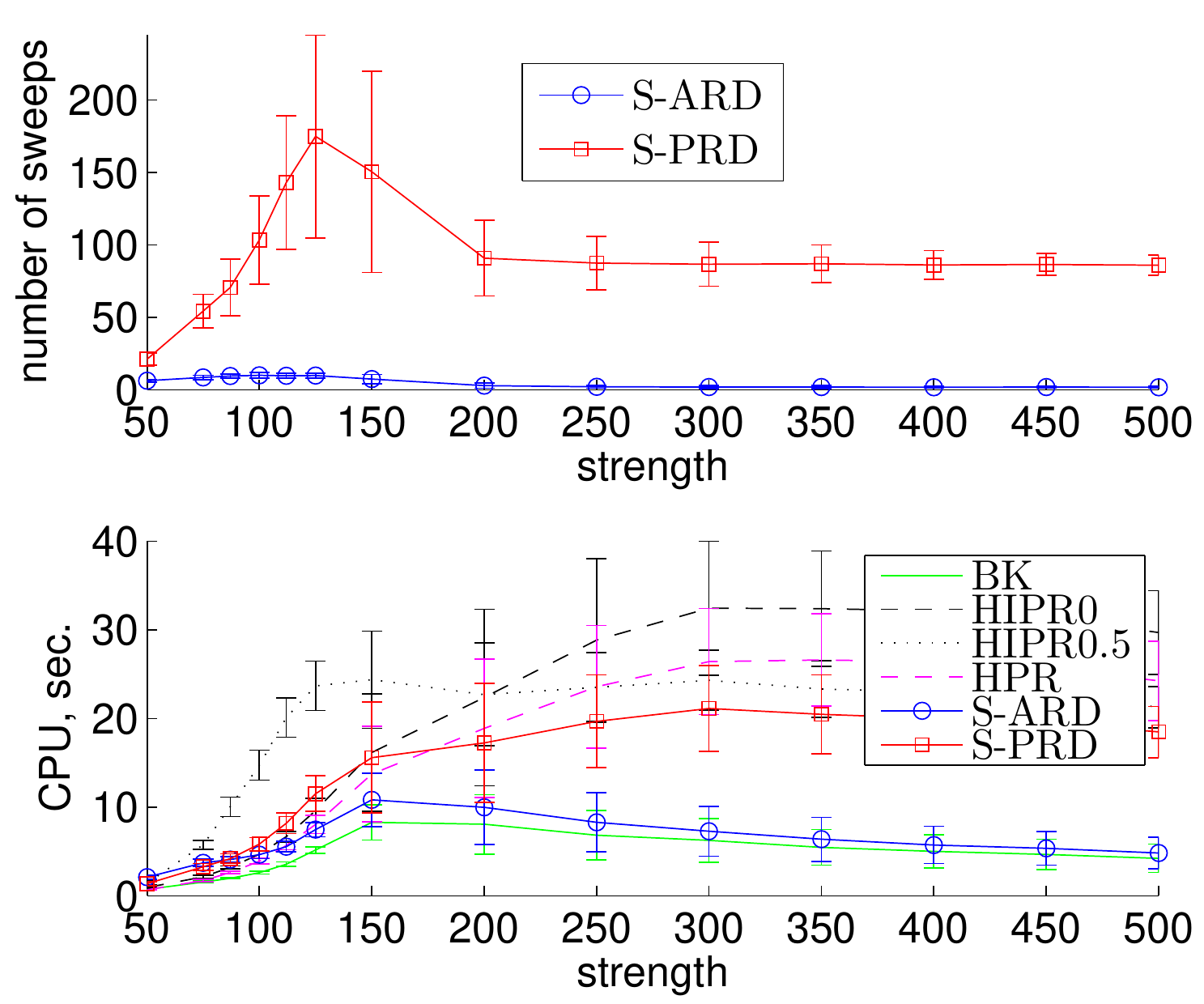}\\
(b)
\end{tabular}
\end{tabular}
\end{center}
\caption{(a) Example of a synthetic problem: a network of the size $6{\times}6$, connectivity 8, partitioned into 4 regions. The source and sink are not shown. (b) Dependence on the interaction strength, for size $1000{\times}1000$, connectivity 8 and 4 regions. Plots show mean values and intervals containing 70\% of the samples.}
\label{fig:graph}
\end{figure}
The edges are added to the nodes at the following relative displacements {\small $(0,1)$, $(1,0)$, $(1,2)$, $(2,1)$, $(1,3)$, $(3,1)$, $(2,3)$, $(3,2)$, $(0,2)$, $(2,0)$, $(2,2)$, $(3,3)$, $(3,4)$, $(4,2)$}. By {\em connectivity} we mean the number of edges incident to a node far enough from the boundary. Adding pairs {\small $(0,1)$, $(1,0)$} results in connectivity 4 and so on. Each node is given an integer excess/deficit distributed uniformly in the interval $[-500\  500]$. A positive number means a source link and a negative number a sink link. All edges in the graph are assigned a constant capacity, called~{\em strength}. %
 The network is partitioned into regions by slicing it in $s$ equal parts in both dimensions. Thus we have 4 parameters: the number of nodes, the connectivity, the strength and the number of regions. We generate 100 instances for each value of the parameters.
\par
\begin{figure}[ht]
\begin{center}
\includegraphics[width=0.6\linewidth]{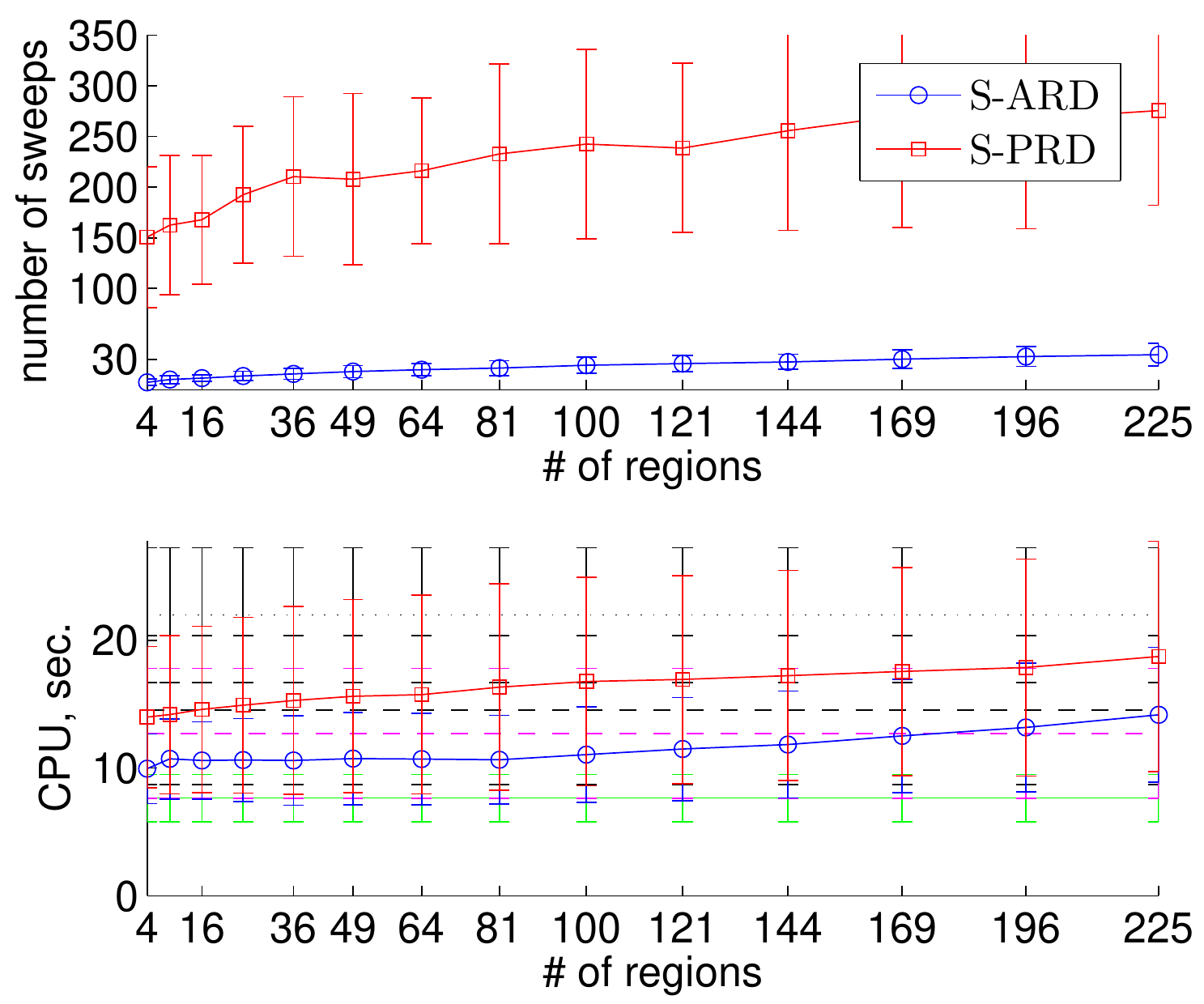}
\end{center}
\caption{Dependence on the number of regions, for size $1000{\times}1000$, connectivity 8, strength 150.}
\label{fig:regions}
\end{figure}
\begin{figure}[ht]
\begin{center}
\includegraphics[width=0.6\linewidth]{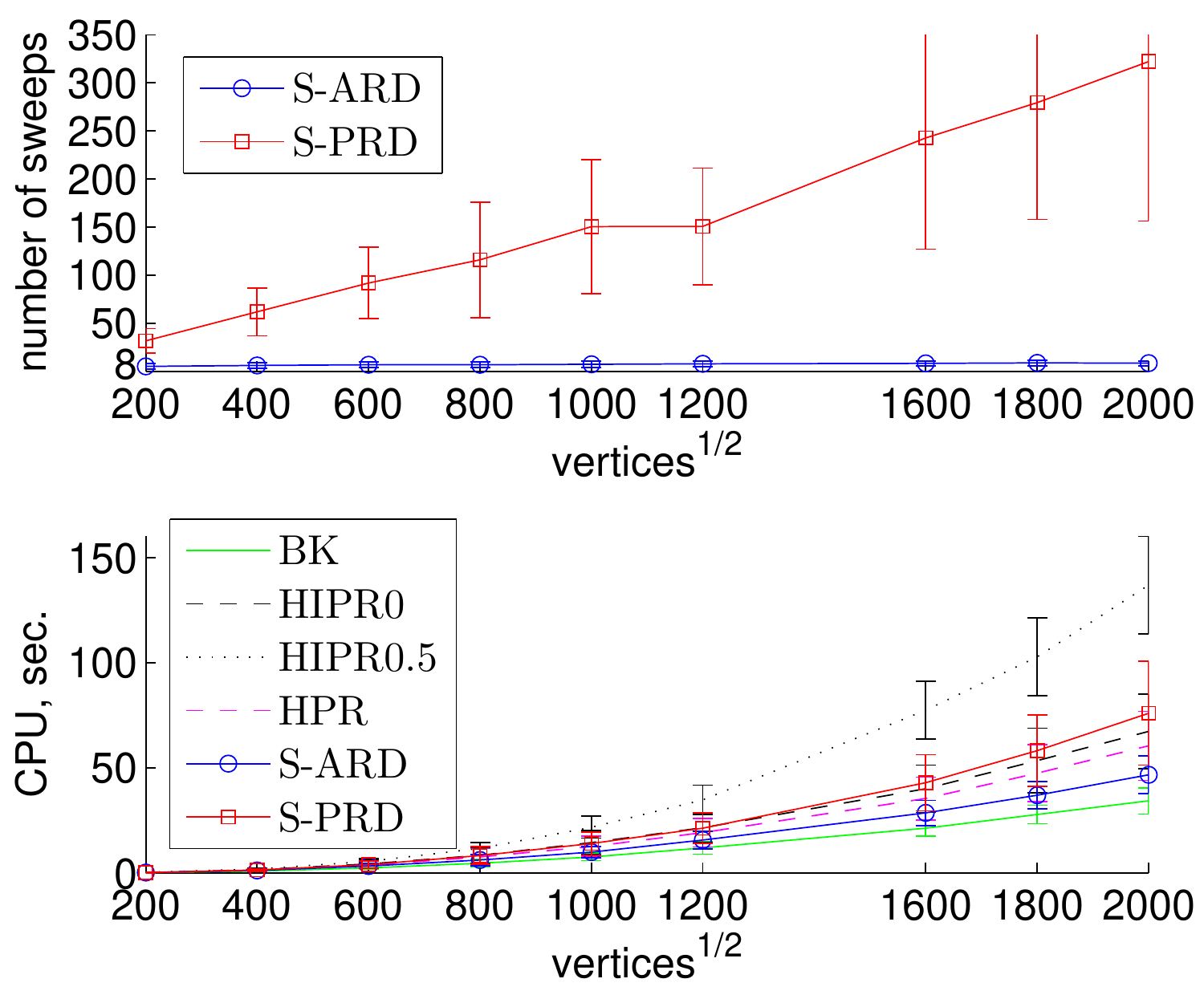}
\end{center}
\caption{Dependence on the problem size, for connectivity 8, strength 150, 4 regions.}
\label{fig:size}
\end{figure}
\begin{figure}[ht]
\begin{center}
\includegraphics[width=0.6\linewidth]{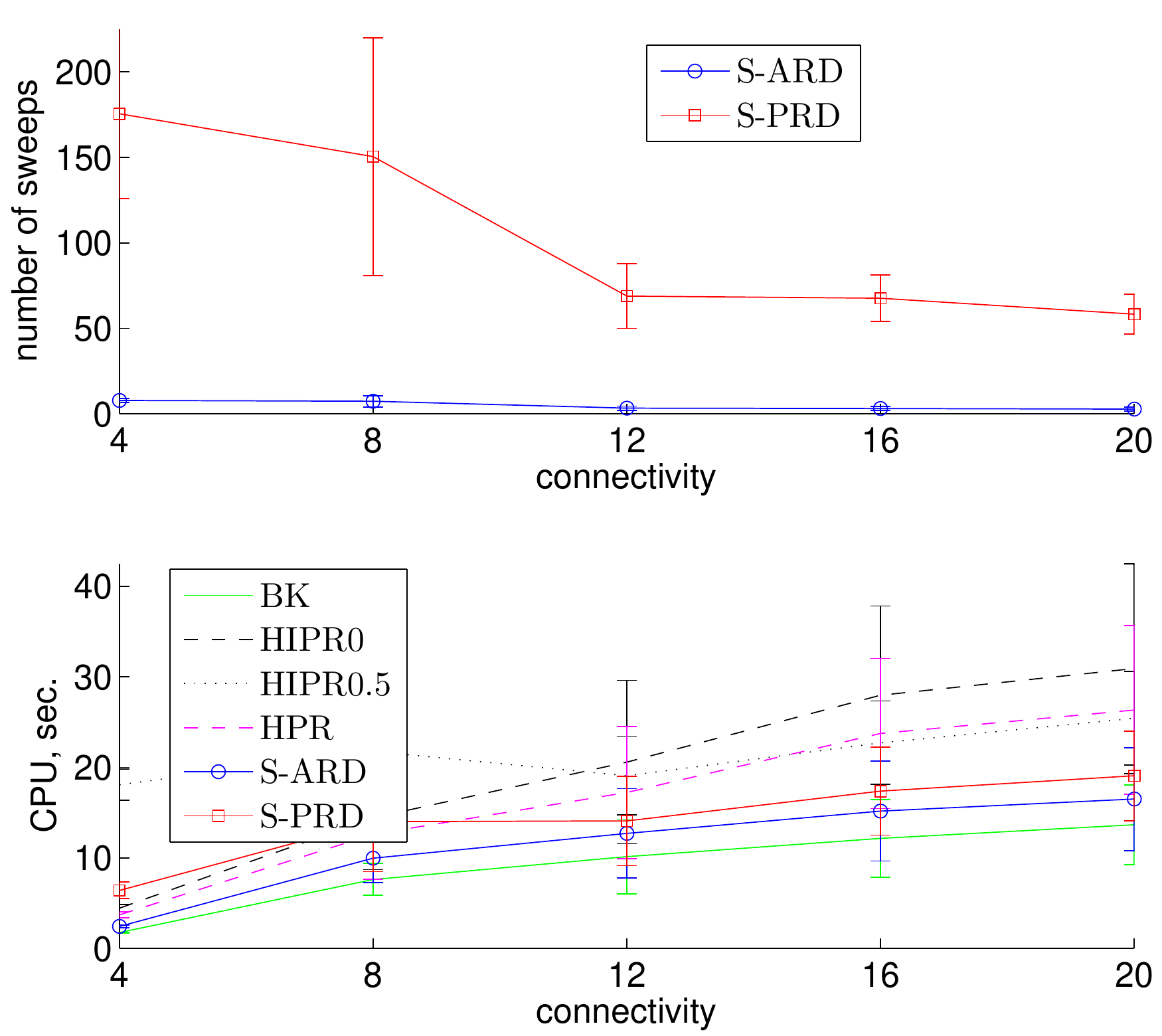}
\end{center}
\caption{
Dependence on the connectivity, for size $1000{\times}1000$, strength = $(150\cdot 8)$/connectivity, 4 regions.
}
\label{fig:connectivity}
\end{figure}
Let us first look at the dependence on the strength shown in Figure~\ref{fig:graph}(b).
Problems with small strength are easy, because they are very local -- long augmentation paths do not occur. For problems with large strength long paths needs to be augmented. However, finding them is easy because bottlenecks are unlikely. Therefore BK and S-ARD have a maximum in the computation time somewhere in the middle. It is more difficult to transfer the flow over long distances for push-relabel algorithms. This is where the global relabel heuristic becomes efficient and HIPR$0.5$ outperforms HIPR$0$. The region-relabel heuristic of S-PRD allows it to outperform other push-relabel variants.
\par
In general, we think all such random 2D networks are too easy. Nevertheless, they are useful and instructive to show basic dependences. We now select the ``difficult'' point for BK with the strength~150 and study other dependencies:
\begin{itemize}
\item The number of regions (Figure~\ref{fig:regions}). For this problem family both the number of sweeps and the computation time grows slowly with the number of regions.
\item The problem size (Figure~\ref{fig:size}). Computation efforts of all algorithms grow proportionally. However, the number of sweeps shows different asymptotes. It is almost constant for S-ARD but grows significantly for S-PRD.
\item Connectivity (Figure~\ref{fig:connectivity}). Connectivity is not independent of the strength. Roughly, $4$ edges with capacity $100$ can transmit as much flow as $8$ edges with capacity $50$. Therefore while increasing the connectivity we also decrease the  strength as $(150\cdot 8)$/connectivity in this plot. 
\item  Workload (Figure~\ref{fig:workload}). This plot shows how much time each of the algorithms spends performing different parts of computation. Note that the problems are solved on a single computer with all regions kept in memory, therefor the time on sending messages should be understood as updates of dynamic data structure of the region \wrt the new labeling and flow on the boundary.
For S-PRD more sweeps are needed, so the total time spent in messages and gap heuristic is increased. Additionally, the gap heuristic has to take into account all nodes, unlike only the boundary nodes in S-ARD.
\end{itemize}
\begin{figure}[ht]\label{fig:workload}
\begin{center}
\includegraphics[width=0.75\linewidth]{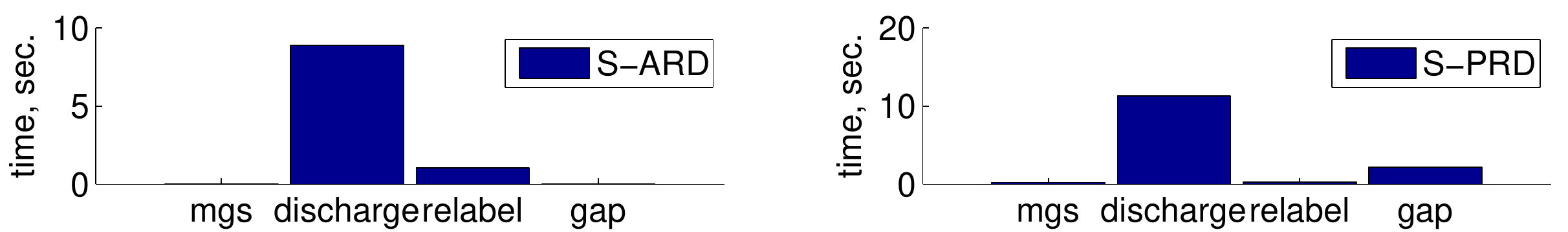}
\end{center}
\caption{Workload distribution, for size $1000{\times}1000$, connectivity 8, 4 regions, strength 150. {\em msg} -- passing the messages (updating flow and labels on the boundary), {\em discharge} -- work done by the core solver (BK for S-ARD and HPR for S-PRD), {\em relabel} -- the region-relabel operation, {\em gap} -- the global gap heuristic.}
\end{figure}
\subsection{Sequential Competition}\label{sec:exp_streaming}
We tested our algorithms on the {\sc maxflow} problem instances published by the Computer Vision Research Group at the University of Western Ontario\footnote{ \url{http://vision.csd.uwo.ca/maxflow-data/}}. The data consist of typical max-flow problems in computer vision, graphics, and biomedical image analysis.
{\bf Stereo} instances are sequences of subproblems (arising in the expansion move algorithm) for which the total time should be reported. There are two models: BVZ~\cite{BVZ}, in which the graph is a 4-connected 2D grid, and KZ2~\cite{KZ2}, in which there are additional long-range links. %{\mypar{Multiview Reconstruction}
{\bf Multiview} 3D reconstruction models LB06~\cite{LB06} and BL06~\cite{BL06}. Graphs of these problems are cellular complexes subdividing the space into 3D cubes and each cube into 24 smaller cells. {\bf Surface} fitting instances LB07~\cite{LB07} are 6-connected 3D grid graphs.
And finally, there is a collection of volumetric {\bf segmentation} instances BJ01~\cite{BJ01}, BF06~\cite{BF06}, BK03~\cite{BK03} with 6-connected and 26-connected 3D grid graphs.
\par
To test our streaming algorithms, we used the \verb=regulargrid= hint
available in the definition of the problems to select the regions by slicing the problem into $4$ parts in each dimension -- into $16$ regions for 2D BVZ grids and into $64$ regions for 3D segmentation instances. Problems KZ2 do not have such a hint (they are not regular grids), so we sliced them into 16 pieces just by the node number. The same we did for the multiview LB06 instances. Though they have a \verb=size= hint, we failed to interpret the node layout correctly (the separator set, $\B$, was unexpectedly large when trying to slice along the dimensions). So we sliced them purely by the node number.
\par
One of the problems we faced is pairing the arcs which are reverse of each other. While in stereo, surface and multiview problems, the reverse arcs are consequent in the files, and can be easily paired, in 3D segmentation they are not. For a generic algorithm, not being aware of the problem's regularity structure, it is actually a non-trivial problem requiring at least the memory to read all of the arcs first. Because our goal is a relative comparison, we did not pair the arcs in 3D segmentation. This means we kept twice as many arcs than necessary for those problems. In Table~\ref{table1} this is seen, \eg for \verb=babyface.n26c100=, which is 26-connected, but we construct a multigraph (has parallel arcs) with average node degree of 49. For some other instances, however, this is not visible, because there could be many zero arcs, \eg \verb=liver.n26c10= which is a 26-connected grid too, but has the average node degree of 10.4 with unpaired arcs. The comparison among different methods is correct, since all of them are given exactly the same multigraph.
\par
%All tested implementations are 32bit. The memory limit on our system was 2GB. 
The results are presented in Table~\ref{table1}. %For all solvers, we measured the time spent purely to solve the problem (denoted CPU in the table). We excluded time for parsing and construction. We also excluded time for external I/O in the case of streaming solvers. 
We did measure the time of disk I/O, however it depends on the hard drive performance, other concurrently running processes as well as on system file caching (has effect for small problems) and therefore we report only bytes written/loaded. Note that disk I/O is not proportional to the number of sweeps, because some regions may be inactive during a sweep and thus skipped. %$K$ -- number of regions. RAM -- memory taken by the solver; for BK in the case it exceeds 2GB memory limit, the expected required memory (in the 32bit model, assuming it would store 32bit indices instead of pointers). 
%Whenever BK algorithm failed to run, exceeding the memory limit, we printed the expected required memory (in the 32bit model, assuming it would store 32bit indices instead of pointers). 
For HIPR we do not monitor the memory usage. It is slightly higher than that of HPR, because of keeping initial arc capacities. %For streaming solvers the division of memory usage into shared and region memory is shown in the form ``shared+region''. I/O -- total number of bytes read or written to the disk.
\par
For verification of solvers, we compared the flow values to the ground truth solution provided in the dataset. Additionally, we saved the cut output from each solver and checked its cost independently. Verifying the cost of the cut is relatively easy: the cut can be kept in memory and the edges can be processed form the DIMACS problem definition file on-line. An independent check of (pre-)flow feasibility would be necessary for full verification of a solver. However, it requires storing the full graph in memory and was not implemented. 
\par
%\par
{\scriptsize
\setlength{\tabcolsep}{2pt}
\begin{longtable}{|lcc|c|c|c|c|ccc|ccc|}
\multicolumn{13}{p{\linewidth}}{{\small\bfseries\tablename\ \thetable{}.} \small Sequential Competition. CPU -- the time spent purely for computation, excluding the time for parsing, construction and disk I/O. The total time to solve the problem is not shown. $K$ -- number of regions. RAM -- memory taken by the solver; for BK in the case it exceeds 2GB limit, the expected required memory; for streaming solvers the sum of shared and region memory. I/O -- total bytes read or written to the disk. \vspace{1mm}} \label{table1}\\
\par
\hline
\multicolumn{3}{|c|}{problem} & BK & HIPR0 & HIPR0.5 & HPR & \multicolumn{3}{|c|}{S-ARD} & \multicolumn{3}{|c|}{S-PRD}\\* 
\hline
\rowcolor{blue!10}
 name & n($10^6)$ & m/n & CPU & CPU & CPU & CPU & CPU & \hskip-1mm sweeps\hskip-1mm &  $K$  & CPU & \hskip-1mm sweeps\hskip-1mm &   $K$ \\* 
\rowcolor{blue!10}
 size &         &     & RAM & RAM & RAM & RAM & RAM &    & I/O & RAM &    & I/O\\* 
\endfirsthead
\multicolumn{13}{c}%
{{\bfseries \tablename\ \thetable{} -- continued from previous page}} \\
\hline
\endhead
\hline \multicolumn{13}{|r|}{{Continued on next page}} \\ \hline
\endfoot
\hline \hline
\endlastfoot
\hline
\multicolumn{13}{|c|}{\bf stereo}\\* 
\hline 
\showrowcolors
\rowcolor{blue!10}
BVZ-sawtooth(20) & $0.2$ & $4.0$ & $0.68$s & $3.0$s & $7.7$s & $3.8$s & $0.63$s & $  6$ & $ 16$ & $3.7$s & $ 32$ & $ 16$ \\* 
\showrowcolors
\rowcolor{blue!10}
\multicolumn{3}{|>{\columncolor[rgb]{0.9,0.9,1}}l|}{434$\times$380, \ \ $ 14$MB} & $ 14$MB &  &  & $ 17$MB & $0.3{+}0.9$MB & & $114$MB & $0.8{+}1.1$MB & & $0.7$GB\\ 
\hiderowcolors
BVZ-tsukuba(16) & $0.1$ & $4.0$ & $0.36$s & $1.9$s & $4.9$s & $2.6$s & $0.35$s & $  5$ & $ 16$ & $2.1$s & $ 29$ & $ 16$ \\* 
\hiderowcolors
\multicolumn{3}{|l|}{384$\times$288, \ \ $8.6$MB} & $9.7$MB &  &  & $ 11$MB & $0.2{+}0.6$MB & & $ 71$MB & $0.5{+}0.8$MB & & $373$MB\\ 
\showrowcolors
\rowcolor{blue!10}
BVZ-venus(22) & $0.2$ & $4.0$ & $1.2$s & $5.7$s & $15$s & $6.2$s & $1.1$s & $  6$ & $ 16$ & $6.6$s & $ 36$ & $ 16$ \\* 
\showrowcolors
\rowcolor{blue!10}
\multicolumn{3}{|>{\columncolor[rgb]{0.9,0.9,1}}l|}{434$\times$383, \ \ $ 14$MB} & $ 15$MB &  &  & $ 17$MB & $0.3{+}0.9$MB & & $119$MB & $0.8{+}1.1$MB & & $0.9$GB\\ 
\hiderowcolors
KZ2-sawtooth(20) & $0.3$ & $5.8$ & $1.8$s & $7.1$s & $22$s & $6.1$s & $2.2$s & $  6$ & $ 16$ & $7.4$s & $ 23$ & $ 16$ \\* 
\hiderowcolors
\multicolumn{3}{|l|}{$ 38$MB} & $ 33$MB &  &  & $ 36$MB & $1.2{+}2.0$MB & & $280$MB & $1.8{+}2.5$MB & & $1.2$GB\\ 
\showrowcolors
\rowcolor{blue!10}
KZ2-tsukuba(16) & $0.2$ & $5.9$ & $1.1$s & $5.3$s & $20$s & $4.4$s & $1.4$s & $  6$ & $ 16$ & $5.9$s & $ 18$ & $ 16$ \\* 
\showrowcolors
\rowcolor{blue!10}
\multicolumn{3}{|>{\columncolor[rgb]{0.9,0.9,1}}l|}{$ 26$MB} & $ 23$MB &  &  & $ 25$MB & $1.1{+}1.4$MB & & $186$MB & $1.4{+}1.7$MB & & $0.7$GB\\ 
\hiderowcolors
KZ2-venus(22) & $0.3$ & $5.8$ & $2.8$s & $13$s & $39$s & $10$s & $3.4$s & $  8$ & $ 16$ & $14$s & $ 36$ & $ 16$ \\* 
\hiderowcolors
\multicolumn{3}{|l|}{$ 38$MB} & $ 34$MB &  &  & $ 37$MB & $1.2{+}2.1$MB & & $330$MB & $1.9{+}2.5$MB & & $1.8$GB\\ 
\hline
\multicolumn{13}{|c|}{\bf multiview}\\* 
\hline 
\showrowcolors
\rowcolor{blue!10}
BL06-camel-lrg & $18.9$ & $4.0$ & $81$s &  &  &  & $63$s & $ 11$ & $ 16$ & $308$s & $418$ & $ 16$ \\* 
\showrowcolors
\rowcolor{blue!10}
\multicolumn{3}{|>{\columncolor[rgb]{0.9,0.9,1}}l|}{100$\times$75$\times$105$\times$24, \ \ $2.0$GB} & $1.6$GB &  &  &  & $ 19{+}103$MB & & $ 28$GB & $ 86{+}122$MB & & $0.6$TB\\ 
\hiderowcolors
BL06-camel-med & $9.7$ & $4.0$ & $25$s & $29$s & $77$s & $59$s & $20$s & $ 12$ & $ 16$ & $118$s & $227$ & $ 16$ \\* 
\hiderowcolors
\multicolumn{3}{|l|}{80$\times$60$\times$84$\times$24, \ \ $1.0$GB} & $0.8$GB &  &  & $1.0$GB & $ 31{+} 53$MB & & $ 16$GB & $ 46{+} 63$MB & & $225$GB\\ 
\showrowcolors
\rowcolor{blue!10}
BL06-camel-sml & $1.2$ & $4.0$ & $0.98$s & $1.5$s & $6.3$s & $1.8$s & $0.96$s & $  9$ & $ 16$ & $4.2$s & $ 47$ & $ 16$ \\* 
\showrowcolors
\rowcolor{blue!10}
\multicolumn{3}{|>{\columncolor[rgb]{0.9,0.9,1}}l|}{40$\times$30$\times$42$\times$24, \ \ $115$MB} & $106$MB &  &  & $124$MB & $8.0{+}7.0$MB & & $1.4$GB & $6.9{+}8.2$MB & & $9.1$GB\\ 
\hiderowcolors
BL06-gargoyle-lrg & $17.2$ & $4.0$ & $245$s &  &  & $91$s & $154$s & $ 21$ & $ 16$ & $318$s & $354$ & $ 16$ \\* 
\hiderowcolors
\multicolumn{3}{|l|}{80$\times$112$\times$80$\times$24, \ \ $1.8$GB} & $1.5$GB &  &  & $1.7$GB & $ 23{+} 95$MB & & $ 35$GB & $ 82{+}112$MB & & $0.8$TB\\ 
\showrowcolors
\rowcolor{blue!10}
BL06-gargoyle-med & $8.8$ & $4.0$ & $115$s & $17$s & $58$s & $37$s & $73$s & $ 16$ & $ 16$ & $143$s & $340$ & $ 16$ \\* 
\showrowcolors
\rowcolor{blue!10}
\multicolumn{3}{|>{\columncolor[rgb]{0.9,0.9,1}}l|}{64$\times$90$\times$64$\times$24, \ \ $0.9$GB} & $0.8$GB &  &  & $0.9$GB & $ 37{+} 50$MB & & $ 14$GB & $ 44{+} 58$MB & & $235$GB\\ 
\hiderowcolors
BL06-gargoyle-sml & $1.1$ & $4.0$ & $6.1$s & $1.2$s & $3.0$s & $1.7$s & $3.9$s & $ 10$ & $ 16$ & $4.4$s & $ 55$ & $ 16$ \\* 
\hiderowcolors
\multicolumn{3}{|l|}{32$\times$45$\times$32$\times$24, \ \ $106$MB} & $ 97$MB &  &  & $114$MB & $9.3{+}6.6$MB & & $1.3$GB & $6.9{+}7.7$MB & & $9.4$GB\\ 
\hline
\multicolumn{13}{|c|}{\bf surface}\\* 
\hline 
\showrowcolors
\rowcolor{blue!10}
LB07-bunny-lrg & $49.5$ & $6.0$ &  &  &  &  & $15$s & $  6$ & $ 64$ & $416$s & $ 43$ & $ 64$ \\* 
\showrowcolors
\rowcolor{blue!10}
\multicolumn{3}{|>{\columncolor[rgb]{0.9,0.9,1}}l|}{401$\times$396$\times$312, \ \ $6.6$GB} & $5.7$GB &  &  &  & $130{+} 87$MB & & $ 49$GB & $226{+} 99$MB & & $276$GB\\ 
\hiderowcolors
LB07-bunny-med & $6.3$ & $6.0$ & $1.6$s & $20$s & $41$s & $26$s & $2.1$s & $ 10$ & $ 64$ & $16$s & $ 27$ & $ 64$ \\* 
\hiderowcolors
\multicolumn{3}{|l|}{202$\times$199$\times$157, \ \ $0.8$GB} & $0.7$GB &  &  & $0.8$GB & $ 33{+} 12$MB & & $6.5$GB & $ 43{+}863$MB & & $0.0$MB\\ 
\showrowcolors
\rowcolor{blue!10}
LB07-bunny-sml & $0.8$ & $5.9$ & $0.17$s & $0.80$s & $1.8$s & $1.1$s & $0.32$s & $  9$ & $ 64$ & $0.86$s & $ 19$ & $ 64$ \\* 
\showrowcolors
\rowcolor{blue!10}
\multicolumn{3}{|>{\columncolor[rgb]{0.9,0.9,1}}l|}{102$\times$100$\times$79, \ \ $ 94$MB} & $ 95$MB &  &  & $101$MB & $8.2{+}1.6$MB & & $0.8$GB & $7.9{+}1.9$MB & & $2.0$GB\\ 
\hline
\multicolumn{13}{|c|}{\bf segm}\\* 
\hline 
\hiderowcolors
liver.n26c10 & $4.2$ & $10.4$ & $6.4$s & $18$s & $18$s & $34$s & $14$s & $ 13$ & $ 64$ & $39$s & $157$ & $ 64$ \\* 
\hiderowcolors
\multicolumn{3}{|l|}{170$\times$170$\times$144, \ \ $2.1$GB} & $0.8$GB &  &  & $0.7$GB & $ 36{+} 12$MB & & $ 13$GB & $ 30{+} 13$MB & & $ 82$GB\\ 
\showrowcolors
\rowcolor{blue!10}
liver.n26c100 & $4.2$ & $11.1$ & $12$s & $26$s & $28$s & $39$s & $24$s & $ 15$ & $ 64$ & $35$s & $ 98$ & $ 64$ \\* 
\showrowcolors
\rowcolor{blue!10}
\multicolumn{3}{|>{\columncolor[rgb]{0.9,0.9,1}}l|}{170$\times$170$\times$144, \ \ $2.1$GB} & $0.8$GB &  &  & $0.7$GB & $ 38{+} 13$MB & & $ 16$GB & $ 30{+} 14$MB & & $ 66$GB\\ 
\hiderowcolors
liver.n6c10 & $4.2$ & $9.8$ & $7.2$s & $17$s & $25$s & $40$s & $14$s & $ 16$ & $ 64$ & $36$s & $151$ & $ 64$ \\* 
\hiderowcolors
\multicolumn{3}{|l|}{170$\times$170$\times$144, \ \ $498$MB} & $0.7$GB &  &  & $0.7$GB & $ 33{+} 12$MB & & $ 15$GB & $ 28{+} 13$MB & & $ 79$GB\\ 
\showrowcolors
\rowcolor{blue!10}
liver.n6c100 & $4.2$ & $10.5$ & $15$s & $30$s & $34$s & $44$s & $19$s & $ 17$ & $ 64$ & $32$s & $ 94$ & $ 64$ \\* 
\showrowcolors
\rowcolor{blue!10}
\multicolumn{3}{|>{\columncolor[rgb]{0.9,0.9,1}}l|}{170$\times$170$\times$144, \ \ $512$MB} & $0.8$GB &  &  & $0.7$GB & $ 35{+} 12$MB & & $ 14$GB & $ 29{+} 13$MB & & $ 70$GB\\ 
\hiderowcolors
babyface.n26c10 & $5.1$ & $47.3$ &  &  &  &  & $179$s & $ 38$ & $ 64$ & $222$s & $169$ & $ 64$ \\* 
\hiderowcolors
\multicolumn{3}{|l|}{250$\times$250$\times$81, \ \ $2.5$GB} & $3.7$GB &  &  &  & $156{+} 56$MB & & $102$GB & $173{+} 58$MB & & $0.6$TB\\ 
\showrowcolors
\rowcolor{blue!10}
babyface.n26c100 & $5.1$ & $49.0$ &  &  &  &  & $231$s & $ 44$ & $ 64$ & $262$s & $116$ & $ 64$ \\* 
\showrowcolors
\rowcolor{blue!10}
\multicolumn{3}{|>{\columncolor[rgb]{0.9,0.9,1}}l|}{250$\times$250$\times$81, \ \ $2.7$GB} & $3.8$GB &  &  &  & $156{+} 56$MB & & $115$GB & $180{+} 57$MB & & $0.6$TB\\ 
\hiderowcolors
babyface.n6c10 & $5.1$ & $11.1$ & $6.8$s & $38$s & $51$s & $68$s & $20$s & $ 17$ & $ 64$ & $100$s & $275$ & $ 64$ \\* 
\hiderowcolors
\multicolumn{3}{|l|}{250$\times$250$\times$81, \ \ $0.6$GB} & $1.0$GB &  &  & $0.9$GB & $ 22{+} 16$MB & & $ 19$GB & $ 37{+} 17$MB & & $261$GB\\ 
\showrowcolors
\rowcolor{blue!10}
babyface.n6c100 & $5.1$ & $11.5$ & $13$s & $71$s & $65$s & $87$s & $24$s & $ 19$ & $ 64$ & $74$s & $191$ & $ 64$ \\* 
\showrowcolors
\rowcolor{blue!10}
\multicolumn{3}{|>{\columncolor[rgb]{0.9,0.9,1}}l|}{250$\times$250$\times$81, \ \ $0.6$GB} & $1.0$GB &  &  & $0.9$GB & $ 22{+} 16$MB & & $ 18$GB & $ 37{+} 17$MB & & $189$GB\\ 
\hiderowcolors
adhead.n26c10 & $12.6$ & $31.5$ &  &  &  &  & $128$s & $ 17$ & $ 64$ & $224$s & $109$ & $ 64$ \\* 
\hiderowcolors
\multicolumn{3}{|l|}{256$\times$256$\times$192, \ \ $6.5$GB} & $6.2$GB &  &  &  & $153{+} 83$MB & & $ 84$GB & $195{+} 86$MB & & $0.8$TB\\ 
\showrowcolors
\rowcolor{blue!10}
adhead.n26c100 & $12.6$ & $31.6$ &  &  &  &  & $174$s & $ 21$ & $ 64$ & $269$s & $129$ & $ 64$ \\* 
\showrowcolors
\rowcolor{blue!10}
\multicolumn{3}{|>{\columncolor[rgb]{0.9,0.9,1}}l|}{256$\times$256$\times$192, \ \ $6.7$GB} & $6.3$GB &  &  &  & $153{+} 83$MB & & $ 90$GB & $196{+} 86$MB & & $0.8$TB\\ 
\hiderowcolors
adhead.n6c10 & $12.6$ & $11.6$ &  &  &  &  & $36$s & $ 13$ & $ 64$ & $119$s & $161$ & $ 64$ \\* 
\hiderowcolors
\multicolumn{3}{|l|}{256$\times$256$\times$192, \ \ $1.6$GB} & $2.5$GB &  &  &  & $ 34{+} 36$MB & & $ 31$GB & $ 77{+} 39$MB & & $372$GB\\ 
\showrowcolors
\rowcolor{blue!10}
adhead.n6c100 & $12.6$ & $11.7$ &  &  &  &  & $49$s & $ 20$ & $ 64$ & $121$s & $165$ & $ 64$ \\* 
\showrowcolors
\rowcolor{blue!10}
\multicolumn{3}{|>{\columncolor[rgb]{0.9,0.9,1}}l|}{256$\times$256$\times$192, \ \ $1.6$GB} & $2.5$GB &  &  &  & $ 34{+} 36$MB & & $ 36$GB & $ 77{+} 39$MB & & $354$GB\\ 
\hiderowcolors
bone.n26c10 & $7.8$ & $32.3$ &  &  &  &  & $25$s & $ 12$ & $ 64$ & $96$s & $148$ & $ 64$ \\* 
\hiderowcolors
\multicolumn{3}{|l|}{256$\times$256$\times$119, \ \ $3.9$GB} & $4.0$GB &  &  &  & $122{+} 61$MB & & $ 35$GB & $147{+} 63$MB & & $470$GB\\ 
\showrowcolors
\rowcolor{blue!10}
bone.n26c100 & $7.8$ & $32.4$ &  &  &  &  & $29$s & $ 14$ & $ 64$ & $68$s & $124$ & $ 64$ \\* 
\showrowcolors
\rowcolor{blue!10}
\multicolumn{3}{|>{\columncolor[rgb]{0.9,0.9,1}}l|}{256$\times$256$\times$119, \ \ $4.1$GB} & $4.0$GB &  &  &  & $122{+} 61$MB & & $ 39$GB & $147{+} 63$MB & & $321$GB\\ 
\hiderowcolors
bone.n6c10 & $7.8$ & $11.5$ & $7.7$s & $5.7$s & $17$s & $12$s & $7.2$s & $  9$ & $ 64$ & $37$s & $195$ & $ 64$ \\* 
\hiderowcolors
\multicolumn{3}{|l|}{256$\times$256$\times$119, \ \ $0.9$GB} & $1.5$GB &  &  & $1.4$GB & $ 62{+} 23$MB & & $ 13$GB & $ 52{+} 25$MB & & $188$GB\\ 
\showrowcolors
\rowcolor{blue!10}
bone.n6c100 & $7.8$ & $11.6$ & $9.1$s & $9.1$s & $22$s & $14$s & $8.7$s & $ 10$ & $ 64$ & $23$s & $ 65$ & $ 64$ \\* 
\showrowcolors
\rowcolor{blue!10}
\multicolumn{3}{|>{\columncolor[rgb]{0.9,0.9,1}}l|}{256$\times$256$\times$119, \ \ $1.0$GB} & $1.6$GB &  &  & $1.5$GB & $ 62{+} 23$MB & & $ 13$GB & $ 52{+} 25$MB & & $104$GB\\ 
\hiderowcolors
bone\_subx.n6c100 & $3.9$ & $11.8$ & $7.1$s & $6.3$s & $12$s & $6.4$s & $5.5$s & $ 12$ & $ 64$ & $9.4$s & $ 42$ & $ 64$ \\* 
\hiderowcolors
\multicolumn{3}{|l|}{128$\times$256$\times$119, \ \ $495$MB} & $0.8$GB &  &  & $0.7$GB & $ 39{+} 12$MB & & $7.1$GB & $ 29{+} 13$MB & & $ 42$GB\\ 
\showrowcolors
\rowcolor{blue!10}
bone\_subxy.n26c100 & $1.9$ & $32.2$ & $5.9$s & $3.9$s & $6.1$s & $4.6$s & $7.3$s & $ 13$ & $ 64$ & $8.7$s & $ 33$ & $ 64$ \\* 
\showrowcolors
\rowcolor{blue!10}
\multicolumn{3}{|>{\columncolor[rgb]{0.9,0.9,1}}l|}{128$\times$128$\times$119, \ \ $1.0$GB} & $1.0$GB &  &  & $0.8$GB & $ 92{+}851$MB & & $0.0$MB & $ 50{+} 16$MB & & $ 39$GB\\ 
\hiderowcolors
abdomen\_long.n6c10 & $144.4$ & $11.8$ &  &  &  &  & $179$s & $ 11$ &  &  & $>35$ &  \\* 
\hiderowcolors
\multicolumn{3}{|l|}{512$\times$512$\times$551, \ \ $ 19$GB} & $ 29$GB &  &  &  & $410{+}403$MB & & $196$GB &  & & ${>}1$TB\\ 
\showrowcolors
\rowcolor{blue!10}
abdomen\_short.n6c10 & $144.4$ & $11.8$ &  &  &  &  & $82$s & $ 11$ &  &  &  &  \\* 
\showrowcolors
\rowcolor{blue!10}
\multicolumn{3}{|>{\columncolor[rgb]{0.9,0.9,1}}l|}{512$\times$512$\times$551, \ \ $ 19$GB} & $ 29$GB &  &  &  & $410{+}403$MB & & $138$GB &  & & \\ 
\hline 
\end{longtable} 

}
\par
Our new algorithms computed flow values for all problems matching those provided in the dataset, except for the following cases:
\begin{itemize}
\item \verb=LB07-bunny-lrg=: no ground truth solution available (we found flow/cut of cost 15537565).
\item \verb=babyfacen26c10= and \verb=babyfacen26c100=: we found higher flow values than those which were provided in the dataset (we found flow/cut of cost 180946 and 1990729 resp.).
\end{itemize}
The latter problems appear to be the most difficult for S-ARD in terms of both time and number of sweeps. Despite this, S-ARD requires much fewer sweeps, and consequently much less disk I/O operations than the push-relabel variant. This means that in the streaming mode, where read and write operations take a lot of time, S-ARD is clearly superior. Additionally, we observe that the time it spends for computation is comparable to that of BK, sometimes even significantly smaller.
\par
Next, we studied the dependency of computation time and number of sweeps on the number of regions in the partition. We selected 3 representative instances of different problems and solved them using partitions into different number of regions. The results are presented in the Fig.~\ref{fig:sweeps_vs_regions}. The instance \verb=BL06-gargoyle-sml= was partitioned by the node number and the rest two problems were partitioned in 3D according to their grid sizes using variable number of slices in each dimension. These results shows that the computation time required is stable over a large range of partitions and the number of sweeps does not grow rapidly. Therefore, the partition for S-ARD can be selected to meet other requirements: memory consumption, number of computation units, \etc. We should note however, that with refining the partition the amount of shared memory grows proportionally to the number of boundary edges. In the limit of single-vertex regions the algorithm will turn into a very inefficient implementation of pure push-relabel.
\begin{figure}[ht]
\begin{center}
\setlength{\tabcolsep}{0pt}
\begin{tabular}{c}\includegraphics[width=0.5\linewidth]{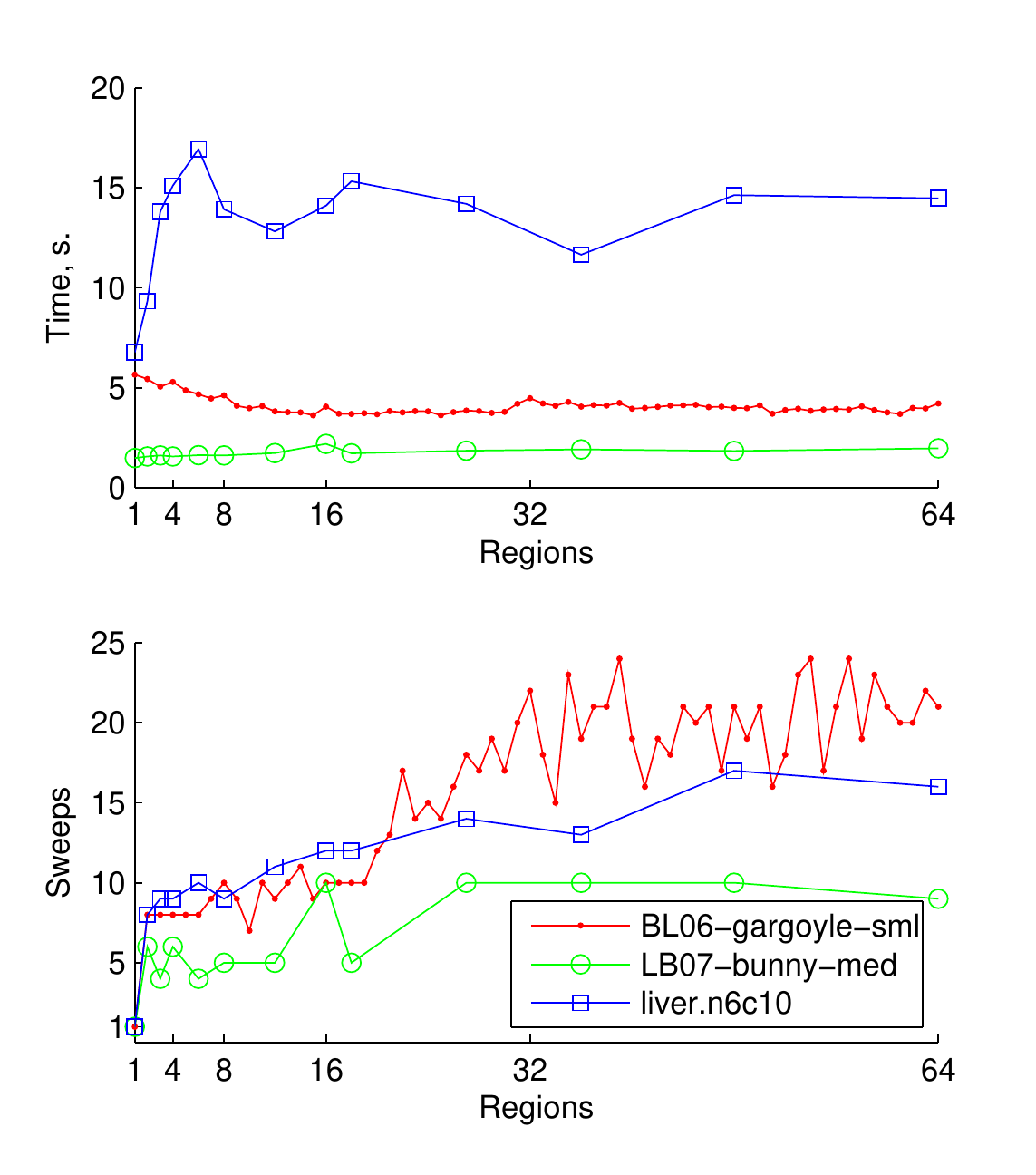}\\
\end{tabular}
\end{center}
\caption{Dependence on the number of regions for the representative instances of multiview, stereo and segmentation. Top: CPU time used. Bottom: number of sweeps.}
\label{fig:sweeps_vs_regions}
\end{figure}
\subsection{Parallel Competition}\label{sec:exp_parallel}
In this section we test parallel versions of our algorithms and compare them with two state-of-the-art methods. The experiments are conducted on the same machine as above (Intel Core 2 Quad CPU@2.66Hz) but allowing the use of all 4 CPUs. The goal is to see how the distributed algorithms perform in the simplified setting when they are run not in the network but on a single machine. For P-ARD/PRD we expect that the total required work would increase compared to the sequential versions because the discharges are executed concurrently. The relative speed-up therefore would be sublinear even if we managed to distribute the work between CPUs evenly. The tests are conducted on small and medium size problems (taking under 2GB of memory). For P-ARD and P-PRD we use the same partition into regions as in Table~\ref{table1}. For other solvers, discussed next, we tried to meet better their requirements.
\mypar{DD} The dual decomposition method\footnote{Multi-threaded maxflow library~\url{http://www.maths.lth.se/matematiklth/personal/petter/cppmaxflow.php}}~\cite{Strandmark10}. The algorithm for the integer maxflow/mincut problem presented in~\cite{Strandmark10} is a heuristic which has no guarantees. In particular, it is not guaranteed to terminate. The actual implementation uses additional randomization helping to ``guess'' the last bit of the solution. With this randomization switched off the algorithm did not terminate in 1000 iterations on a simple example of 4 nodes. The algorithm does not scale well with the number of regions in the partition, it performs better with fewer regions. We tested it with partitions into 2 regions and 4 regions (denoted {\bf DDx2} and {\bf DDx4} resp.). Naturally, with 2 regions it can utilize only 2 CPUs. To our surprise, the algorithm terminated on all of the stereo problems in a small number of iterations. However, on larger problems partitioned into 4 regions it exceeded its internal bound of iterations (1000) in many cases and returned without optima flow/cut. In such a case it provides only an approximate solution to the problem. Whether such a solution is of practical value is beyond us.
\mypar{RPR} A recently published implementation of Region Push Relabel~\cite{Delong08} by Sameh Khamis (v1.01,~\url{http://vision.csd.uwo.ca/code/}).
For RPR we constructed partition of the problem into smaller "blocks". Because regions in RPR are composed dynamically out of blocks (default is 8 blocks per region) we partitioned 2D problems into $64=8^2$ blocks and 3D problems into $512=8^3$ blocks. This partitioning was also empirically faster than a coarser one. The parameter \verb=DischargesPerBlock= was set by recommendation of authors to 500 for small problems (stereo) and to 15000 for big problems.
The implementation is specialized for regular grids, therefore multiview and KZ2 problems which do not have \verb=regulargrid= hint cannot be solved by this method. Because of the fixed graph layout in RPR, arcs which are reverse of each other are automatically grouped together, so RPR computes on a reduced graph compared to other methods. Let us also note that because of the dynamic regions, RPR is not fully suitable to run in a distributed system.
\par
The method~\cite{Liu10} (parallel, but not distributed) would probably be the fastest one in this competition (as could be estimated from the results reported in~\cite{Liu10}), however the implementation is not publicly available.
\par
%\par
{\scriptsize
\setlength{\tabcolsep}{5pt}
\setlength{\minrowclearance}{1pt}
\begin{longtable}{|l|c|lr|lr|lr|lr|lr|}
\caption[Parallel Competition]{Parallel Competition} \label{table_parallel} \\
\hline
\multicolumn{1}{|c|}{problem} & BK~\cite{BK-maxflow} & \multicolumn{2}{|c|}{DDx2~\cite{Strandmark10}} & \multicolumn{2}{|c|}{DDx4~\cite{Strandmark10}} & \multicolumn{2}{|c|}{P-ARD} & \multicolumn{2}{|c|}{P-PRD} & \multicolumn{2}{|c|}{RPR~\cite{Delong08}}\\*
\hline
\rowcolor{blue!10}
  & time & \multicolumn{10}{|>{\columncolor[rgb]{0.9,0.9,1}}c|}{time\ \ \ sweeps}\\*
\endfirsthead
\multicolumn{12}{c}%
{{\bfseries \tablename\ \thetable{} -- continued from previous page}} \\
\hline
\endhead
\hline \multicolumn{12}{|r|}{{Continued on next page}} \\ \hline
\endfoot
\hline \hline
\endlastfoot
\hline
\multicolumn{12}{|c|}{\bf stereo}\\* 
\hline 
\showrowcolors
\rowcolor{blue!10}
BVZ-sawtooth(20) & $0.68$s & $0.52$s & $  7$ & $0.37$s & $ 11$ & $0.30$s & $  7$ & $2.4$s & $ 31$ & $4.8$s & $274$\\* 
\hiderowcolors
BVZ-tsukuba(16) & $0.36$s & $0.28$s & $  6$ & $0.20$s & $  8$ & $0.17$s & $  5$ & $1.5$s & $ 33$ & $2.1$s & $197$\\* 
\showrowcolors
\rowcolor{blue!10}
BVZ-venus(22) & $1.2$s & $0.84$s & $  7$ & $0.59$s & $  9$ & $0.50$s & $  7$ & $4.9$s & $ 36$ & $8.0$s & $466$\\* 
\hiderowcolors
KZ2-sawtooth(20) & $1.8$s & $1.2$s & $ 11$ & $0.91$s & $ 16$ & $0.96$s & $  6$ & $4.9$s & $ 23$ &  & \\* 
\showrowcolors
\rowcolor{blue!10}
KZ2-tsukuba(16) & $1.1$s & $0.67$s & $  7$ & $0.52$s & $ 11$ & $0.70$s & $  8$ & $4.9$s & $ 22$ &  & \\* 
\hiderowcolors
KZ2-venus(22) & $2.8$s & $1.9$s & $  7$ & $1.3$s & $ 12$ & $1.5$s & $ 10$ & $10$s & $ 39$ &  & \\* 
\hline
\multicolumn{12}{|c|}{\bf multiview}\\* 
\hline 
\showrowcolors
\rowcolor{blue!10}
BL06-camel-med & $25$s & $18$s & $221$ & $13$s & $260$ & $8.7$s & $ 14$ & $81$s & $322$ &  & \\* 
\hiderowcolors
BL06-camel-sml & $0.98$s & $0.63$s & $ 11$ & $0.49$s & $ 27$ & $0.49$s & $ 10$ & $2.5$s & $ 70$ &  & \\* 
\showrowcolors
\rowcolor{blue!10}
BL06-gargoyle-lrg & $245$s & $120$s & $517$ & mem &  & $58$s & $ 23$ & mem &  &  & \\* 
\hiderowcolors
BL06-gargoyle-med & $115$s & $59$s & $ 20$ & $38$s & $ 50$ & $27$s & $ 21$ & $79$s & $219$ &  & \\* 
\showrowcolors
\rowcolor{blue!10}
BL06-gargoyle-sml & $6.1$s & $3.0$s & $ 19$ & $1.9$s & $ 19$ & $1.6$s & $ 10$ & $2.4$s & $ 52$ &  & \\* 
\hline
\multicolumn{12}{|c|}{\bf surface}\\* 
\hline 
\hiderowcolors
LB07-bunny-med & $1.6$s & $1.3$s & $ 11$ & $1.1$s & $ 11$ & $1.3$s & $ 13$ & $12$s & $ 35$ & $37$s & $349$\\* 
\showrowcolors
\rowcolor{blue!10}
LB07-bunny-sml & $0.17$s & $0.12$s & $ 11$ & $0.12$s & $ 11$ & $0.21$s & $  8$ & $0.58$s & $ 21$ & $3.5$s & $ 99$\\* 
\hline
\multicolumn{12}{|c|}{\bf segm}\\* 
\hline 
\hiderowcolors
liver.n6c10 & $7.2$s & {\red\bfseries\sffamily X }$7.6$s & $1000$ & {\red\bfseries\sffamily X }$22$s & $1000$ & $8.9$s & $ 23$ & $23$s & $164$ & $5.1$s & $1298$\\* 
\showrowcolors
\rowcolor{blue!10}
liver.n6c100 & $15$s & $17$s & $ 31$ & {\red\bfseries\sffamily X }$21$s & $1000$ & $12$s & $ 17$ & $23$s & $102$ & $7.3$s & $1722$\\* 
\hiderowcolors
babyface.n6c10 & $6.8$s & $8.8$s & $ 61$ & {\red\bfseries\sffamily X }$24$s & $1000$ & $12$s & $ 22$ & $61$s & $135$ & $17$s & $4399$\\* 
\showrowcolors
\rowcolor{blue!10}
babyface.n6c100 & $13$s & $16$s & $338$ & {\red\bfseries\sffamily X }$20$s & $1000$ & $17$s & $ 23$ & $61$s & $179$ & $22$s & $4833$\\* 
\hiderowcolors
bone.n6c10 & $7.7$s & $5.2$s & $ 22$ & {\red\bfseries\sffamily X }$8.2$s & $1000$ & $4.9$s & $ 17$ & $16$s & $182$ & $6.3$s & $918$\\* 
\showrowcolors
\rowcolor{blue!10}
bone.n6c100 & $9.1$s & $5.3$s & $ 12$ & $4.1$s & $ 17$ & $6.2$s & $ 13$ & $14$s & $ 70$ & $7.9$s & $1070$\\* 
\hiderowcolors
bone\_subx.n6c100 & $7.1$s & $6.3$s & $ 24$ & $5.2$s & $ 34$ & $3.9$s & $ 17$ & $5.8$s & $ 48$ & $1.5$s & $747$\\* 
\showrowcolors
\rowcolor{blue!10}
bone\_subxy.n26c100 & $5.9$s & $3.4$s & $ 11$ & $3.2$s & $ 12$ & $5.8$s & $ 16$ & $6.0$s & $ 37$ & hang & \\* 
\hline 
\end{longtable} 

}
\par The results are summarized in table~\ref{table_parallel}. The time reported is the wall clock time passed in the calculation phase, not including any time for graph construction. The number of sweeps for DD has the same meaning as for P-ARD/PRD, it is the number of times all regions are synchronously processed. RPR however is asynchronous and uses dynamic regions. For it we define sweeps = \verb=block_discharges=/\verb=number_of_blocks=.%the number of sweeps for RPR is %the cells marked ``mem'' indicate that the solver could not allocate enough memory, cells 
\par
Comparing to Table~\ref{table1}, we see that P-ARD on 4 CPUs is about $1.5-2.5$ times faster than S-ARD. The speed-up over BK varies from 0.8 on \verb=livern6c10= to more than 4 on \verb=gargoyle=. 
\par
We see that DD gets lucky some times and solves the problem really quickly, but often it fails to terminate. We also observe that our variant of P-PRD (based on highest first selections rule) is a relatively slow, but robust distributed method. RPR, which is based on FIFO selection rule, is competitive on the 3D segmentation problems but is slow on other problems, despite its compile-time optimization for the particular graph structure. It is also uses relatively higher number of blocks, The version we tested always returned the correct flow value but often a wrong (non-optimal) cut. Additionally, for 26 connected \verb=bone_subxy.n26c100= it failed to terminated within 1 hour.

\section{Region Reduction}\label{sec:preprocessing}
Some vertices become disconnected from the sink in the course of the studied	algorithms. If they are still reachable from the source, they must belong to the source set of any optimal cut. Such vertices do not participate in further computations and the problem can be reduced by excluding them. Unfortunately, the opposite case, when a vertex must be strictly in the sink set is not discovered until the very end of the algorithm.
\par
The following algorithm attempts to identify as many nodes as possible for a given region. It is based on the following simple consideration: if a node is disconnected from the sink in $G^R$ as well as from the region boundary, $B^R$, then it is disconnected from the sink in $G$; if a node is not reachable from the source in $G^R$ as well as from $B^R$ then it is not reachable from the source in $G$.
\par 
Let us say that a node $v$ is a {\em strong source node} (resp. a {\em strong sink node}) if for any optimal cut $(C, \bar C)$ $v\in C$ (resp. $v\in \bar C$). Similarly, $v$ will be called a {\em weak source node} (resp. {\em weak sink node}), if there exist an optimal cut $(C, \bar C)$ such that $v\in C$ (resp. $v\in \bar C$).
\par
Kovtun~\cite{KovtunPhD} suggested to solve two auxiliary problems, modifying network $G^R$ by adding infinite capacity links from the boundary nodes to the sink and in the second problem adding infinite capacity links from the source to the boundary nodes. In the first case, if $v$ is a strong source node in the modified network $G^R$, it is also a strong source node in $G$. Similarly, the second auxiliary problem allows to identify strong sink nodes in $G$. It requires solving a maxflow problem on $G^R$ twice. We improve this construction by reformulating it as the following algorithm finding a single flow in $G^R$.
\par
\begin{algorithm}
\tcc{Input: network $G^R$, boundary $B^R$}
$\Augment(s,t)$\;\label{AP augment st}
$B^S := \{v \mid v\in B^R, s\rightarrow v\}$\tcc*{source boundary set}\label{AP boundary split}
$B^T := \{v \mid v\in B^R, v\rightarrow t\}$\tcc*{sink boundary set}
$\Augment(s,B^S)$\;\label{AP augment S}
$\Augment(B^T,t)$\;\label{AP augment T}
\ForEach{$v\in R$}{
	\lIf{$s\rightarrow v$}{$v$ is strong source node}\;
	\lIf{$v\rightarrow t$}{$v$ is strong sink node}\;
	\Other{
		\lIf{$v\nrightarrow B^R$}{$v$ is weak source node}\;% and remove from $G$
		\lIf{$B\nrightarrow v^R$}{$v$ is weak sink node}
	}
}
% $v\rightarrow t$ assign $v$ to strong $T$
%\item vertices $v$ such that $s\nrightarrow v$ and $v\nrightarrow t$ assign to $S$.
%\item group vertices with $s\nrightarrow v$ and $v\nrightarrow t$ into connected components, assign each component to either $S$ or $T$ and remove it from $G$.
\caption{Region Reduction $(G^R,B^R)$}\label{alg:preprocessing}
\end{algorithm}

%\begin{myalgorithm}{Region Reduction}
%\label{alg:preprocessing}
%\item {\bf Input:} Network $G^R$, boundary $B$
%\item\label{AP augment st}{\bf Augment}$(s,t)$
%\item\label{AP boundary split} Let source boundary set $B^S = \{v| v\in B, s\rightarrow v\}$
%\item[] Let sink boundary set $B^T = \{v| v\in B, v\rightarrow t\}$
%\item\label{AP augment S} {\bf Augment}$(s,B^S)$.
%\item\label{AP augment T} {\bf Augment}$(B^T,t)$.
%\item {\bf if} $s\rightarrow v$ then $v$ is strong source node% and remove from $G$
%\item {\bf if} $v\rightarrow t$ then $v$ is strong sink node% and remove from $G$
%\item {\bf otherwise}
%\item\hspace{1cm} {\bf if} $v\nrightarrow B$ then $v$ is weak source node% and remove from $G$
%\item\hspace{1cm} {\bf if} $B\nrightarrow v$ then $v$ is weak sink node
%% $v\rightarrow t$ assign $v$ to strong $T$
%%\item vertices $v$ such that $s\nrightarrow v$ and $v\nrightarrow t$ assign to $S$.
%%\item group vertices with $s\nrightarrow v$ and $v\nrightarrow t$ into connected components, assign each component to either $S$ or $T$ and remove it from $G$.
%\end{myalgorithm}
%\vskip-2mm
%\par
\begin{statement}
Sets $B^S$ and $B^T$ constructed in step~\ref{AP boundary split} are disjoint.
\end{statement}
\begin{proof}
We have $s\nrightarrow t$ after step~\ref{AP augment st}, hence there cannot exist simultaneously a path from~$s$ to $v$ and a path from~$v$ to $t$.
\qed
\end{proof}
After step~\ref{AP augment st} the network $G^R$ is split into two disconnected networks: with nodes reachable from $s$ and nodes from which $t$ is reachable. Therefore any augmentations occurring in steps \ref{AP augment S} and \ref{AP augment T} act on their respective subnetworks and can be carried independently of each other. On the output of \Algorithm{alg:preprocessing} we have: $s\nrightarrow B^R\cup\{t\}$ and $B^R\cup\{s\}\nrightarrow t$. The classification of nodes is shown in \Figure{fig:preproc}.
\begin{figure}[!ht]
\centering
\begin{tabular}{ccc}
\begin{tabular}{c}
\includegraphics[width=0.25\linewidth]{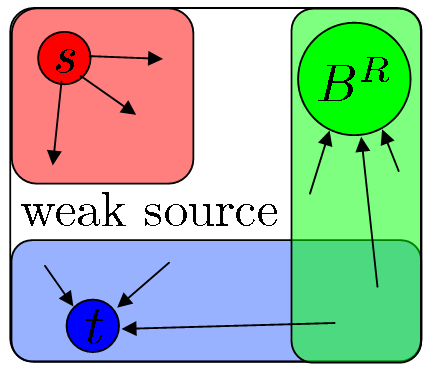}
\end{tabular} & \parbox{0.1\linewidth}{\ } &
\begin{tabular}{c}
\includegraphics[width=0.25\linewidth]{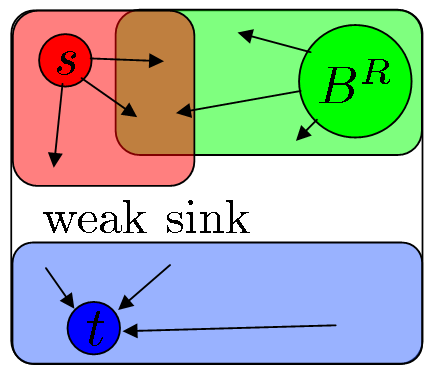}
\end{tabular}\\
(a) & & (b)
\end{tabular}
\caption{Classification of nodes in $V^R$ build by \Algorithm{alg:preprocessing}. Nodes reachable from $s$ are strong source nodes. Nodes from which $t$ is reachable are strong sink nodes. The remaining nodes can be classified as weak source (a) if they cannot reach boundary, or as weak sink (b) if they are not reachable from the boundary. Some nodes are both: weak source and weak sink, this means they can be on both sides of an optimal cut (but not independently).}
\label{fig:preproc}
\end{figure}
\par
Augmenting on $(s,t)$ in step~\ref{AP augment st} and on $(s,B^S)$ in the step~\ref{AP augment S} is the same work as done in ARD (where $(s,B^S)$ paths are augmented in the order of labels of $B^S$). This is not a coincidence, these algorithms are very much related. However, the augmentation on $(B^T,t)$ in step~\ref{AP augment T} cannot be executed during ARD. It would destroy validity of the labeling. It may only be executed during the first sweep of S-ARD, as an initialization. Otherwise, we may consider Algorithm~\ref{alg:preprocessing} as a general preprocessing.
%Hence, we only consider this algorithm as a preprocessing, which may be seamlessly incorporated in the first sweep of S-ARD.
\par
If $v$ is a weak source node, it follows that it is not a strong sink node. In the preflow pushing algorithms we find the cut $(\bar T,T)$, where $T$ is the set of all strong sink nodes in $G$. We consider that v is {\em decided} if it is a strong sink or a weak source node.
\begin{table}[!ht]
\scriptsize
\begin{center}
\setlength{\tabcolsep}{2pt}
\begin{tabular}{|lc|lc|lc|lc|}
\hline
BVZ-sawtooth(20) & 80.0\% & LB07-bunny-sml & 15.6\% & bone.n26c100 & 6.9\% & bone\_subxyz.n6c100 & 6.6\%\\
BVZ-tsukuba(16) & 72.8\% & liver.n26c10 & 7.1\% & bone.n6c10 & 8.8\% & bone\_subxyz\_subx.n26c10 & 7.9\%\\
BVZ-venus(22) & 70.2\% & liver.n26c100 & 5.3\% & bone.n6c100 & 7.0\% & bone\_subxyz\_subx.n26c100 & 6.6\%\\
KZ2-sawtooth(20) & 85.0\% & liver.n6c10 & 7.2\% & bone\_subx.n26c10 & 6.6\% & bone\_subxyz\_subx.n6c10 & 8.2\%\\
KZ2-tsukuba(16) & 69.9\% & liver.n6c100 & 5.3\% & bone\_subx.n26c100 & 6.6\% & bone\_subxyz\_subx.n6c100 & 6.6\%\\
KZ2-venus(22) & 75.8\% & babyface.n26c10 & 29.3\% & bone\_subx.n6c10 & 6.3\% & bone\_subxyz\_subxy.n26c10 & 11.3\%\\
BL06-camel-lrg & 2.0\% & babyface.n26c100 & 30.9\% & bone\_subx.n6c100 & 6.3\% & bone\_subxyz\_subxy.n26c100 & 9.5\%\\
BL06-camel-med & 2.3\% & babyface.n6c10 & 35.4\% & bone\_subxy.n26c10 & 6.6\% & bone\_subxyz\_subxy.n6c10 & 12.7\%\\
BL06-camel-sml & 4.6\% & babyface.n6c100 & 33.7\% & bone\_subxy.n26c100 & 6.6\% & bone\_subxyz\_subxy.n6c100 & 9.3\%\\
BL06-gargoyle-lrg & 6.0\% & adhead.n26c10 & 0.3\% & bone\_subxy.n6c10 & 6.4\% & abdomen\_long.n6c10 & 1.7\%\\
BL06-gargoyle-med & 2.4\% & adhead.n26c100 & 0.3\% & bone\_subxy.n6c100 & 6.3\% & abdomen\_short.n6c10 & 6.3\%\\
BL06-gargoyle-sml & 9.8\% & adhead.n6c10 & 0.2\% & bone\_subxyz.n26c10 & 6.6\% &  & \\
LB07-bunny-lrg & 11.4\% & adhead.n6c100 & 0.1\% & bone\_subxyz.n26c100 & 6.6\% &  & \\
LB07-bunny-med & 13.1\% & bone.n26c10 & 8.7\% & bone\_subxyz.n6c10 & 6.6\% &  & \\
\hline
\end{tabular}

\end{center}
\caption{Percentage of nodes which can be decided by preprocessing. The problems are partitioned into regions the same way as in Table~\ref{table1}. For stereo problems the average number over subproblems is shown.}\label{table2}
\end{table}
\par
%Experimental results of excl are given 
Table~\ref{table2} gives the percentage of how many vertices are decided (and hence can be excluded from the problem) by Algorithm~\ref{alg:preprocessing} for computer vision problems. It is seen that in stereo problems, a large percent of vertices is decided. These problems are rather local and potentially can be fully solved by using Algorithm~\ref{alg:preprocessing} several times in overlapping windows. In contrast, only a small fraction can be decided locally for the other problems.
%

%\section{Complementary Results}
%\subsection{Dual Decomposition View}
%\section{Parallel PR}
\section{Appendix A: Tightness of $O(n^2)$ bound for PRD.}
Here we give an example of a network, its partition into regions and a sequence of valid push and relabel operations, implementing PRD, such that sequential region discharge requires $O(n^2)$ sweeps. Because there is only one active node at any time, it also applies to parallel PRD.
\par
We start by an auxiliary example, in which the preflow is transfered from a vertex to a boundary vertex with a higher label. In this example some inner vertices of a region are relabeled, but not the boundary vertices. Therefore the total number of sweeps cannot be bounded by the number of relabellings of the boundary vertices only.%When this happens, boundary verticeis are 
\begin{example}
Consider a network of 6 regular nodes in Figure~\ref{fig:example1}. Assume all edges have infinite capacity, so only non-saturating pushes occur. There are two regions $R_1 = \{1,2,3,4,5\}$ and $R_2 = \{6\}$. Figure~\ref{fig:example1} shows a sequence of valid push and relabel operations. We see that some nodes were raised due to relabel, but the net effect is that flow excess from node 1 was transfered to node 6, which has a higher label. Moreover none of the boundary nodes (nodes 1,5,6) were relabeled.
\end{example}
\begin{figure}[!ht]
\begin{center}
\begin{tabular}{c}
\includegraphics[width=\linewidth]{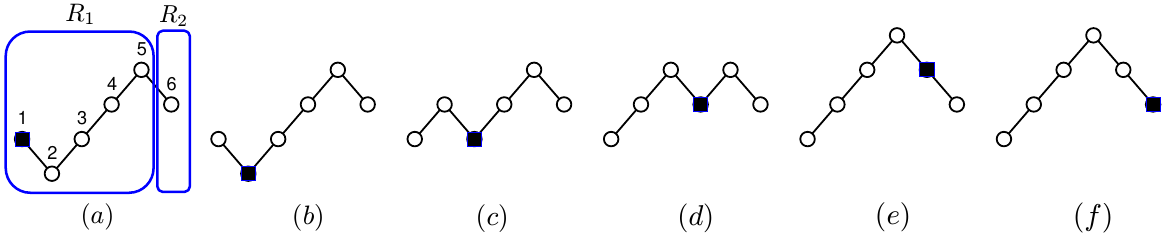}\\
\end{tabular}\\
%\noindent
%\hspace{-1cm}
%\hrulefill (a) \hrulefill (b) \hrulefill (c) \hrulefill (d) \hrulefill (e) \hrulefill (f) \hrulefill\\
\end{center}
\caption{Steps of Example 1. Node's height correspond to its label. Black box shows the node with excess. The source node and think node are not shown. (a) (b) flow excess is pushed to node 2; (c) node 2 is relabeled, so that two pushes are available and excess is pushed to node 3; (d-f) similar.
}
\label{fig:example1}
\end{figure}
\begin{figure}[ht]
\begin{center}
\includegraphics[width=\linewidth]{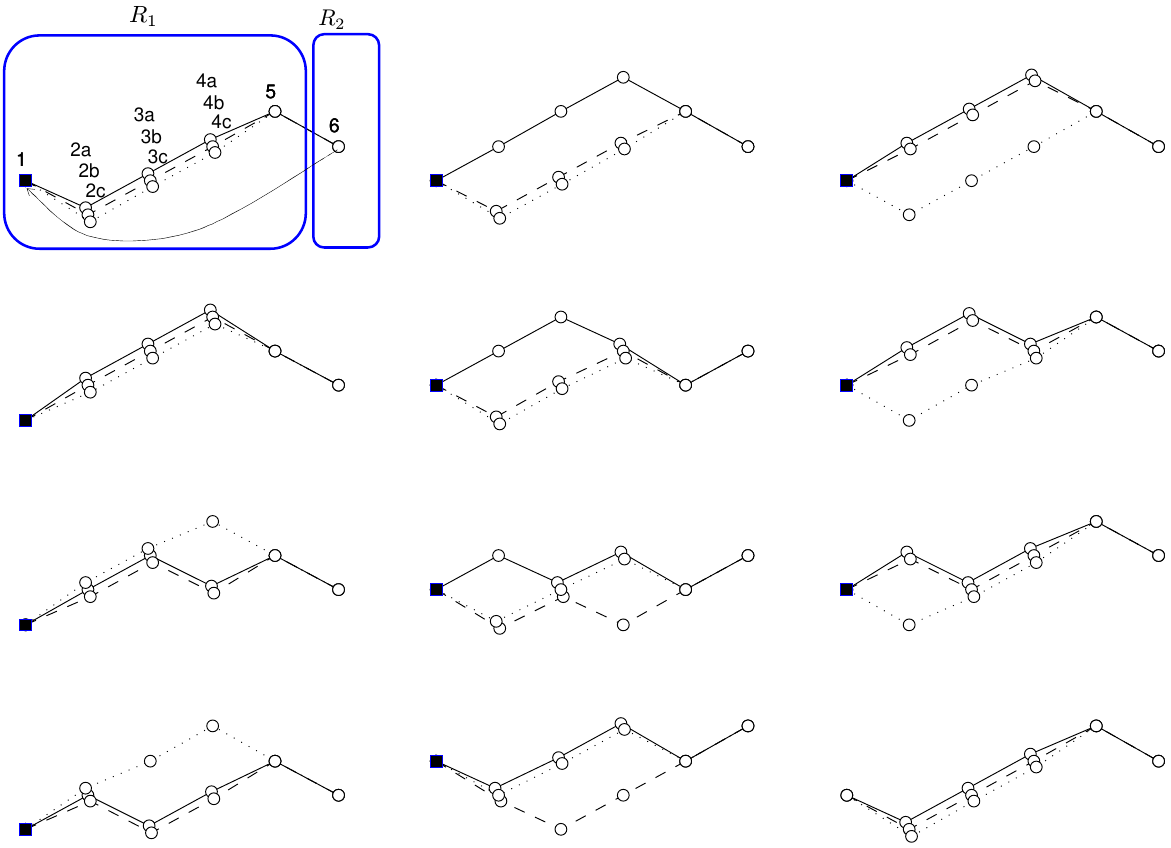}
\end{center}
\caption{Steps of Example 2. Top left: a network with several chains of nodes like in Example 1. Nodes 1,5,6 are common for all chains but there are separate copies of nodes 2,3,4 denoted by letters. In addition, there is a reverse arc from node 6 to node 1. From left to right, top to bottom: one step of transferring a flow from node 1 to node 6 using one of the chains and then pushing it through the arc (6,1), relabeling 6 when necessary.}
\label{fig:example2}
\end{figure}
\begin{example}% We now considers several chains of nodes like in example 1, which have nodes 1,5,6 in common and separate nodes 2,3,4. In addition, there is a reverse arc from node 6 to node 1. Fom left to right, top to bottom: one step of transferring a flow from node 1 to node 6 using one of the chains and then pushing it on the arc (6,1), relabeling 6 when necessary. 
\par
Consider the network in Figure~\ref{fig:example2}. The first step corresponds to a sequence of push and relabel operations (same as in Figure~\ref{fig:example1}) applied to the chain $(1, 2a, 3a, 4a, 5, 6)$. Each next step starts with the excess at node 1. Chains are selected in turn in the order $a,b,c$. It can be verified from the figure that each step is a valid possible outcome of PRD applied first to $R_1$ and then to $R_2$. The last configuration repeats the first one with all labels raised by 2, so exactly the same loop may be repeated many times.
\par
It is seen that nodes $1,5,6$ are relabeled only during pushes on paths $a$ and $b$ and never during pushes on path $c$. If there were more paths like path $c$, it would take many iterations (= number of region discharge operations) before boundary vertices are risen.
Let there be $k$ paths in the graph ($d$, $e$\dots), handled exactly the same way as path $c$. 
The number of nodes in the graph is $n=O(k)$. It will take $O(n)$ region discharge operations to perform each loop, raising all nodes by 2. Therefore until node 1 reaches the maximal label it will take $O(n^2)$ steps.
\par
Because there is only one active node at any time, this example is independent of the rule used to select the active node (highest label selection rule or FIFO rule). Also, noting that the number of boundary nodes is 3, we see that our S-ARD algorithm will terminate in a constant number of sweeps for arbitrary $k$.
%This means if we add more paths to the graph (d,e\dots), they can be processed exactly like the path $c$.
%\par
%During the loop nodes $1, 5, 6$ are raised by 2.
\end{example}

\section{Appendix B: Relation to Dual Decomposition}\label{sec:DD}
%In this section we provide a simple observation, how the dual variables in dual decomposition for {\sc mincut}~\cite{Strandmark10} correspond to the flow in an 
In our approach we partition the set of vertices into regions and couple the regions by sending the flow through the inter-region edges. In the dual decomposition for {\sc mincut}~\cite{Strandmark10} detailed below, a separator set of the graph is selected, each subproblem gets a copy of the separator set and the coupling is achieved via the constraint that the cut of the separator set must be consistent across the copies. We now show how the dual variables in~\cite{Strandmark10} can be interpreted as flow, thus relating approach~\cite{Strandmark10} to ours.
%
%When the optimal flow on all inter-region edges is known the problem is essentially solved as it decouples into solving the independent mincut/maxflow problems on the regions. In the dual decomposition approach the same holds when the optimal dual variables are known. We will show now how the dual variables in~\cite{Strandmark10} can be interpreted as flows, thus relating~\cite{Strandmark10} to our approach.
%
\begin{figure}[ht]
\begin{center}
\setlength{\tabcolsep}{0pt}
\begin{tabular}{c}\includegraphics[width=0.8\linewidth]{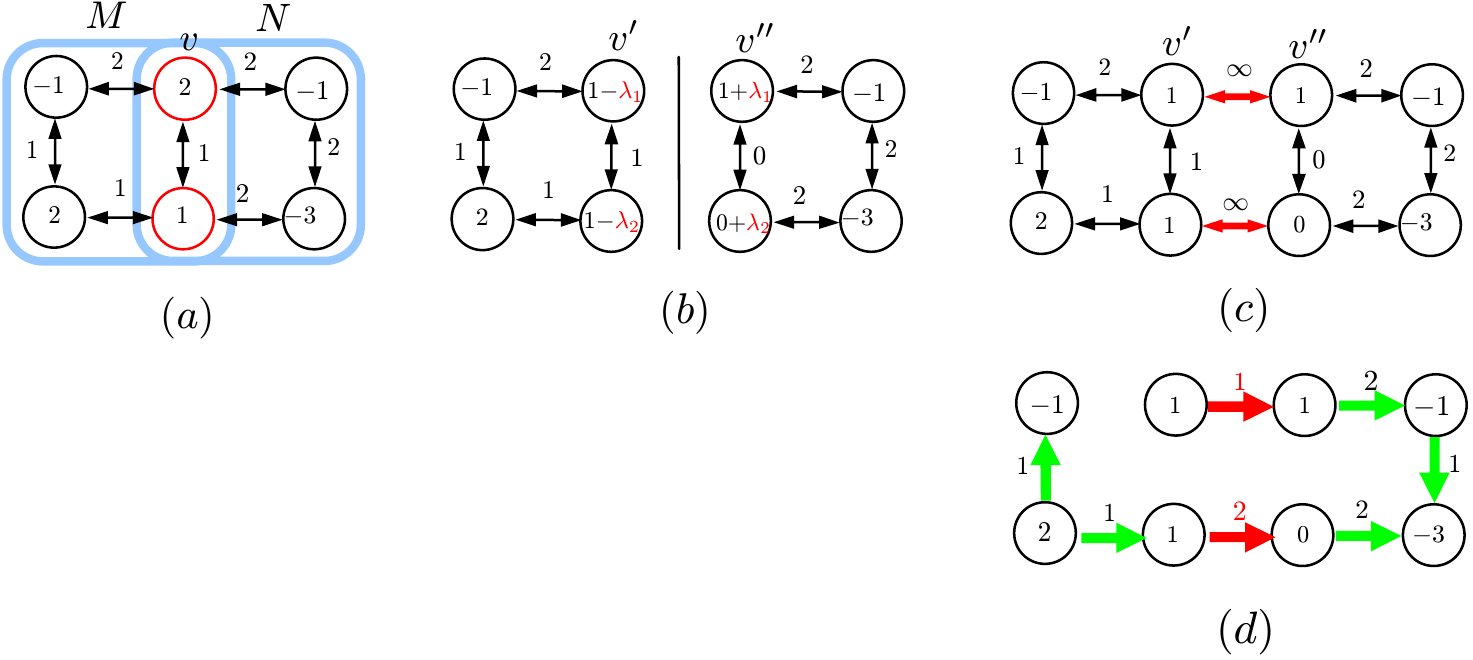}\\
\end{tabular}
\end{center}
\caption{Interpretation of dual decomposition. (a) Example of a network with denoted capacities. Terminal capacities are shown in circles, where ``$+$'' denotes $s$-link and ``$-$'' denotes $t$-link. $M\cap N$ is a separator set. (b) The network is decomposed into two networks holding copies of the separator set. The associated capacities are divided (not necessarily evenly) between two copies. The variable $\lambda_1$ is the Lagrangian multiplier of the constraint $x_v=y_v$. (c) Introducing edges of infinite capacity enforces the same constraint, that $v'$ and $v''$ are necessarily in the same cut set of any optimal cut. (d) A maximum flow in the network (c), the flow value on the red edges corresponds to the optimal value of dual variables $\lambda$.}
\label{fig:dd_graph}
\end{figure}
\par
Decomposition of the {\sc mincut} problem into two parts is formulated in~\cite{Strandmark10} as follows. Let $M,N\subset V$ are such that $M\cup N = V$, $\{s,t\}\subset M\cap N$ and there are no edges in $E$ from $M\backslash N$ to $N\backslash M$ and vice-versa. Let $x\colon M\to \{0,1\}$ and $y\colon N\to \{0,1\}$ be the indicator variables of the cut set, where $0$ corresponds to the source set. Then {\sc mincut} problem can be reformulated as:
\begin{equation}
\begin{aligned}
& \min_{x,y} C^M(x)+C^N(y),\\
& \st \begin{cases}
x_s = x_s = 0,\\
x_t = y_t = 1,\\
x_i = y_i,\tab \forall i\in M\cap N,
\end{cases}
\end{aligned}
\end{equation}
where 
\begin{subequations}
\begin{align}
& C^M(x) = \sum_{(i,j)\in E^M} c^M(i,j)(1-x_i)x_i,\\
& C^N(y) = \sum_{(i,j)\in E^N} c^N(i,j)(1-y_i)y_i;
\end{align}
\end{subequations}
\begin{subequations}
\begin{align}
& c^M(i,j)+c^N(i,j) = c(i,j),\\
& c^M(i,j) = 0 \tab \forall i,j\in N\backslash M,\\
& c^N(i,j) = 0 \tab \forall i,j\in M\backslash N,
\end{align}
\end{subequations}
$E^M = (M,M) = (M{\times}M)\cap E$ and $E^N = (N,N)$.
The minimization over $x$ and $y$ decouples once the constraint $x_i=y_i$ is absent. The dual decomposition approach is to solve the dual problem:
\begin{equation}
\max\limits_{\lambda} \Big[ \min\limits_{\substack{x\\ x_s=0\\ x_t=1}}\Big(C_M(x) +\sum_{i\in M\cup N}\lambda_i (1-x_s)x_i \Big) +  \min\limits_{\substack{x\\ x_s=0\\ x_t=1}}\Big(C_N(y) -\sum_{i\in M\cup N}\lambda_i (1-y_s)y_i \Big) \Big].
\end{equation}
\par
We observe that dual variables $\lambda$ correspond to flow on the artificial edges of infinite capacity between variable $x_i$ and $y_i$ like it is explained in Fig.~\ref{fig:dd_graph}. For a problem with integer capacities there exist an integer optimal flow. This observation provides an alternative proof of~\cite[Theorem 2]{Strandmark10}, stating that there exist an integer optimal $\lambda$.
\par
The algorithm we introduced could be applied to such a decomposition by running it on the extended graph (\cf Fig.~\ref{fig:dd_graph}(c)), where vertices of the separator set are duplicated and linked by additional edges of infinite capacity. It could be observed, however, that this construction does not allow to reduce the number of boundary vertices or the number of inter-region edges, while the size of the regions increases. Therefore it is not beneficial with our approach.
\par
\section{Conclusion}
We developed a new algorithm for {\sc mincut} problem on sparse graphs, which combines augmenting paths and push-relabel approaches. The main result of this work is the worst case complexity guarantee of $O(|\B|^2)$ sweeps for the sequential and parallel variants of the algorithm (S/P-ARD). We also gave a novel parallel version of the region push-relabel algorithm of~\cite{Delong08}, an improved algorithm for local problem reduction (Sect.~\ref{sec:preprocessing}) and a number of auxiliary results.
\par
Both is theory and practice (randomized tests) S-ARD has a better asymptote in the number of sweeps than the push-relabel variant. Experiments on real instances showed that when run on a single CPU and the whole problem fits into memory, S-ARD is comparable in speed with the non-distributed BK implementation, being even significantly faster in some cases. When only one region is loaded into memory at a time, S-ARD used much fewer disk I/O than S-PRD. We also demonstrated that the running time and the number of sweeps are very stable with respect to  the partition of the problem into up to $64$ regions. In the parallel mode, using 4 CPUs, P-ARD achieves a relative speedup of about $1.5-2.5$ times over S-ARD and uses just slightly larger number of sweeps. P-ARD compares favorably to other parallel algorithms, being a robust method suitable for a use in a distributed system.
\par
Our algorithms are implemented for generic graphs. Clearly, it is possible to specialize the implementation for grid graphs, which would reduce the memory consumption and might reduce the computation time as well. 
%that S-ARD uses less disk operations and CPU time than S-PRD and runs in about . This suggests it would not be very competitive in the parallel mode, where accessing different regions is not a bottleneck. Unless of course it is improved.
\par
%We proposed parallel versions of the algorithms but have not implemented and tested them yet. 
A practically useful mode could be actually a combination of parallel and sequential, when several regions are loaded into the memory at once and processed in parallel. There are several particularly interesting combinations of algorithm parallelization and hardware, which may be exploited: 1) parallel on several CPUs, 2) parallel on several network computers, 3) sequential, using Solid State Drive, 4) sequential, using GPU for solving region discharge.
%\end{compactenum}
%Most probably, S-ARD may be beneficial in the modes 2 and 3. 
%We do not know how the ARD and PRD variants would compare in these cases.
\par
%We cannot select regions dynamically, nevertheless, 
There is the following simple way how to allow region overlaps in our framework. Consider a sequential algorithm, which is allowed to keep 2 regions in memory at a time. It can then load pairs of regions $(1,2)$, $(2,3)$, $(3,4)$\dots, and alternate between the regions is a pair until both are discharged. With PRD this is efficiently equivalent to discharging twice larger regions with a ${1/2}$ overlap and may significantly decrease the number of sweeps required. In the case of a 3D grid, it would take $8$ times more regions to allow overlaps in all dimensions. However, to meet the same memory limit, the regions have to be $8$ times smaller. It has to be verified experimentally whether it is beneficial. The RPR implementation of~\cite{Delong08} uses exactly this strategy, a dynamic region is composed out of a number of smaller blocks and blocks are discharged until the whole region is not discharged. It is likely that with this approach we could further reduce the disk I/O in the case of the streaming solver.

%\par
%\par
%As we mentioned earlier, ARD admits a more efficient implementation.
%There are a few combinations of hardware and possible parallel/distributed implementation, which may be of particular interest: 1) several CPUs, 2) GPU solving region subproblems,  
%We are yet to test parallel algorithms, utilizing several CPUs of one computer 
%\begin{itemize}
%\item Parallel competition (P-ARD, P-PRD). Multithread feature.
%\item Network distributed competition (N-ARD, N-PRD).
%\item Overlaps
%\item Preprocessing
%\item Problem reduction: how many nodes are guaranteed optimal on each sweep. Preprocessing feature.
%\end{itemize}

%\begin{figure}[tr]
%\begin{center}
%\setlength{\tabcolsep}{0pt}
%\setlength{\doublerulesep}{0pt}
%\includegraphics[width=\linewidth]{fig/example1}\\
%\caption{Example1}
%\label{fig:example1}
%\end{center}
%\vspace{-5mm}
%\end{figure}

%

\nocite{GT88,
Delong08,
Delong09-segm,
Juan-06-active-cuts,
Strandmark10,
KovtunPhD,
Liu10,
Liu09,%paint selection
Anderson95,
Goldberg91,
Goldberg-PhD,
BK-maxflow,
Cherkassky-94,
BVZ,%stereo
KZ2,%stereo
LB06,%multiview 3D
BL06,%multiview 3D
LB07,%surface fitting
BJ01,%3D segmentation
BF06,%3D segmentation
BK03,%3D segmentation
Labutit09,%3D reconstruction, tetrahonolized volume
Jancosek11}
\clearpage

{\footnotesize
\bibliographystyle{abbrv}
\bibliography{../bib/references}
}
\end{document}